\documentclass[11pt,times,letter]{article}

\usepackage{amsmath, amsfonts,amsthm}
\usepackage{dsfont}
\usepackage{bbm}

\usepackage[english]{babel}

\usepackage[normalem]{ulem}

\usepackage{latexsym,amssymb,graphicx}
\usepackage{verbatim}
\usepackage{mathrsfs}
\usepackage{epsfig}

\usepackage{xcolor}
\usepackage{tikz}
\usepackage{pgfplots}
\usepgfplotslibrary{fillbetween}

\usepackage{url}
\usepackage{bm}
\usepackage{natbib}

\usepackage[colorlinks,linkcolor=blue,citecolor=blue,anchorcolor=blue]{hyperref}

\newcommand{\N}{ \mathbb{N} }

\def\esssup_#1{\underset{#1}{\mathrm{ess\,sup\, }}}
\def\essinf_#1{\underset{#1}{\mathrm{ess\,inf\, }}}
\def\argmax_#1{\underset{#1}{\mathrm{arg\,max\, }}}
\def\argmin_#1{\underset{#1}{\mathrm{arg\,min\, }}}

\newcommand{\cL}{\mathcal{L}}

\def\E{\mathbb{E}}
\def\EE{\mathcal{E}}
\def\P{\mathbb{P}}
\def\TT{\mathcal{T}}
\def\R{\mathbb{R}}
\def\B{\mathcal{B}}
\def\N{\mathbb{N}}
\def\d{\mathrm{d}}

\newcommand{\F}{\mathcal{F}}
\newcommand{\FF}{\mathbb{F}}

\newtheorem{theorem}{Theorem}[section]
\newtheorem{definition}{Definition}[section]
\numberwithin{equation}{section}

\newtheorem{proposition}[theorem]{Proposition}

\newtheorem{remark}[theorem]{Remark}
\newtheorem{lemma}[theorem]{Lemma}

\allowdisplaybreaks[4]

\definecolor{Red}{rgb}{1.00, 0.00, 0.00}

\definecolor{DRed}{rgb}{0.5, 0.00, 0.00}

\definecolor{Blue}{rgb}{0.00, 0.00, 1.00}

\definecolor{Green}{rgb}{0.0, 0.4, 0.0}

\definecolor{mycolor1}{rgb}{0.00000,0.44700,0.74100}%
\definecolor{mycolor2}{rgb}{0.85000,0.32500,0.09800}%
\definecolor{mycolor3}{rgb}{0.92900,0.69400,0.12500}%
\definecolor{mycolor4}{rgb}{0.49400,0.18400,0.55600}%
\definecolor{mycolor5}{rgb}{0.46600,0.67400,0.18800}%
\definecolor{mycolor6}{rgb}{0.00000,1.00000,1.00000}%

\pgfplotsset{compat=newest}

\title{On time-consistent equilibrium stopping under aggregation of diverse discount rates}

\author{Shuoqing Deng \thanks{Email: masdeng@ust.hk, Department of Mathematics, The Hong Kong University of Science and Technology, Clear Water Bay, Kowloon, Hong Kong.}
\and
Xiang Yu\thanks{Email: xiang.yu@polyu.edu.hk, Department of Applied Mathematics, The Hong Kong Polytechnic University, Hung Hom, Kowloon, Hong Kong.}
\and
Jiacheng Zhang \thanks{Email: jiachengzhang@cuhk.edu.hk, Department of Statistics, the Chinese University of Hong Kong,  Shatin, NT, Hong Kong.}
}
\date{\vspace{-1cm}}

\begin{document}
\maketitle

\begin{abstract}
This paper studies a central planner's decision making on behalf of a group of members with diverse discount rates. In the context of optimal stopping, we work with an aggregation preference to incorporate all discount rates via an attitude function that reflects the aggregation rule chosen by the central planner. The problem formulation is also applicable to single agent's stopping problem with uncertain discount rate, where our aggregation preference coincides with the conventional smooth ambiguity preference. The resulting optimal stopping problem is time inconsistent, for which we develop an iterative approach using consistent planning and characterize all time-consistent mild equilibria as fixed points of an operator in the setting of one-dimensional diffusion processes. We provide some sufficient conditions on the underlying models and the attitude function such that the smallest mild equilibrium attains the optimal equilibrium. In addition, we show that the optimal equilibrium is a weak equilibrium. When the sufficient condition of the attitude function is violated, we illustrate by various examples that the characterization of the optimal equilibrium may differ significantly from some existing results for a single agent, which now sensitively depends on the attitude function and the diversity distribution of discount rates within the group.

\vspace{0.6 cm}

\noindent{\textbf{Mathematics Subject Classification (2020)}: 60G40, 91B10, 91B14, 91A15 }
\vspace{0.1in}

\noindent{\textbf{Keywords}: Diverse discount rates, aggregation attitude function, time inconsistency, mild equilibrium, optimal equilibrium, weak equilibrium}
\end{abstract}

\section{Introduction}
 It has been well documented in experimental studies that individuals show heterogeneous discount rates in evaluating long-term profits and costs. How to reconcile the decision making for a society or a group in face of wide disagreements has been an important and challenging problem in many financial and economic applications. \cite{Weiz} surveyed a large number of economists and asked them to provide the discount rate to evaluate some long-term projects such as proposals to abate climate change. It was shown that the mean of answers is $3.96\%$ with a standard deviation of $2.94\%$, demonstrating a widely dispersed disagreement over the discount rate. \cite{FLO} examined several methods in estimating individual's discount rates, which differ dramatically across different studies and across individuals within the same study. There is no convergence toward an agreed-on or unique rate of impatience in many real life problems. People have divergent opinions on a wide range of subjects, from the growth rate of the economy in future years, the profitability of a new technology to the risk of global warming. The discount rate is subjective to the individual that may vary due to the asymmetric information, the age, the gender, the education background, etc. Aggregating different time preferences in representing a group of members is inevitable in many applications, from government's budgeting decisions and capital investment to cost-benefit analysis of collective risk prevention. 

Some previous studies to address the integration of heterogeneous opinion lean upon the Pareto efficiency, see \cite{Rein2002}, \cite{Jouini10}, \cite{Chambers18}, in which one collective discount rate needs to be endorsed by all members who would agree with the decision under this representative discount rate. However, this approach based on the consensus of the Pareto discount rate by all members is dictatorial and it is in fact equivalent to picking a time preference of one member only among the whole group. On the other hand, an aggregation rule that respects unanimous opinions and diverse discount rates has been shown to be generically time-inconsistent (see, for example, \cite{Jackson15} and \cite{Millner18}), which was regarded as one mathematical challenge in solving the dynamic optimization problem. 

{ 
In the present paper, we consider the context of optimal stopping by embracing all heterogeneous discount rates from group members, and we propose a method to resolve the issue of time inconsistency in the dynamic decision making. First, to capture diverse discount rates from the group members and reflect their impacts in determining the stopping policy, we introduce a random variable discount rate $\rho$ with the distribution function $F_{\rho}(r)$. In addition, we employ an attitude function $\phi(x)$ to depict the way that the central planner puts the weight on the larger or the smaller discount rates in aggregating their expected profits. For example, a linear aggregation rule as $\phi(x)=x$ for diverse discount rates has been studied in \cite{EWZhou20} using the extended HJB equation approach. This paper attempts to generalize the choices of the attitude function $\phi(x)$ by developing the iterative approach. Mathematically, we are interested in the following infinite horizon optimal stopping problem under an aggregation preference over the distribution of discount rates:
\begin{align}\label{stopping-pb}
\sup_{\tau\in \TT} \int_0^\infty \,\phi\left(\E\left[e^{-r \tau}g(X^x_\tau)\right]\right)dF_\rho(r).
\end{align}
Here, $(X_t^x)_{t\geq 0}$ is the underlying state process with the initial value $X^x_0=x$ and $g(x)$ is the payoff function.

Meanwhile, it is worth noting that the problem formulation in \eqref{stopping-pb} is also directly applicable to a single agent's optimal stopping problem when the decision maker is uncertain about his subjective discount rate, which is a type of optimal control problems under model ambiguity. In this regard, our choice of aggregation preference in \eqref{stopping-pb} is akin to the smooth ambiguity preference under model uncertainty coined by \cite{Klibanoff05} and \cite{Klibanoff09}, where the attitude function $\phi(x)$ in \eqref{stopping-pb} over the possible values of the discount rate coincides exactly with the ambiguity aversion function proposed in \cite{Klibanoff05} and \cite{Klibanoff09} to measure the decision maker's attitude towards different levels of model parameters; see Remark \ref{rem-phi} for more detailed discussions. }

Due to the nonlinear nature of the double expectation in \eqref{stopping-pb}, the optimal stopping problem fails to be time consistent. That is, an optimal strategy for today may not retain optimal at future dates so that our future selves may deviate from the optimal strategy that we set out to follow today. As proposed in  \cite{Strotz}, a more meaningful solution to the time inconsistent problem is the consistent planning: the agent chooses the best present action today by taking the future obedience as a constraint such that all future selves will not overturn the current plan.

There has been a surge of interest in time inconsistent stopping problems in the past decade. In addressing the time inconsistency due to non-exponential discounting, an iterative approach has been developed in \cite{HuangNgu18}. A time consistent equilibrium stopping policy therein, also called \textit{mild} equilibrium in later studies \cite{bayraktar21} and \cite{bayraktar22}, corresponds to a fixed point of the iteration operator. This iterative approach has been refined in \cite{HuangZhou19}, \cite{bayraktar19}, \cite{HuangZhou20} and \cite{HuangWang}, in which the existence and characterization of the optimal equilibrium have been established when the non-exponential discounting function is decreasing impatient. Later, this  approach is further generalized in solving some optimal stopping problems under other causes of time inconsistency such as the probability distortion in \cite{HuangNgZ20} and the model uncertainty and $\alpha$-maxmin nonlinear expectation in \cite{HuangYu21}. On the other hand, some time inconsistent stopping problems due to dependence on initial data or nonlinearity in expected rewards have been studied recently in \cite{ChristensenL18} and \cite{ChristensenL20}, in which the weak equilibrium can be characterized by the extended HJB equation method similar to \cite{bjork} in addressing some time inconsistent stochastic control problems under hyperbolic discounting.  
 
We note that, for the linear aggregation $\phi(x)=x$ considered in \cite{EWZhou20}, one key step is that  $$\sup_{\tau\in \TT} \int_0^\infty \E\left[e^{-r \tau}g(X^x_\tau)\right]dF_\rho(r)=\sup_{\tau\in \TT}\E\left[ \int_0^\infty  e^{-r\tau }dF_{\rho}(r)\cdot g(X^x_\tau)\right]=\sup_{\tau\in \TT}\E\left[ \delta (\tau)g(X^x_\tau)\right],$$ 
where the discount rate $\delta(\tau):=\int_0^\infty  e^{-r\tau }dF_{\rho}(r)$ satisfies the decreasing impatience property. In contrast, we are interested in a general diversity attitude function $\phi(x)$ to further reflect the weight of importance in larger discount rates (impatient members) and smaller discount rates (patient members) within the group. As a consequence, we cannot simplify the nonlinear double expectation into one and work with $\delta(\tau)$ as in \cite{EWZhou20}. Instead, we attempt to generalize the iterative approach in \cite{HuangNgu18} to investigate the characterization of mild equilibrium stopping policy under our aggregation preference. Admittedly, due to its definition in a weak sense, there may exist multiple mild equilibria in applications. Therefore, we also aim to discuss some sufficient conditions on the underlying model and the attitude function such that there exists a unique optimal equilibrium. It is interesting to see from our examples that the characterization of the optimal equilibrium may differ significantly from previous results in \cite{HuangZhou19} and \cite{HuangZhou20} as the smallest equilibrium may no longer be optimal in general. Indeed, we provide several examples where the characterization of the optimal equilibrium sensitively depends on the attitude function $\phi(x)$ and the distribution of the random discount rate.

The contributions of the present paper are three-fold. Firstly, we propose an aggregation preference in the optimal stopping problem to respect and incorporate diverse discount rates from all group members. Furthermore, we allow the central planner to choose the attitude function $\phi(x)$ to reflect the weight on different levels of discount rates. As a price to pay, the optimal stopping problem in \eqref{stopping-pb} is time inconsistent, for which we develop an iterative approach to characterize all time-consistent equilibrium stopping policy as fixed points of a carefully defined operator in the setting of one dimensional diffusion processes; See Definition \ref{def:equilibrium} and Theorem \ref{thm-fixed}. Secondly, admitting that there may exist many time-consistent equilibria using the iterative approach, we further discuss some sufficient conditions on the underlying model and the attitude function such that there exists a unique optimal equilibrium, whose value function dominates the ones under all other equilibria when we focus on the Put payoff $g(x)=(K-x)^+$. In particular, under our given assumptions, it is revealed that the optimal equilibrium is the smallest one among all one-barrier equilibria (see Theorem \ref{small}) and it coincides with the optimal one under the linear aggregation function $\phi(x)=x$ (see Remark \ref{rem-equiv}). In addition, when the attitude function $\phi(x)$ is assumed to be concave, we can further show that the obtained optimal equilibrium is a weak equilibrium as proposed in \cite{ChristensenL18} and some subsequent studies; See Proposition \ref{prop-weak}. Thirdly, in concrete examples of geometric Brownian motion and the Bessel process, we also examine the impacts on the optimal equilibrium by $\phi(x)$ under the special choice of $\phi(x)=\min(x,\alpha)$ for some $\alpha>0$ as well as the distribution of the random discount rate $\rho$. It is interesting to see that both $\phi(x)$ and the distribution $F_{\rho}(r)$ may distort the characterization of the optimal equilibrium significantly. Different cases may occur: (i) If $\alpha$ is sufficiently large, the optimal equilibrium is still the smallest equilibrium regardless of the distribution function $F_{\rho}(r)$; (ii) For some small values of $\alpha$ and some $F_{\rho}(r)$, the optimal equilibrium exists and differs from the smallest equilibrium; (iii) For some values of $\alpha$ and some $F_{\rho}(r)$, the optimal equilibrium does not exist, i.e., we cannot find one equilibrium such that its value function can globally dominate the ones under all other equilibria.

The rest of the paper is organized as follows. Section \ref{sec:model} first introduces the model setup with diverse discount rates and the time inconsistent optimal stopping problem under the smooth aggregation preference. In Section \ref{sec:char}, it is shown that any initial stopping policy will converge to a time-consistent equilibrium through a fixed-point iteration for one-dimensional diffusion processes. Section \ref{sec:optim} provides some sufficient conditions on underlying models and the aggregation attitude function $\phi(x)$ when the smallest equilibrium achieves the optimality in terms of its value function. Furthermore, the optimal equilibrium is shown to be a weak equilibrium as studied in the literature under some mild conditions. In Section \ref{sec:impact}, for concrete examples of geometric Brownian motion and Bessel process with $\nu=1/2$, it is illustrated that the optimal equilibrium may differ from the smallest equilibrium, and its existence and characterization depends sensitively on the aggregation attitude function $\phi(x)$ and the distribution function $F_{\rho}(r)$ of diverse discount rates.

\section{Problem Formulation}\label{sec:model}
Let $(\Omega, \F, \P)$ be a probability space supporting a one-dimensional time-homogeneous diffusion process $X=(X_t)_{t\geq 0}\subseteq\R$ (continuous and strong Markov process). Let $\FF=(\F_t)_{t\geq 0}$ be the filtration generated by $X$ and let $\TT$ be the collection of all $\FF$-stopping times. We consider the optimal stopping problem on behalf of a group of agents with diverse discount rates. To capture the heterogeneity in discount rates from group members, let us assume that the discount factor $\rho$ is a random variable with a given distribution function $F_{\rho}(r)$ on $[0,\infty)$, and $\rho$ is independent of the diffusion process $X$. 

{ 
As different discount rate by the group member leads to different recommendation of the stopping policy, how to choose the collective stopping policy by the central planner on behalf of a group becomes an important step in applications. On the other hand, the mathematical problem is closely connected to the framework of a single agent's optimal stopping problem under uncertain discount rates, i.e., the individual decision maker cannot specify the precise value of his subjective discount rate but holds a range of possible discount rates according to a given distribution. In this situation, the decision maker also needs to choose the stopping policy by aggregating the ambiguous discount rates. This motivates us to adopt the idea of the smooth ambiguity preference proposed in \cite{Klibanoff05} and \cite{Klibanoff09}) to address the decision making under Knightian uncertainty over a distribution of discount rates, where an ambiguity aversion function $\phi(x)$ is introduced to describe the weighted rule of the expected utility given the uncertain model parameter.    

From both perspectives of the central planner's problem with diverse discount rates and of the single agent's problem with uncertain discount rates, we are motivated to introduce a distribution $F_\rho$ for the random discount rate $\rho$ and an aggregation attitude function $\phi(x)$ to take into account different expected payoffs under distinct discount rates. This leads us to study the following infinite horizon optimal stopping problem as 
\begin{align*}
\sup_{\tau\in \TT} \int_0^\infty \,\phi\left(\E\left[e^{-r \tau}g(X^x_\tau)\right]\right)dF_\rho(r).
\end{align*}
Mimicking the conventional assumption on the ambiguity attitude function in 
\cite{Klibanoff05} and \cite{Klibanoff09}, we shall mandate that the aggregation attitude function to be strictly increasing.

\begin{remark}\label{rem-phi}
More precisely, our aggregation preference in \eqref{stopping-pb} is inspired by the smooth ambiguity preference in \cite{Klibanoff05} and \cite{Klibanoff09}) defined by
\begin{align}\label{ambut}
\int_{\Delta}\phi\left(\int_S u(f)d\pi\right)d\mu,
\end{align}
where $f$ is a real-valued function defined on a state space $S$, $u$ is a von Neumann-Morgenstern utility function, $\pi$ is a probability distribution on $S$, $\mu$ is the decision maker's subjective prior over the set $\Delta$ of possible probabilities $\pi$ over $S$, and $\phi(\cdot)$ is the ambiguity attitude function characterizing the aversion to mean preserving spreads in $\mu_f$, where $\mu_f$ is the distribution over expected utility values induced by $\mu$ and $f$.  

Comparing with the formulation in \eqref{ambut}, we are in face of the diversity of the discount rates (or the uncertainty of the discount rates in the single agent's problem) following the distribution $F_{\rho}(r)$ similar to the model uncertainty described by the measure $\mu$ over the set of all priors. In particular, our attitude function $\phi(\cdot)$ essentially plays the same role as the one in \eqref{ambut} to indicate the aggregation rules towards different values of $\E\left[e^{-r \tau}g(X^x_\tau)\right]$ under the diverse or uncertain discount rates. 

Similar to Theorem 2 and Corollary 3 in  \cite{Klibanoff05}, one can introduce the attitude coefficient defined by $\zeta(x):=-\frac{\phi''(x)}{\phi(x)}$ in the context of decision making on behalf of a group. For two central planners $1$ and $2$, if $\zeta_1(x)\geq \zeta_2(x)$ for all $x$, the central planner $1$ values more on the larger discount rate in the group than the central planner 2, i.e., the central planner 1 will choose the stopping policy more on behalf of the older members who are more impatient; vice versa, the central planner 2 will cling more to the side of patient members in the group.

\end{remark}
}

Similar to \cite{HuangNgu18}, \cite{HuangZhou20} and some subsequent studies, we choose to study the time-consistent equilibrium stopping within the framework of the mild equilibrium. We aim to extend the iterative approach proposed in \cite{HuangNgu18} and \cite{HuangZhou20} in the current framework under the nonlinear double integrals. 

{ 
Due to the issue of time inconsistency, the central planner may re-evaluate and change the collective stopping policy over-time. The consistent planning in  \cite{Strotz} suggests a game-theoretic thinking against all future selves. In the current context, the future selves refer to the central planner who makes the stopping policy on behalf of the group. The goal of the central planner is to find an equilibrium stopping policy $R$ such that all future selves of the central planner will not overturn. To simplify some future arguments, let us only focus on stopping policies among all hitting times, i.e., the central planner decides to stop or to continue depending on the current state $x\in\mathbb{R}$.
}

\begin{definition}[Stopping Policy]
A Borel measurable $\tau:\R\to\{0,1\}$ is called a stopping policy and it is equivalent to
\begin{equation*}
	\tau(x)=\mathbbm{1}_{R^c}(x),\qquad\text{for some }R\in\B(\R).
\end{equation*}
\end{definition}

The game theoretic thinking can be carried out in the following way: Assuming that all future selves will follow an arbitrary stopping policy $R\in\B(\R)$, the decision maker has two options: stopping or continuation. If the decision maker stops, he can get the immediate payoff $\phi(g(x))$; if he decides to continue and suppose the current state is $x\in\R$, the decision maker will eventually stop at the instant
\begin{equation*}
	\tau(x,R):=\inf\{t>0,X_t^x\in R\},
\end{equation*}
and the expected future payoff $J(x,R)$ under the smooth aggregation preference is given by
\begin{equation*}
	J(x,R):=\int_0^\infty \phi\Big(\E\Big[e^{-r \tau(x,R)}g(X^x_{\tau(x,R)})\Big]\Big) dF_\rho(r).
\end{equation*}

To determine the stopping policy for today, in response to future selves following $R\in \B(\R)$, the decision maker can improve the current stopping policy $R\in\B(\R)$ by comparing $\phi(g(x))$ and $J(x,R)$. It then leads to the following \textit{policy improvement operator} $\Theta(\cdot)$ defined by
\begin{equation*}
	\Theta(R):=S_R\cup (I_R\cap R).
\end{equation*}  
where we define the stopping region, the indifference region and the continuation region by $S_R,I_R$ and $C_R$ respectively:
\begin{equation*}
	\begin{aligned}
	S_R:=\{x\in\R:\;\phi(g(x))>J(x,R)\},
	\\
	I_R:=\{x\in\R:\;\phi(g(x))=J(x,R)\},
	\\
	C_R:=\{x\in\R:\;\phi(g(x))<J(x,R)\}.
	\end{aligned}
\end{equation*}

Using Lemma 4.4 of \cite{HuangZhou20}, we have the following result.
\begin{lemma}\label{lem-onediff}
If $X$ is a one dimension diffusion process, we have $\tau(x,R)=\tau(x,\bar R)=0$ for any $x\in \bar R$ and $R\in\B(\R)$, where $\bar R$ is the closure of $R$. Consequently, $\phi(g(x)) = J(x,R)$, for any $x\in \bar R$ and $R\in\B(\R)$.

\end{lemma}
By Lemma \ref{lem-onediff}, the policy improvement operator can be reduced to 
$$\Theta(R)=S_R\cup R,$$ 
and it hence follows that $R\subseteq \Theta(R)$.

\section{Characterization of Time-consistent Equilibria}\label{sec:char}
Based on the previous game theoretic thinking and the iteration improvement by the operator $\Theta(\cdot)$, a time-consistent equilibrium stopping policy $R$ is the one when $\Theta(R)$ cannot improve the stopping region $R$ and hence all future selves have no incentives to deviate from $R$. It then naturally leads to a fixed point of the operator $\Theta(\cdot)$. We can then give the definition of an equilibrium in the following sense. 

\begin{definition}[Equilibrium] \label{def:equilibrium}
	$R\in\B(\R)$ is called an equilibrium if it is the fixed point of the operator $\Theta(\cdot)$ that $\Theta(R)=R$. Let us denote $\EE$ as the set of all equilibria.
\end{definition}

Note that $R=\R$ is a trivial equilibrium, and hence the set $\EE$ is not empty. The next natural question is: How to characterize all time-consistent equilibria in the set $\EE$? In response to this, we look for an equilibrium using the iteration operator: Suppose that we start with an arbitrary $R\in\B(\R)$, can we apply $\Theta$ repetitively to reach the fixed point? As we have $R\subseteq\Theta(R)$, it follows that $(\Theta^n(R))_{n\in\mathbb{N}}$ is a nondecreasing sequence of sets in $\R$ and hence $R^*:=\lim_{n\rightarrow\infty}\Theta^n(R)=\cup_{n=0}^\infty\Theta^n(R)$ is well-defined for any  $R\in\B(\R)$. Our next result shows that the limit set $R^*$ starting from an arbitrary $R\in\B(\R)$ is a fixed point of $\Theta(\cdot)$ and is a time-consistent equilibrium. We therefore can characterize all equilibria in $\EE$ in the same fashion.

\begin{theorem}\label{thm-fixed}
	The set $\EE$ can be characterized by
	\begin{equation*}
		\EE=\Big\{\lim_{n\to\infty}\Theta^n(R)=\cup_{n=0}^\infty\Theta^n(R):R\in\B(\R)\Big\}
	\end{equation*}
\end{theorem}

\begin{proof}
By the definition of the equilibrium, if $R\in\EE$, then $\Theta^n(R)=R$, which leads to $\lim_{n\to\infty}\Theta^n(R)=R$ and verifies that
\begin{equation*} 
\EE\subset\Big\{\lim_{n\to\infty}\Theta^n(R)=\cup_{n=0}^\infty\Theta^n(R):R\in\B(\R)\Big\}
\end{equation*}
To prove the other direction, let us denote $R_n:=\Theta^n(R)$ and since $R\subset \Theta(R)$, we have $R_n\subset R_{n+1}$. Therefore, it suffices to show that $\Theta(R^*)=R^*$, where $R^*:=\lim_{n\to\infty}R_n=\cup_{n=0}^\infty R_n$.  It is equivalent to prove that for $x\notin R^*$, $J(x,R^*)\geq \phi(g(x))$. As $x\notin R^*$, we know that $x\notin \Theta(R_n)$ for all $n\in\N$, and therefore 
\begin{equation*}
\int_0^\infty\phi\Big(\E\Big[e^{-r \tau(x,R_n)}g(X^x_{\tau(x,R_n)})\Big]\Big)dF_\rho(r)=J(x,R_n)\geq \phi(g(x)).
\end{equation*}
{ Note that $\tau(x,R_n)$ and $\tau(x,R)$ only depend on the dynamics of $X$ regardless of the choice of the attitude function $\phi(x)$. We therefore can adopt Lemma 3.10 in \cite{HuangYu21} for one-dimensional diffusion process to get the convergence that}, for each $x\in\R$,
\begin{equation*}
	\tau(x,R_n)\to\tau(x,R^*),\qquad\text{a.s.}
\end{equation*}
Let us define $p_n$, $q_n$ by
$$
p_n:=\sup\{y\in R_n:y<x\}\text{ and }q_n:=\inf\{y\in R_n:y>x\},
$$
As  $R^*=\lim_{n\to\infty}R_n=\cup_{n=0}^\infty R_n$, we have $p_n\uparrow p$ and $q_n\downarrow q$, where
$$
p:=\sup\{y\in R_*:y<x\}\text{ and }q:=\inf\{y\in R_*:y>x\}.
$$
Note that $X^x_{\tau(x,R_n)}$ is either $p_n$ or $q_n$, then there exist some deterministic $-\infty<p^*<p<q<q^*<\infty$, and $n^*\in\N$ such that for $n>n^*$ and $\tau(x,R_n)<\infty$, we have 
\begin{equation*}
	X^x_{\tau(x,R_n)}\in[p^*,q^*],\qquad\text{a.s.}
\end{equation*}
and 
\begin{equation*}
	\mathbbm{1}_{\{\tau(x,R^*)=\infty\}}\cdot e^{-\rho \tau(x,R^*)}g(X^x_{\tau(x,R^*)})=\mathbbm{1}_{\{\tau(x,R_n)=\infty\}}\cdot e^{-\rho \tau(x,R_n)}g(X^x_{\tau(x,R_n)})=0,\qquad\text{a.s..}
\end{equation*}
Finally, by the continuity of $g$ and $\phi$ and the Dominate Convergence Theorem, we have
\begin{equation*}
	\begin{aligned}
	J(x,R^*)&=\int_0^\infty\phi\Big(\E\Big[\mathbbm{1}_{\{\tau(x,R^*)<\infty\}}\cdot e^{-r \tau(x,R^*)}g(X^x_{\tau(x,R^*)})\Big]\Big)dF_{\rho}(r)
    \\
	&=\lim_{n\to\infty}\int_0^\infty\phi\Big(\E\Big[\mathbbm{1}_{\{\tau(x,R_n)<\infty\}}\cdot e^{-r \tau(x,R_n)}g(X^x_{\tau(x,R_n)})\Big]\Big)dF_\rho(r)
	\\
	&=\lim_{n\to\infty} J(x,R_n)\geq \phi(g(x)),
	\end{aligned}
\end{equation*}
which completes the proof.
\end{proof}

{ 
\begin{remark}
We note that our arguments above are only applicable to the setting of one-dimensional diffusion processes. In particular, the ordered inclusion $R\subseteq \Theta(R)$, and hence the monotonicity $R_n\subseteq R_{n+1}$, relies on Lemma 4.4 of \cite{HuangZhou20} in the one-dimensional diffusion setting. Note that in \cite{HuangNgu18}, this monotonicity result $R_n\subseteq R_{n+1}$ is also established in a multi-dimensional setting under the additional assumption of the decreasing impatience structure of the discount, which unfortunately does not hold in our framework due to the double expectation under the attitude function.    
\end{remark}
}

\section{Sufficient Conditions on the Existence of Optimal Equilibrium}\label{sec:optim}

Let us consider the real option valuation and capital budgeting decision by the government who needs to reconcile the intergenerational disagreement on the discount rate. In particular, we shall consider the decision making by the government on the cessation of investment in a project such as an R\&D project. By adopting the financial option valuation techniques, we encounter an infinite horizon optimal stopping problem with the aggregation of diverse discount rates that
\begin{align*}
\sup_{\tau\in \TT} \int_0^\infty \,\phi\left(\E\left[e^{-r \tau}g(X^x_\tau)\right]\right)dF_\rho(r)
\end{align*}
with the Put payoff function $g(x)=(K-x)^+$ for some $K>0$. For technical convenience, we shall only focus on the Put payoff function in the present paper, and will leave the general case for future investigation. The above {  $F_\rho(r)$ is the distribution function of the discounting rate $\rho$ with support $[0, \infty)$ that} describes the distribution of diverse discount rates from the group members, and the attitude function $\phi(x)$ captures the central planner's weight towards different time preference. {  Throughout the paper, when we refer to \textit{increasing} (resp. \textit{decreasing}) we mean \textit{non-decreasing} (resp. \textit{non-increasing}).}

Our goal in this section is to examine some sufficient conditions on the underlying one dimensional diffusion process and the attitude function $\phi(x)$ such that the smallest time-consistent equilibrium is the optimal equilibrium for our time inconsistent stopping problem \eqref{stopping-pb} when all equilibrium stopping policies fit the type of one-barrier stopping region. We shall only focus on the underlying project process modeled by a non-negative one dimensional diffusion process $X_t\geq 0$ a.s. for $t\geq 0$. More precisely, it satisfies the stochastic differential equation:
 $$
 d X_t = \mu(X_t) dt + \sigma(X_t) d W_t,
 $$
 where $\mu, \sigma: \R \to \R$ are Lipschitz continuous, and $\sigma^2(x) > 0$ when $x>0$. This guarantees that the above equation admits a unique strong solution given that $X_0 = x \in {   [0,\infty)}$. 

{   
Let us first introduce some notations. Fix some number $\epsilon>0$, we define the scale function $u(x)$ and the speed measure $m(y)\d y$ as below:
\begin{align*}
m(y)\d y := \frac{1}{\sigma^2(y)} \exp \left( \int_\epsilon^y \frac{2 \mu(z)}{\sigma^2(z)} dz \right) \d y,
\\
u(x):= \int_\epsilon^x \exp \left( - \int_\epsilon^y \frac{2 \mu(z)}{\sigma^2(z)} \d z \right) \d y.
\end{align*}

The following conditions on model parameters and the attitude function are imposed in this section.
\ \\
\textbf{C-(i)}: It is assumed that $\mu(x) \geq 0$ for all $x \geq 0$. 

\vspace{2mm}

\hspace{-7mm} 
\textbf{C-(ii)}: We define
$$
\phi_1(x):=\frac{\mu(x) + \sqrt{\mu^2(x)+ 2r \sigma^2(x)}}{\sigma^2(x)},
$$
and further define $I_1,I_2$ by
$$
I_1:=\int_\epsilon^{\infty} u'(x)\cdot \Big(\int_\epsilon^xm(y) \d y\Big) \d x,\quad 
I_2:=\int_\epsilon^{\infty} u(x)m(x) \d x.
$$

We assume that the following conditions hold:
\begin{enumerate}
    \item $\phi'_1(x)<0$ for all $x \geq 0$, and
$$
\limsup_{x\to\infty} \frac{\mu(x)+1}{\sigma^2(x)} < \infty.
$$
    \item $I_1=+\infty$, $I_2=+\infty$.
\end{enumerate}

\vspace{2mm}

\hspace{-7mm} 
\textbf{C-(iii)}: The attitude function $\phi(x)$ is $C^1$ and satisfies that $\phi'(x)x$ is an increasing function and $\phi(x)$ is strictly increasing.
}

\subsection{Discussions on the conditions imposed}\label{sec4.1}

\textbf{Discussion on condition  \textbf{C-(i)}}

Similar to Lemma 4.3 of \cite{HuangYu21}, under Assumption \textbf{C-(i)},  the time-consistent equilibrium is of the one-barrier type, satisfying the form of $[0,a]$ for some $0\leq a\leq K$.

\begin{lemma}\label{lem-oneb}
If Condition \textbf{C-(i)} holds, {  $F_{\rho}(0) < 1$} and $\phi(x)$ is strictly increasing, any closed time-consistent equilibrium contained in $[0,K]$ satisfies the form of $R=[0,a]$ for some $0\leq a\leq K$. 
\end{lemma}
\begin{proof}
As $\mu(x) \geq 0$ for all $x \geq 0$, $K-X_t$ is supermartingale.
Let $R\in\mathcal{E}$ be a time-consistent equilibrium contained in $[0,K]$ and $a:=\sup\{x:x\in R,x\leq K\}$. Suppose that there exists $x\in (0,a)$ such that $x\notin R$. By the closeness of $R$, we have $\tau(x,R)>0$, $\mathbb{P}$-a.s.. Since $R$ is contained in $[0,K]$, we have $0<x<a\leq K$ and therefore {  $X_{\tau(x,R)}^x\leq K$}.  Hence
\begin{align*}
J(x,R)&=\int_0^\infty \phi\left(\mathbb{E}[e^{-r\tau(x,R)} (K-X^x_{\tau(x,R)})^+]\right)d F_\rho(r)\\
&=\int_0^\infty \phi\left(\mathbb{E}[e^{-r\tau(x,R)} (K-X^x_{\tau(x,R)})]\right)d F_\rho(r)
\\
&<\int_0^\infty \phi\left(\mathbb{E}[K-X^x_{\tau(x,R)}]\right)\mathbf{1}_{\{r>0\}} d F_\rho(r) + \int_0^\infty \phi\left(\mathbb{E}[K-X^x_{\tau(x,R)}]\right)\mathbf{1}_{\{r=0\}} d F_\rho(r) \\
&\leq  \phi(K-x) \mathbf{1}_{\{r>0\}} +  \phi(K-x) \mathbf{1}_{\{r=0\}} =\phi(K-x),
\end{align*}
which is a contradiction to the definition of the equilibrium. Here, the first strict inequality is due to the fact that $e^{-r\tau(x,R)} < 1$ when $r > 0$ and $e^{-r\tau(x,R)} \leq 1$ when $r \geq 0$; The second inequality results from the super-martingale property. Therefore, it holds that any time-consistent equilibrium satisfies the form of $R=[0,a]$ under the given assumptions.
\end{proof}
\vspace{0.1in}

\noindent\textbf{Discussion on condition  \textbf{C-(ii)}}
{ 

In what follows, we show that Condition \textbf{C-(ii)} is related to some monotonicity conditions of the
moment generating function of the first hitting time $\tau_a^x$.

 \begin{proposition}\label{prop4.2} Define $H(r,a,x):=\ln(\mathbb{E}[e^{-r \tau_a^x}])$, where $\tau_a^x:=\inf\{t>0: X_t^x=a\}$ denotes the first hitting time to the barrier $a$ by the diffusion process $X^x_t$ with $X^x_0=x\geq a$. Under Condition \textbf{C-(ii)}, we have
\begin{itemize}
\item[($a$)] 
$\lim_{x\downarrow a}\frac{\partial }{\partial x}H(r,a,x)$ is increasing in $a\in [0, K]$ for the fixed $r$;
\item[($b$)] $\frac{\partial}{\partial a} H(r,a,x)$ is increasing in $r$ for the fixed $a\leq x$.
\end{itemize}
     \end{proposition}

\begin{remark}\label{cond-process}
As $(X_t)_{t\geq 0}$ is a one-dimensional diffusion process, it follows from \cite[Part 1, Chapter II, Item 10]{BorodinSalminen}, that $\E [e^{-r \tau_a^x}]=\frac{\varphi_{r}(x)}{\varphi_{r}(a)}$, where $\varphi_{r}(x)$ is the (up to a multiplicative constant) decreasing positive solution of the ODE
\begin{align}\label{phiODE}
\frac{1}{2} \sigma^2(x) u''(x) + \mu(x) u'(x) = r u(x),
\end{align}
see \cite[Part 1, Chapter II, Item 9]{BorodinSalminen}. 
It follows that $H(r,a,x)=\ln (\varphi_{r}(x))-\ln(\varphi_{r}(a))$. Therefore, conditions $(a)$ and $(b)$ in Proposition \eqref{prop4.2}  on the hitting time $\tau_a^x$ can actually be equivalently written as the ones on the function $\varphi_{r}(x)$ associated to the underlying process $(X_t)_{t\geq 0}$ that
\begin{itemize}
\item[($a'$)] $\frac{\varphi'_{r}(x)}{\varphi_{r}(x)}$ is increasing in $x\in[0,K]$ for the fixed $r$;

\item[($b'$)] $-\frac{\varphi'_{r}(x)}{\varphi_{r}(x)}$ is increasing in $r$ for the fixed $x\in [0,K]$.
\end{itemize}
Moreover, as $\varphi_{r}(x)$ is the unique solution to the ODE \eqref{phiODE}, let us define $m_{r}(x):=\frac{\varphi'_{r}(x)}{\varphi_{r}(x)}$, it holds that $m_{r}(x)$ is the unique solution to the Riccati ODE
\begin{align}\label{mODE}
\frac{1}{2}{\sigma}^2(x)m'_{r}(x)+\frac{1}{2}\sigma^2(x)m_{r}^2(x)+\mu(x)m_{r}(x)-r=0, 
\end{align}
That is, for a given one-dimensional diffusion process $(X_t)_{t\geq 0}$, the conditions $(a)$ and $(b)$ in \textbf{C-(ii)} are fulfilled if and only if the solution $m(x)$ to the Riccati ODE \eqref{mODE} satisfies
\begin{itemize}
\item[($a''$)] $m_{r}(x)$ is increasing in $x\in[0,K]$ for the fixed $r$; 
\item[$(b'')$] $-m_{r}(x)$ is increasing in $r$ for the fixed $x\in [0,K]$.
\end{itemize}
\end{remark}



Let us first show that the Condition (a) in Proposition \ref{prop4.2} is implied by Condition \textbf{C-(ii)}-$1$.
\begin{proposition}\label{prop:main}
If Condition \textbf{C-(ii)-$1$} is valid, then Condition (a) in Proposition \ref{prop4.2} holds.
\end{proposition}
\begin{proof}
Introducing $\tilde{m}_r(x): = - m_r(x)>0$, we have that $\tilde{m}_r$ satisfies
\begin{equation} \label{eq:riccati}
\tilde{m}'_r(x) = \tilde{m}^2_r(x)-\frac{2\mu(x)}{\sigma^2(x)} \tilde{m}_r(x) - \frac{2r}{\sigma^2(x)}.
\end{equation}
Thanks to Remark \ref{cond-process}, it suffices to justify that $x \mapsto \tilde{m}_r(x)$ decreasing. 
To continue, we will first show the claim in the next lemma. Recall that a point is \textit{accessible}, if it can be reached by the process in finite time with positive probability.
\begin{lemma} \label{lemma:access}
Under Condition \textbf{C-(ii)-$1$}, any point $x_0\in (0, \infty)$ is accessible.
\end{lemma}
\begin{proof}
Following \cite[Theorem 8, item A]{Helland} (see also \cite[Proposition 16.43]{Breiman}), it is enough to check that $u(x_0) < \infty$, and $\int_b^{x_0} (u(x_0)-u(y)) m(dy) < \infty$, for an arbitrary point $b \in (0, x_0)$. This is straightforward, as Condition \textbf{C-(ii)-$1$} guarantees that both $\frac{\mu(x)}{\sigma^2(x)}$ and $\frac{1}{\sigma^2(x)}$ are bounded by some constant $C>0$. Consequently, $u(x_0) \leq x_0 - \epsilon < \infty$, and 
$$
\int_b^{x_0} (u(x_0)-u(y)) m(dy) \leq \int_b^{x_0} \left[ (x_0-\epsilon)) + (y- \epsilon) \right] C e^{C(y-\epsilon)} dy < \infty.
$$
\end{proof}
Thanks to Lemma \ref{lemma:access}, we have $\tilde{m}_r(x) < \infty$ for any $x \in (0, \infty)$, namely, $\tilde{m}_r$ cannot explode at any finite point. Otherwise, suppose there exists $x_0 \in (0, \infty)$, such that $\tilde{m}_r(x_0) = \infty$. From the definition $\tilde{m}_{r}(x):=-\frac{\varphi'_{r}(x)}{\varphi_{r}(x)}$, we have $\varphi_{r}(x_0) = \infty$. As $x \mapsto \varphi_{r}(x)$ is decreasing and positive, that $\varphi_{r}(x)=0$ for any $x \geq x_0$. As $\E [e^{-r \tau_a^x}]=\frac{\varphi_{r}(x)}{\varphi_{r}(a)}$, we have $\tau_a^x = \infty$ for any $x < x_0 < a$ $\P$-almost surely. However, from Lemma \ref{lemma:access}, $\P(\tau_a^x < \infty) >0$, and we have a contradiction.

Consider the corresponding characteristic function:
$$
\phi^2(x) - \frac{2 \mu(x)}{\sigma^2(x)} \phi(x) - \frac{2 r}{\sigma^2(x)} = 0.
$$
The two roots are
$$
\phi_1(x)=\frac{\mu(x) + \sqrt{\mu^2(x)+ 2r \sigma^2(x)}}{\sigma^2(x)}>0, \ \ \phi_2(x)=\frac{\mu(x) - \sqrt{\mu^2(x)+ 2r \sigma^2(x)}}{\sigma^2(x)}<0.
$$
Therefore,
$$
\tilde{m}'_r(x) = \tilde{m}^2_r(x)-\frac{2\mu(x)}{\sigma^2(x)} \tilde{m}_r(x) - \frac{2r}{\sigma^2(x)}=\big(\tilde{m}_r(x)-\phi_1(x)\big)\big(\tilde{m}_r(x)-\phi_2(x)\big).
$$
The next lemma establishes $\tilde{m}_r(x) < \phi_1(x)$ for all $x>0$. 
\begin{lemma} \label{lemma: mono_x}
Under Condition \textbf{C-(ii)-$1$}, $\tilde{m}_r(x) < \phi_1(x)$ for all $x>0$.
\end{lemma}
\begin{proof}
In the following, we shall consider two sub-scenarios:

Case 1: Suppose there exists $x_0>0$, such that $\tilde{m}_r(x_0)>\phi_1(x_0)>0$, then (as $\phi_2<0$)
$$
\tilde{m}'_r(x_0) = (\tilde{m}_r(x_0) - \phi_1(x_0)) (\tilde{m}_r(x_0) - \phi_2(x_0)) > 0.
$$
Then $\tilde{m}'_r(x)>0$ on $(x_0, +\infty)$(equivalently, $\tilde{m}_r(x)>\phi_1(x)$). Otherwise, denote 
$$
x_1:= \min \{ x > x_0: \tilde{m}'_r(x) =0 \}.
$$
we have $\tilde{m}'_r(x_1)=0$, and consequently $\tilde{m}_r(x_1)= \phi_1(x_1)$. Using the fact that $\tilde{m}_r(x_0) > \phi_1(x_0)$, we have
$$
0 < \int_{x_0}^{x_1} \tilde{m}'_r(x) dx = \tilde{m}_r(x_1) - \tilde{m}_r(x_0) < \phi_1(x_1) - \phi_1(x_0) \leq 0.
$$
Hence a contraction. On the above, the first inequality is because of $\tilde{m}'_r(x)>0$ for all $x \in (x_0, x_1)$. Now for $x > x_0$
$$
\tilde{m}'_r(x) \geq (\tilde{m}_r(x)-\phi_1(x)) \tilde{m}_r(x) \geq (\tilde{m}_r(x_0)-\phi_1(x_0)) \tilde{m}_r(x_0):=C_1 >0.
$$
Consequently, $\tilde{m}_r(x) \to +\infty$ when $x \to +\infty$. From \eqref{eq:riccati}, we get 
$$
\left( \frac{1}{\tilde{m}_r(x) } \right)' = - \frac{\tilde{m}'_r(x)}{\tilde{m}^2_r(x)} =  -1 + \frac{2\mu(x)}{\sigma^2(x) \tilde{m}_r(x)} + \frac{2r}{\sigma^2(x) \tilde{m}^2_r(x)},
$$
As when $x \to +\infty$, $\frac{2r}{\sigma^2(x)}, \frac{\mu(x)}{\sigma^2(x)}$ are bounded, $\tilde{m}_r(x) \to +\infty$, it follows that there exists some $0<C_2<1$ and some constant $C_3$, such that
$$
 \frac{1}{\tilde{m}_r(x)} \leq -C_2 x + C_3,
$$
and is strictly negative for $x$ large enough. This is a contradiction to $\tilde{m}_r(x) \to +\infty$.

Case 2: If there exists $x_0 \in (0, \infty)$, such that $\tilde{m}_r(x_0)=\phi_1(x_0)$. As $\phi_1'(x_0) <0$, $\tilde{m}'_r(x_0)=0$, there exists $\delta>0$ small enough, such that 
$$
\tilde{m}_r(x_0+\delta) - \phi_1(x_0+\delta) >0.
$$
Hence we return to Case 1.

Combining Case 1 and Case 2, we draw the conclusion.
\end{proof}
Given that $\tilde{m}_r(x) < \phi_1(x)$ for all $x>0$, and when this is combined with $\tilde{m}_r(x) >0 > \phi_2(x)$,it follows that  $\tilde{m}'_r(x)<0$, indicating that the function $x \mapsto \tilde{m}_r(x)$ decreasing.
\end{proof}

Next, we move to the Condition (b) in Proposition \ref{prop4.2}. 

\begin{proposition}
    If Condition \textbf{C-(ii)} holds, then Condition (b) in Proposition \ref{prop4.2} also holds.
\end{proposition}
\begin{proof}
    By Condition $(b'')$ in Remark \ref{cond-process}, and the definition of $\tilde m_r(x)$ in the previous proposition, it suffices to show that $r \mapsto \tilde m_r(x)$ is increasing, that is for any $r_1>r_2$, $\tilde m_{r_1}(x)\geq\tilde m_{r_2}(x)$ for all $x\geq 0$. Note that $\tilde m_{r_1}$ and $\tilde m_{r_2}$ satisfy the Ricatti equations
    \begin{equation}
        \begin{aligned}
            \tilde m_{r_1}'(x)=\tilde m_{r_1}^2(x)-\frac{2\mu(x)}{\sigma^2(x)}\tilde m_{r_1}(x)-\frac{2r_1}{\sigma^2(x)},
            \\
            \tilde m_{r_2}'(x)=\tilde m_{r_2}^2(x)-\frac{2\mu(x)}{\sigma^2(x)}\tilde m_{r_2}(x)-\frac{2r_2}{\sigma^2(x)},
        \end{aligned}
    \end{equation}
    Define $a(x):=\tilde m_{r_1}(x)-\tilde m_{r_2}(x)$ and $b(x)=\tilde m_{r_1}(x)+\tilde m_{r_2}(x)-\frac{2\mu(x)}{\sigma^2(x)}$ and $c(x)=\frac{2r_1(x)}{\sigma^2(x)}$, then
    $$
    a'(x)=a(x)b(x)-c(x).
    $$
    It is a linear equation and we can solve it explicitly. Define 
    $$
    f(x):=\exp\Big(-\int_{\epsilon}^xb(y)\d y\Big),
    $$
    for some $\epsilon>0$. Since $\tilde m_{r_i}(x)<\phi_1(x)$ for $i=1,2;$ and $\limsup_{x\to\infty}\frac{\mu(x)+1}{\sigma^2(x)}<\infty$, we know that $f(x)$ is well-defined on $x\in(0,\infty)$. Further, since $f(x)\geq 0$ and $c(x)\geq 0$ for all $x\in(0,\infty)$, $\int_x^\infty f(y)c(y)\d y$ is well-defined and takes value in $[0,+\infty]$. Then, by $\big(a(x)f(x)\big)'=a'(x)f(x)-a(x)f(x)b(x)=-f(x)c(x)$,
    we have
    $$
    a(x)f(x)=\lim_{y\to\infty }a(y)f(y)+\int_x^\infty f(y)c(y)\d y.
    $$
    We claim that $\lim_{y\to\infty }a(y)f(y)=0$, and this leads to our desired result $a(x)\geq 0.$ Referring back to the definition $\varphi_r(x)$ in Remark \ref{cond-process} and noting that $m_r(x)=\frac{\varphi_r'(x)}{\varphi_r(x)}=\big(\log m_r(x)\big)'$, we can derive the following relationship,
    $$
    \varphi_r(x)=\exp\Big(\int_\epsilon^xm_r(y)\d y\Big)=\exp\Big(\int_\epsilon^x-\tilde m_r(y)\d y\Big),\quad\varphi'_r(x)=-\tilde m_r(x)\exp\Big(-\int_\epsilon^x\tilde m_r(y)\d y\Big).
    $$
    Note that 
    \begin{align*}
    |a(x)f(x)|&=\bigg|\Big(\tilde m_{r_1}(x)-\tilde m_{r_2}(x)\Big)\exp\Big(-\int_{\epsilon}^x\Big(\tilde m_{r_1}(y)+\tilde m_{r_2}(y)-\frac{2\mu(y)}{\sigma^2(y)}\Big)\d y\Big)\bigg|
    \\
    &\leq \Big(\tilde m_{r_1}(x)+\tilde m_{r_2}(x)\Big)\exp\Big(-\int_{\epsilon}^x\Big(\tilde m_{r_1}(y)+\tilde m_{r_2}(y)-\frac{2\mu(y)}{\sigma^2(y)}\Big)\d y\Big)
    \\
    &=-\exp\Big(\int_{\epsilon}^x\frac{2\mu(y)}{\sigma^2(y)}\d y\Big)\Big(\varphi_{r_1}'(x)\varphi_{r_2}(x)+\varphi_{r_1}(x)\varphi_{r_2}'(x)\Big).
    \end{align*}
    Recall that $\varphi_r(x)$ solves the linear ODE
    $$
    \frac{1}{2} \sigma^2(x) \varphi''(x) + \mu(x) \varphi_r'(x) = r \varphi_r(x),
    $$
    Since $I_1=I_2=+\infty$, we can apply \cite[Theorem 2 \& 3]{Cecchi1989} to deduce that 
    $$
    \varphi_r(x)\to 0,\quad\varphi_r'(x)\exp\Big(\int_{\epsilon}^x\frac{2\mu(y)}{\sigma^2(y)}\d y\Big)\to 0\quad\text{ as }x\to +\infty.
    $$
    This proves that $a(x)f(x)\to 0$ and thus completes our proof.
    
\end{proof}

The following three examples illustrate that our conditions on the underlying model can be satisfied by many popular stochastic processes in financial applications. 

\noindent
\textbf{Example-1}: If $X_t$ is a geometric Brownian motion that $dX_t=\mu X_tdt+\sigma X_tdW_t$ with $\mu>0$ and $\sigma>0$, and hence the condition \textbf{C-(i)} holds. 
We can calculate $\phi_1(x)=\frac{\mu}{\sigma^2 x} + \sqrt{\frac{\mu^2 + 2r \sigma^2}{\sigma^4 x^2}}$ and it is clear that $\phi_1(x)$ is decreasing in $x$. In addition, $m(y) = \frac{1}{\sigma^2 y^2} (\frac{y}{\epsilon})^{\frac{2\mu}{\sigma^2}}$ and 
$u(x)= \int_\epsilon^x (\frac{y}{\epsilon})^{-\frac{2\mu}{\sigma^2}} \d y$. Consequently, 
$$I_1:=\int_\epsilon^{\infty} x^{-\frac{2\mu}{\sigma^2}} \int_{\epsilon}^x  \frac{1}{\sigma^2 y^2} y^{\frac{2\mu}{\sigma^2}} \d y \d x.$$ 
When $2\mu=\sigma^2$, 
$$
I_1=\frac{1}{\sigma^2} \frac{(\ln y)^2}{2} \bigg|_{1}^{\infty}=\infty.
$$
When $2\mu \neq \sigma^2$, 
$$
I_1= \frac{1}{2\mu-\sigma^2}\int_{\epsilon}^x \left[ x^{-1} - x^{-\frac{2\mu}{\sigma^2}} \epsilon^{\frac{2\mu}{\sigma^2}-1} \right] dx = +\infty.
$$
On the other hand,
$$
I_2:=\int_\epsilon^{\infty} u(x)m(x) \d x= \int_\epsilon^{\infty}\left( \int_\epsilon^x y^{-\frac{2\mu}{\sigma^2}} \d y \right) \frac{1}{\sigma^2 x^2} x^{\frac{2\mu}{\sigma^2}} \d x.
$$
When $2\mu = \sigma^2$,
$$
I_2=\frac{1}{\sigma^2} \frac{(\ln y)^2}{2} \bigg|_{1}^{\infty}=+\infty.
$$
When $2\mu \neq \sigma^2$, we have
$$
I_2= \frac{1}{\sigma^2-2\mu}\int_{\epsilon}^\infty \left[ x^{-1} - x^{\frac{2\mu}{\sigma^2}-2} \epsilon^{-\frac{2\mu}{\sigma^2}+1} \right] dx = +\infty.
$$
To conclude, we justify Condition \textbf{C-(ii)}-$2$.

\noindent
\textbf{Example-2}:  If $X_t$ is a general Bessel process with degree $\nu$ and $n=2 \nu+2$, in this case 
$$
\d X_t=\frac{n-1}{2X_t}\d t+\d W_t,
$$
with $\sigma(x)=1, \mu(x)=\frac{n-1}{2x}$. Then $\phi_1(x)=\frac{n-1}{2x}+\sqrt{\big(\frac{n-1}{2x}\big)^2+2r}$ is clearly decreasing with respect to $x$.
Further, $m(y) := e^{(n-1)\ln(\frac y \epsilon)}=\big(\frac y\epsilon\big)^{n-1} $, and
$$u(x)= \int_\epsilon^x e^{-(n-1)\ln(\frac y \epsilon)} \d y=\int_\epsilon^x\Big(\frac{y}\epsilon\Big)^{1-n}\d y=\begin{cases}
    \frac{x^{2-n}}{(2-n)\epsilon^{1-n}}-\frac{\epsilon}{2-n},&n\neq 2;\\
    \epsilon\ln(x/\epsilon),&n=2.
\end{cases}$$ 
Consequently, 
\begin{align*}
I_1&=\int_\epsilon^{\infty} u'(x)\Big(\int_\epsilon^x m(y) \d y\Big) \d x=\int_\epsilon^{\infty} \Big(\frac{x}\epsilon\Big)^{1-n}\Big(\int_\epsilon^x \Big(\frac{y}\epsilon\Big)^{n-1} \d y\Big) \d x
\\
&=\begin{cases}
    \int_\epsilon^{\infty} \big(x/\epsilon\big)^{1-n}\cdot\frac{\epsilon}{n}\Big((x/\epsilon)^n-1\Big) \d x=+\infty,&n\neq 0;
    \\
    \int_\epsilon^{\infty} x\cdot\ln(x/\epsilon) \d x=+\infty,&n=0,
\end{cases}
\end{align*}
and
\begin{align*}
    I_2&=\int_\epsilon^{\infty} u(x)m(x) \d x=\int_\epsilon^{\infty}u(x)\big(\frac x\epsilon\big)^{n-1}\d x
    \\
    &=\begin{cases}
    \int_\epsilon^{\infty}\big(\frac{x^{2-n}}{(2-n)\epsilon^{1-n}}-\frac{\epsilon}{(2-n)}\big)\big(x/\epsilon\big)^{n-1}\d x=+\infty,&n\neq 2;\\
    \int_\epsilon^{\infty}\ln(x/\epsilon)\cdot x\d x=+\infty,&n=2.
    \end{cases}
\end{align*}
Thus Condition \textbf{C-(ii)-$2$} holds.

\noindent \textbf{Example-3}: We can also consider the next example of diffusion process satisfying 
$$
d X_t = \mu X_t \d t + \sigma\sqrt{X_t} \d W_t.
$$
In this case, we have $\mu(x)=\mu x$ and $\sigma(x) =\sigma  \sqrt{x}$. There is no explicit solution to $\varphi_r(x)$ and $H$, hence it is difficult to directly verify the two conditions in Proposition \ref{prop4.2}. We can nevertheless directly calculate that $\phi_1(x)=1 + \sqrt{1+ 2r \frac{1}{x}}$, and it is clear that $\phi_1$ is decreasing in $x$ and increasing in $r$. In addition, $m(y) = \frac{1}{\sigma^2y} e^{\frac{2\mu(y-\epsilon)}{\sigma^2}} $,
$u(x)= \int_\epsilon^x e^{-\frac{2\mu(y-\epsilon)}{\sigma^2}} \d y$. Consequently, 
\begin{align*}
I_1&=\int_\epsilon^{\infty} u'(x)\Big(\int_\epsilon^ym(y)\d y \Big)\d x=\int_\epsilon^{\infty} e^{-\frac{2\mu(x-\epsilon)}{\sigma^2}}  \Big(\int_\epsilon^x \frac{1}{y} e^{\frac{2\mu(y-\epsilon)}{\sigma^2}} \d y \Big) \d x
\\
&=\int_\epsilon^{\infty} \Big(\int_\epsilon^x \frac{1}{y} e^{\frac{2\mu(y-x)}{\sigma^2}} \d y  \Big)\d x=\int_\epsilon^{\infty} \frac{1}{y} \Big(\int_y^\infty  e^{\frac{2\mu(y-x)}{\sigma^2}} \d x\Big)  \d y=\int_\epsilon^{\infty} \frac{\sigma^2}{2\mu y} \d y=+\infty,
\end{align*}
and
\begin{align*}
I_2&=\int_\epsilon^{\infty} u(x)m(x) \d x=\int_\epsilon^{\infty} \Big(\int_\epsilon^x e^{-\frac{2\mu(y-\epsilon)}{\sigma^2}} \d y\Big) \frac{1}{x} e^{\frac{2\mu(x-\epsilon)}{\sigma^2}} \d x
\\
&=\int_\epsilon^{\infty} \Big(\int_\epsilon^x e^{-\frac{2\mu(y-x)}{\sigma^2}} \d y\Big) \frac{1}{x} \d x\geq\int_\epsilon^{\infty} \frac{x-\epsilon}{x} \d x =+\infty,
\end{align*}
which leads to Condition \textbf{C-(ii)-$2$.}}


\vspace{0.1in}

\noindent\textbf{Discussion on condition  \textbf{C-(iii)}}

{   Condition \textbf{C-(iii)} on the attitude function $\phi(x)$ is imposed to guarantee that the smallest equilibrium stopping policy is the optimal equilibrium. Here are some typical selections for the attitude functions that meet the specified condition  \textbf{C-(iii)}.} If $\phi(x)=\ln x$ or $\phi(x)=\frac{x^p}{p}$ with $0 < p\leq 1$, the condition \textbf{C-(iii)} holds trivially. If $\phi(x)$ is a strictly increasing {  differentiable} convex function, the condition \textbf{C-(iii)} also holds. However, the function $\phi(x)=-e^{-px}$ with $p>0$ does not satisfy the condition \textbf{C-(iii)}.

\subsection{Main results}

We next introduce an optimality criterion for an equilibrium in a general setting. For any $R\in\mathcal{E}$, let us define
\begin{align*}
V(x,R):=\phi(g(x))\vee J(x,R),\quad x\geq 0. 
\end{align*}

\begin{definition}\label{defoptV} $R^*\in\mathcal{E}$ is called an optimal equilibrium, if for any $R\in\mathcal{E}$, we have
\begin{align*}
V(x,R^*)\geq V(x,R),\quad x\geq 0.
\end{align*}
\end{definition}

The following result characterizes all one-barrier equilibria and confirms that the smallest equilibrium is the optimal equilibrium.

\begin{theorem}\label{small}
Assume that conditions \textbf{C-(i)} to \textbf{C-(iii)} hold. There exists some $a^*\geq 0$ such that the stopping region $[0,a]$ is a time-consistent equilibrium if and only if $a\geq a^*$. In addition, the equilibrium $[0,a^*]$ is the optimal equilibrium.    
\end{theorem}

\begin{proof}
For a given $0\leq a\leq K$ and any $x\geq a$, let us define
\begin{align}\label{Lambda_eq}
\Lambda(x,a):=J(x,(0,a])&=\int_0^\infty \phi\left(\mathbb{E}[e^{-r\tau_a^x} (K-X^x_{\tau_a^x})^+]\right)dF_\rho(r)\\
&=\int_0^\infty \phi\left((K-a)\mathbb{E}[e^{-r\tau_a^x}] \right)dF_\rho(r).
\end{align} 
Note that $\E[e^{-r\tau_a^x}]$ is decreasing in $x$, $x\mapsto\phi'(x)x>0$ is increasing and $\frac{\partial}{\partial x}H(r,a,x)=\frac{\varphi_r(x)'}{\varphi_r(x)}<0$ is increasing in $x$ from Remark \ref{cond-process} after condition \textbf{C-(ii)}, therefore 
\begin{align*}
x\mapsto \phi'\left( (K-a)\mathbb{E}[e^{-r\tau_a^x}] \right) (K-a) \mathbb{E}[e^{-r\tau_a^x}]\frac{\partial}{\partial x}H(r,a,x)\text{ is increasing}.
\end{align*}
It follows that the right partial derivative of $\Lambda$ with respect to $x$ at $x=a$ is given by (note that $\Lambda(x,a)$ is not defined for $x < a$)
\begin{align}
\lim_{x\downarrow a}\frac{\partial\Lambda(x,a)}{\partial x}&=\lim_{x\downarrow a}\int_0^\infty\phi'\left( (K-a)\mathbb{E}[e^{-r\tau_a^x}] \right) (K-a) \mathbb{E}[e^{-r\tau_a^x}]\frac{\partial}{\partial x}H(r,a,x)dF_\rho(r)\notag\\
&=\int_0^\infty \phi'\left( (K-a) \right) (K-a) \lim_{x\downarrow a} \frac{\partial}{\partial x}H(r,a,x)dF_\rho(r).\label{com-ineq}
\end{align}
where the final equality holds by Monotone Convergence Theorem. Let us denote
\begin{align}\label{G-a}
G(a):=\int_0^\infty \lim_{x\downarrow a}\frac{\partial }{\partial x}H(r,a,x)dF_\rho(r).
\end{align}
By the condition \textbf{C-(ii)}-(a), $G(a)$ is increasing in $a$. In addition, from Remark \ref{cond-process}, we have that
\begin{align*}
G: a \mapsto \int_0^{\infty} \frac{\varphi_r' (a)}{\varphi_r(a)} d F_{\rho}(r)
\end{align*}
is clearly continuous. Therefore, it holds that 
\begin{align*}
G(a)-\frac{(K-a)'}{K-a}=G(a)+\frac{1}{K-a}
\end{align*}
is strictly increasing in $a\in {  [0,K)}$. Moreover, $\lim_{a\rightarrow K} G(a)+\frac{1}{K-a}>0$. 

Let us define $a^*=0$ if $G(a)+\frac{1}{K-a}>0$ for all $a\in {  [0,K)}$ and define $a^*=\hat{a}$ if there exists a unique root $\hat{a}\in [0,K)$ such that $G(\hat{a})+\frac{1}{K-\hat{a}}=0$. It then follows that 
\begin{align}\label{ineq-g}
G(a)>-\frac{1}{K-a}
\end{align}
if and only if $a>a^*$. Then, by \eqref{com-ineq}, we have that 
\begin{align*}
\lim_{x\downarrow a}\frac{\partial\Lambda(x,a)}{\partial x}\geq \lim_{x\downarrow a}(\phi(K-x))'
\end{align*} 
if and only if $a\geq a^*$. As a result, any region $[0,a]$ for $a<a^*$ cannot be an equilibrium.

Now, for any $a>a^*$ and any $x\geq a$, it holds that
\begin{align*}
\frac{\partial \Lambda(x,a)}{\partial a}=\int_0^\infty \phi'\left( (K-a)\mathbb{E}[e^{-r\tau_a^x}] \right) (K-a) \mathbb{E}[e^{-r\tau_a^x}] \left(-\frac{1}{K-a}+\frac{\partial }{\partial a}H(r,a,x)\right)dF_\rho(r).
\end{align*}
By conditions \textbf{C-(ii)}-(b) and \textbf{C-(iii)}, we can derive that  
\begin{align*}
    \phi'\left( (K-a)\mathbb{E}[e^{-r\tau_a^x}] \right) (K-a) \mathbb{E}[e^{-r\tau_a^x}] \text { is a decreasing function with respect to } r,
\end{align*}
and 
\begin{align*}
    \left(-\frac{1}{K-a}+\frac{\partial }{\partial a}H(r,a,x)\right) \text { is an increasing function with respect to } r.
\end{align*}
Therefore, we have 
\begin{align*}
\frac{\partial \Lambda(x,a)}{\partial a}&\leq \left(\int_0^\infty \phi'\left( (K-a)\mathbb{E}[e^{-r\tau_a^x}] \right) (K-a) \mathbb{E}[e^{-r\tau_a^x}] dF_\rho(r)\right)\\
&\times\left(-\frac{1}{K-a}+\int_0^\infty\frac{\partial}{\partial a}H(r,a,x)dF_\rho(r)\right)\\
&< \left(\int_0^\infty \phi'\left( (K-a)\mathbb{E}[e^{-r\tau_a^x}] \right) (K-a) \mathbb{E}[e^{-r\tau_a^x}] dF_\rho(r)\right)
\\
&\qquad\times\left(\int_0^\infty\lim_{x\downarrow a} \frac{\partial}{\partial x}H(r,a,x)dF_\rho(r)+\int_0^\infty\frac{\partial}{\partial a}H(r,a,x)dF_\rho(r) \right).
\end{align*}
where the first inequality comes from Harris-FKG inequality and the second inequality follows from \eqref{ineq-g} for $a>a^*$.

Recall from Remark \ref{cond-process} that $\E [e^{-r \tau_a^x}]=\frac{\varphi_{r}(x)}{\varphi_{r}(a)}$. It is clear to see that $\frac{\partial}{\partial a} H(r, a, x)$ is independent of $x$. Moreover, as $H(r,a,x)=\ln (\varphi_{r}(x))-\ln(\varphi_{r}(a))$, it is trivial to see that $\lim_{x\downarrow a}\frac{\partial}{\partial x}H(r,a,x)=-\lim_{x\downarrow a}\frac{\partial}{\partial a} H(r,a,x)$. We can then deduce that
\begin{align*}
\frac{\partial \Lambda(x,a)}{\partial a}&< \left(\int_0^\infty\phi'\left( g(a)\mathbb{E}[e^{-r\tau_a^x}] \right) g(a) \mathbb{E}[e^{-r\tau_a^x}] dF_\rho(r)\right)
\\
&\qquad \times\left(-\int_0^\infty\lim_{x\downarrow a}\frac{\partial}{\partial a}H(r,a,x)dF_\rho(r)+\int_0^\infty\frac{\partial}{\partial a}H(r,a,x)dF_\rho(r) \right)\\
&= 0.
\end{align*}
As a result, we conclude that, for any $a>a^*$ and any $x>a$,  
\begin{align}
\phi(g(x))=\Lambda(x,x)<\Lambda(x,a)< \Lambda(x,a^*).
\end{align}
The two inequalities then yield that for any $a\geq a^*$, $[0,a]$ is an equilibrium, and the second inequality verifies that the smallest equilibrium $[0,a^*]$ is indeed the optimal equilibrium.
\end{proof}
{ Let us revisit the previous two examples (Geometric Brownian motion and Bessel process) in the discussion on condition \textbf{C-(ii)} in Section \ref{sec4.1}. In both cases, we can derive the optimal equilibrium explicitly.}

\noindent\textbf{Example-1} (Continued) If $X_t$ is a geometric Brownian motion $dX_t=\mu X_tdt+\sigma X_tdW_t$ with $\mu>0$ and $\sigma>0$ and in addition, we assume that the attitude function $\phi(x)$ satisfies the condition \textbf{C-(iii)}. { Let us define
\begin{equation}\label{eq:frho}
	f(r) := \sqrt{\Big(\frac{\mu}{\sigma^2}-\frac12\Big)^2+\frac{2r}{\sigma^2}}+\frac{\mu}{\sigma^2}-\frac12 {  \geq } 0.
\end{equation}
By formula (2.0.1) on page 628 of \cite{BorodinSalminen}, we have that $\varphi_{r}(x)=\left( \frac{1}{x} \right)^{f(r)}$ and $H(r,a,x)=\ln(\mathbb{E}[e^{-r \tau_a^x}])=f(r)(\ln a-\ln x)$. }
From the proof of Theorem \ref{small}, we get that $G(a)=\int_0^\infty \lim_{x\downarrow a}\frac{\partial }{\partial x}H(r,a,x)dF_\rho(r)=-\frac{\int_0^\infty f(r)dF_\rho(r)}{a}$, and the unique solution to the equation 
$$G(a)+\frac{1}{K-a}=\frac{1}{K-a}-\frac{\int_0^\infty f(r)dF_\rho(r)}{a}=0,$$ 
is 
\begin{align}\label{opt-ex-1}
a^*=\frac{\int_0^\infty f(r)dF_\rho(r)}{1+ \int_0^\infty f(r)dF_\rho(r)}K\in (0,K).
\end{align}
The optimal equilibrium in this model is explicitly characterized by $[0,a^*]$ according to Theorem \ref{small} where $a^*$ is given in \eqref{opt-ex-1}.
\ \\

\noindent
\textbf{Example-2} (Continued) If $X_t$ is a general Bessel process with degree $\nu$ and $n=2 \nu+2$ and in addition, the attitude function $\phi(x)$ is assumed to satisfy the condition \textbf{C-(iii)}.
{ By the formula (2.0.1) on page $404$ of \cite{BorodinSalminen}, we have that $\varphi_{r}(x)=x^{-\nu} K_{\nu} (x \sqrt{2 r})$ and 
$$
H(r, a, x)= -\nu \ln x + \ln K_{\nu} (x \sqrt{2 r}) +  \nu \ln a - \ln K_{\nu} (a \sqrt{2 r}),
$$
where $K_{\nu}$ is the modified Bessel function of the second kind with degree $\nu$.}
We have that 
\begin{align*}
G(a)=\int_0^\infty \lim_{x\downarrow a}\frac{\partial }{\partial x}H(r,a,x)dF_\rho(r)=-\int_0^\infty\frac{K_{\nu+1} (a \sqrt{2 r})}{K_{\nu} (a \sqrt{2 r})} \sqrt{2r}\,dF_\rho(r).
\end{align*}
Therefore, it holds that
\begin{align*}
\lim_{a\rightarrow 0}\left(-\int_0^\infty\frac{K_{\nu+1} (a \sqrt{2 r})}{K_{\nu} (a \sqrt{2 r})} \sqrt{2r}\,dF_\rho(r)+\frac{1}{K-a}\right)=-\infty,
\\
\lim_{a\rightarrow K}\left(-\int_0^\infty\frac{K_{\nu+1} (a \sqrt{2 r})}{K_{\nu} (a \sqrt{2 r})} \sqrt{2r}\,dF_\rho(r)+\frac{1}{K-a}\right)=\infty.
\end{align*}
As $G(a)+\frac{1}{K-a}$ is strictly increasing, there exists a unique solution $a^*$ to the equation $G(a)+\frac{1}{K-a}=0$. By Theorem \ref{small}, the optimal equilibrium is then explicitly characterized by $[0,a^*]$.

In particular, for the case $n=3$ that $X_t=\sqrt{(W_t^1)^2+(W_t^2)^2+(W_t^3)^2}$ where $(W^1,W^2,W^3)$ is a three-dimensional Brownian motion, we can compute $G(a)=-\int_0^\infty\sqrt{2r}\,dF_\rho(r)-\frac{1}{a}$. Therefore, the unique solution to the equation 
$$G(a)+\frac{1}{K-a}=-\int_0^\infty\sqrt{2r}\,dF_\rho(r)-\frac{1}{a}+\frac{1}{K-a}=0$$
admits the explicit form that
\begin{align}\label{opt-ex-2}
a^*=\frac{\int_0^\infty\sqrt{2r}\,dF_\rho(r)K-2+\sqrt{4+\left(\int_0^\infty\sqrt{2r}\,dF_\rho(r)\right)^2 K^2}}{2\int_0^\infty\sqrt{2r}\,dF_\rho(r)}\in (0,K).
\end{align}

\begin{remark}\label{rem-equiv}
In both examples when $X_t$ is a geometric Brownian motion with $\mu>0$ and $X_t$ is a general Bessel process, when the attitude function $\phi(x)$ satisfies the condition \textbf{C-(iii)}, the thresholds $a^*$ of the optimal equilibrium in \eqref{opt-ex-1} and \eqref{opt-ex-2} can be explicitly characterized that highly depend on the distribution of random discount rate $\rho$, i.e., the level of diversity in disagreed discount rates within the group. On the other hand, it is also interesting to observe that $a^*$ in both \eqref{opt-ex-1} and \eqref{opt-ex-2} under different attitude function $\phi(x)$ satisfying the condition \textbf{C-(iii)} always coincide with the one under the linear aggregation attitude function $\phi(x)=x$. In fact, in the proof of Theorem \ref{small}, the definition of $a^*$ in the optimal equilibrium is also invariant with respect to the choice of attitude function $\phi(x)$ as the operator $G$ in \eqref{G-a} is independent of $\phi(x)$. 

In our context, the linear attitude function $\phi(x)=x$ indicates that the central planner aggregates different discount rates with the same weight and has neutral attitude towards different levels of discount rate in the group. Therefore, from our result in Theorem \ref{small}, as along as the attitude function $\phi(x)$ satisfies the condition \textbf{C-(iii)} (for example, $\phi(x)=\frac{x^p}{p}$ for $p<1$ and $p\neq 0$ or $\phi(x)=\ln x$), the derived optimal equilibrium is equivalent to the one under neutral attitude function $\phi(x)=x$.

\end{remark}

\subsection{Connection between the optimal equilibrium and weak equilibrium}

\cite{bayraktar21, bayraktar22} studied different concepts of time-consistent equilibrium and particularly discussed the relationship between the optimal mild equilibrium and the weak equilibrium, where they used ``mild equilibrium" to refer for our definition of the equilibrium in Definition \ref{def:equilibrium}.  In this subsection, we aim to establish the connection in our context under the aggregation preference. First, let us recall the definition of weak equilibrium firstly proposed in \cite{ChristensenL18}, and further studied in \cite{ChristensenL20}, \cite{bayraktar22}.

\begin{definition} \label{def:weak}
A closed set $R \in [0, \infty)$ is said to be a weak equilibrium, if
\begin{equation} \label{eq:weak}
\left\{
\begin{aligned}
 &J(x,R) \geq \phi(g(x)),  & \mbox{ if } & x \notin R,  \\
&\liminf_{\epsilon \searrow 0} \frac{\phi(g(x)) -  \int_0^\infty \phi\left(\E^x\left[e^{-r \tau^{\epsilon}_R}g(X_{\tau^{\epsilon}_R})\right]\right)dF_\rho(r)}{\epsilon} \geq 0, & \mbox{ if } & x \in R.
\end{aligned}
 \right.
\end{equation}
where on the above $\tau^{\epsilon}_R:= \inf \{t \geq \epsilon: X_t \in R \}$, and $\E^{x} \left[ \cdot \right] : =\E \left[ \cdot | X_0 =x \right]$. {  In what follows, we shall also denote $\tau_R:= \inf \{t > 0: X_t \in R \}$. }
\end{definition}

By definition, any weak equilibrium is a mild equilibrium. We next justify in the following proposition that under extra conditions on the aggregation attitude function $\phi(x)$, the optimal mild equilibrium is also a weak equilibrium, which is consistent with the result in \cite{bayraktar22}.

\begin{proposition}\label{prop-weak}
Assume that conditions \textbf{C-(i)} to \textbf{C-(iii)} hold, and in addition assume that $\phi''(x) \leq 0$ for all $x \geq 0$. Then the optimal mild equilibrium $[0,a^*]$ is also a weak equilibrium.
\end{proposition}
\begin{proof}
We verify directly Definition \ref{def:weak}, and it is enough to verify the second inequality in \eqref{eq:weak}. When $a^*=0$, the condition $x \in R$ becomes $x=0$, and 
$$
 \int_0^\infty \phi\left(\E^x\left[e^{-r \tau^{\epsilon}_R}g(X_{\tau^{\epsilon}_R})\right]\right)dF_\rho(r) \leq \phi(g(0)),
$$
as $g$ and $\phi$ are respectively decreasing and increasing, and $X$ is a non-negative process. Hence \eqref{eq:weak} is automatically valid. In the following, we shall only consider the case $a^*>0$, where $a^*$ is characterized as the unique solution of equation
$$
\int_0^\infty  \left( g(a)\frac{\varphi_r'(a)}{\varphi_r(a)} + 1 \right)  dF_\rho(r)=0.
$$
For any $t, x \geq 0$, define 
$$
v_r(t,x) :=  e^{-rt} \E^x \left[e^{-r \tau_{[0,a^*]}}g(X_{\tau_{[0,a^*]}})\right].
$$
It is clear that
$$
v_r(t,x) = \left\{
\begin{aligned}
& e^{-rt} g(x) & \mbox{ if } x \leq a^*, \\
& e^{-rt} g(a^*) \frac{\varphi_{r} (x) }{\varphi_r (a^*)} & \mbox{ if } x > a^*.
\end{aligned}
 \right.
$$
By straightforward calculations, we can get
$$
\partial_t v_r(t, x) = \left\{
\begin{aligned}
& -r e^{-rt} g(x) & \mbox{ if } x \leq a^*, \\
& -r e^{-rt} g(a^*) \frac{ \varphi_{r} (x) }{\varphi_r (a^*)} & \mbox{ if } x > a^*,
\end{aligned}
 \right. \ \ \ \  \partial_x v_r(t, x) = \left\{
\begin{aligned}
& e^{-rt} g'(x) = -e^{-rt}  & \mbox{ if } x \leq a^*, \\
& e^{-rt} g(a^*) \frac{ \varphi_{r}' (x) }{\varphi_r (a^*)} & \mbox{ if } x > a^*,
\end{aligned}
 \right.
$$
and
$$
\partial^{2}_{xx} v_{r}(t,x) = \left\{
\begin{aligned}
& e^{-rt} g''(x) = 0 & \mbox{ if } x \leq a^*, \\
& e^{-rt} g(a^*) \frac{ \varphi_{r}'' (x) }{\varphi_r (a^*)} & \mbox{ if } x > a^*.
\end{aligned}
 \right.
$$
Noticing that 
$$
\E^x \left[ v_r(\epsilon, X_{\epsilon}) \right] = \E^x \left[ e^{-r \epsilon} \E^{X_{\epsilon}} \left[e^{-r \tau_{[0,a^*]}}g(X_{\tau_{[0,a^*]}})\right] \right] =  \E^x \left[ e^{-r \tau^{\epsilon}_{[0,a^*]}} g(X_{\tau^{\epsilon}_{[0,a^*]}}) \right],
$$
and by the second inequality of \eqref{eq:weak}, it is sufficient to prove
$$
\liminf_{\epsilon \searrow 0} \frac{\phi(g(x)) -  \int_0^\infty \phi \left( \E^x \left[ v_r(\epsilon, X_{\epsilon}) \right] \right) dF_\rho(r)}{\epsilon} \geq 0.
$$
{ Using the It\^o's formula involving the local time integral as in Lemma 2.15 of \cite{bayraktar22} (see also \cite{Pesk07}) with $x_0=a^*$}, we obtain that
$$
\begin{aligned}
& { \E^x} \left[ v_r(\epsilon \wedge \tau_{\partial B(x,h)}, X_{\epsilon \wedge \tau_{\partial B(x,h)}}) - v_r(0,x) \right]  \\
= \ & \E^{x} \left[ \int_0^{\epsilon \wedge \tau_{\partial B(x,h)}} \frac12 ( \cL v_r(s, X_s-) + \cL v_r(s, X_s+) ) d s \right] \\
~~~~~~~~~~~~~~ & +  \E^{x} \left[ \frac12 \int_0^{\epsilon \wedge \tau_{\partial B(x,h)}} (\partial_x v_r(s, a^*+) - \partial_x v_r(s, a^*-)) d L_s^{a^*} \right],
\end{aligned}
$$
where on the above, $B(x,h)$ is the set of points within distance $h$ from $x$, {  $\partial B(x,h)$ is its boundary,} $L_s^x$ is the local time at point $x$ up to time $t$, and $\cL$ is the operator defined as
$$
\cL v(t,x) := \partial_t v(t,x) + \mu(x) \partial_x v(t,x) + \frac12 \sigma^2(x) \partial^2_{xx} v(t,x),
$$
for any function $v\in \mathcal{C}^{1,2}([0,\infty)^2)$.
Using the previous expressions, we get for $x > a^*$,
$$
\cL v_r(t, x) = e^{-rt} \frac{g(a^*)}{\varphi_r (a^*)} \left( -r \varphi_{r} (x) + \mu(x) \varphi_{r}' (x) + \frac12 \sigma^2(x) \varphi_{r}'' (x) \right) = 0,
$$
and for $x \leq a^*$, 
$$
\cL v_r(t, x) = e^{-rt} (-r g(x) - \mu(x)) \leq 0.
$$
Hence we have for $x \leq a^*$,
$$
\begin{aligned}
& { \E^x} \left[ v_r(\epsilon \wedge \tau_{\partial B(x,h)}, X_{\epsilon \wedge \tau_{\partial B(x,h)}}) - g(x) \right]  \\
\leq \ & \E^{x} \left[ \frac12 \int_0^{\epsilon \wedge \tau_{\partial B(x,h)}} (\partial_x v_r(s, a^*+) - \partial_x v_r(s, a^*-)) d L_s^{a^*} \right] \\
= \ &\frac12 \left( g(a^*) \frac{\varphi_r'(a^*)}{\varphi_r(a^*)} + 1 \right)  \E^{x} \left[  \int_0^{\epsilon \wedge \tau_{\partial B(x,h)}} e^{-r s} d L_s^{a^*}  \right].
\end{aligned}
$$
Now, it holds that 
\begin{equation}\label{equtemp}
\begin{aligned}
& \liminf_{\epsilon \searrow 0} \frac{1}{\epsilon} \left\{\phi(g(x)) -  \int_0^\infty \phi \left( \E^x \left[ v_r(\epsilon, X_{\epsilon}) \right] \right) dF_\rho(r) \right\} \\
= \ &\liminf_{\epsilon \searrow 0} \frac{1}{\epsilon}    \int_0^\infty \left( \phi(g(x)) - \phi \left( \E^x \left[ v_r(\epsilon \wedge \tau_{\partial B(x,h)}, X_{\epsilon \wedge \tau_{\partial B(x,h)}}) \right] \right) \right)dF_\rho(r)    \\
 \geq \ & \liminf_{\epsilon \searrow 0} \frac{1}{\epsilon}    \int_0^\infty \left( \phi(g(x)) - \phi \left( \frac12 \left( g(a^*) \frac{\varphi_r'(a^*)}{\varphi_r(a^*)} + 1 \right)  \E^{x} \left[  \int_0^{\epsilon \wedge \tau_{\partial B(x,h)}} e^{-r s} d L_s^{a^*}  \right] + g(x) \right) \right)dF_\rho(r),
\end{aligned}
\end{equation}
where the first equation holds because $\phi$ is locally Lipschitz around $g(x)$ and { the following approximation holds that
$$
\E^x \left[ v_r(\epsilon \wedge \tau_{\partial B(x,h)}, X_{\epsilon \wedge \tau_{\partial B(x,h)}}) \right]=\E^x \left[ v_r(\epsilon, X_{\epsilon}) \right] +o(\epsilon^k)$$ 
for all $k\geq 0$ by following the similar arguments of Lemma 3.8 in \cite{bayraktar22}. The rationale behind this approximation is to apply some localization argument to restrict $X$ within a bounded ball $B(x,h)$.}

If $x<a^*$, define $h:= \frac{a^*-x}{2}$, we have that $ \int_0^{\epsilon \wedge \tau_{\partial B(x,h)}} e^{-rs} d L_s^{a^*} =0$, $\P^x$-a.s. Hence
$$
 \liminf_{\epsilon \searrow 0} \frac{1}{\epsilon} \left\{\phi(g(x)) -  \int_0^\infty \phi \left( \E^x \left[ v_r(\epsilon, X_{\epsilon}) \right] \right) dF_\rho(r) \right\} =0.
$$
If $x=a^*$, applying It\^o's formula on $[0,\epsilon \wedge \tau_{\partial B(x,h)}]$ and function $h(t,y) := e^{-rt} |y-x|$, we have
$$
\begin{aligned}
& \E^{x} \left[  \int_0^{\epsilon \wedge \tau_{\partial B(x,h)}} e^{-r s} d L_s^{a^*}  \right]  \\
= \ & \E^{x} \left[ e^{-r(\epsilon \wedge \tau_{\partial B(x,h)})} |X_{\epsilon \wedge \tau_{\partial B(x,h)}}-x | \right]  \\
~~~~~~~~~~~~~&- \E^x \left[ \int_0^{\epsilon \wedge \tau_{\partial B(x,h)}}  \mbox{sgn}(X_s -x) \mu(X_s) e^{-rs} ds  \right] +  \E^x \left[ \int_0^{\epsilon \wedge \tau_{\partial B(x,h)}} r e^{-rs} |X_s-x| ds  \right]  \\
=\ & \underbrace{ \E^{x} \left[ |X_{\epsilon \wedge \tau_{\partial B(x,h)}}-x | \right]  - \E^x \left[ \int_0^{\epsilon \wedge \tau_{\partial B(x,h)}}  \mbox{sgn}(X_s -x) \mu(X_s) ds  \right]}_{I_1}
\\
& + \underbrace{ \E^{x} \left[ (e^{-r(\epsilon \wedge \tau_{\partial B(x,h)})}-1) |X_{\epsilon \wedge \tau_{\partial B(x,h)}}-x | \right] }_{I_2}  - \underbrace{\E^x \left[ \int_0^{\epsilon \wedge \tau_{\partial B(x,h)}}  \mbox{sgn}(X_s -x) \mu(X_s) (e^{-rs}-1) ds  \right]}_{I_3} 
\\
&+ \underbrace{  \E^x \left[ \int_0^{\epsilon \wedge \tau_{\partial B(x,h)}} r e^{-rs} |X_s-x| ds  \right] }_{I_4}.
\end{aligned}
$$

Following the arguments in the proof of Lemma 3.9 of \cite{bayraktar22}, we can obtain that $I_2, I_3, I_4$ are all of $o(\epsilon)$, uniformly in $r$. Therefore, it holds that 
$$
\begin{aligned}
& \liminf_{\epsilon \searrow 0} \frac{1}{\epsilon}    \int_0^\infty \left( \phi(g(a^*)) - \phi \left( \frac12 \left( g(a^*) \frac{\varphi_r'(a^*)}{\varphi_r(a^*)} + 1 \right)  \E^{x} \left[  \int_0^{\epsilon \wedge \tau_{\partial B(x,h)}} e^{-r s} d L_s^{a^*}  \right] + g(a^*) \right) \right)dF_\rho(r) \\
& \ =  \liminf_{\epsilon \searrow 0} \frac{1}{\epsilon}    \int_0^\infty \left( \phi(g(a^*)) - \phi \left( \frac12 \left( g(a^*) \frac{\varphi_r'(a^*)}{\varphi_r(a^*)} + 1 \right) I_1 + g(a^*) \right) \right)dF_\rho(r) \\
& \ \geq  \liminf_{\epsilon \searrow 0} \frac{1}{\epsilon}   \int_0^\infty   \phi'(g(a^*))\cdot  \left(-\frac12 \left( g(a^*) \frac{\varphi_r'(a^*)}{\varphi_r(a^*)} + 1 \right) I_1\right)  dF_\rho(r)
\\
&\ =\liminf_{\epsilon \searrow 0} -\frac{1}{2\epsilon}  \phi'(g(a^*))I_1\cdot\int_0^\infty  \left( g(a^*)\frac{\varphi_r'(a^*)}{\varphi_r(a^*)} + 1 \right)  dF_\rho(r)=0,
\end{aligned}
$$
where the first line holds using the similar arguments in \eqref{equtemp}, the third line is due to the concavity of $\phi$ and the last line holds by the definition of $a^*$. Therefore, we verify the second inequality of \eqref{eq:weak} for all $x\in [0,a^*]$ and thus complete the proof of the proposition.
\end{proof}


\section{Impacts on the Optimal Equilibrium by Attitude Function and Diversity Distribution}\label{sec:impact}

Generally speaking, if the aggregation attitude function $\phi(x)$ does not satisfy the assumption \textbf{C-(iii)} in the previous section, the smallest equilibrium may not necessarily be the optimal equilibrium. That is, the characterization of the optimal equilibrium may sensitively depends on the choice of the attitude function $\phi(x)$ as well as the diversity distribution of $\rho$. In particular, we will consider a simple form of $\phi(x)=\min(x,\alpha)$ for some $\alpha\in(0,1)$ in this section, which clearly does not fulfill the assumption \textbf{C-(iii)}. With this choice of $\phi(x)=\min(x,\alpha)$, we will present several examples to explicitly illustrate how $\phi(x)$ and the diversity distribution of $\rho$ may affect the optimal equilibrium. New to the literature, we will show some concrete examples in which the optimal equilibrium may no longer coincide with the smallest equilibrium or may not exist.

To ease the notation in following examples, let us now extend the definition of $\Lambda(x,a)$ to all $x\in\mathbb{R}$ such that $\Lambda(x,a)=\phi(g(x))$ for $x<a$. As a result, to characterize the optimal equilibrium attaining the optimal $V(x,R)$ in Definition \ref{defoptV}, it is equivalent to find the global maximum of $\Lambda(x,a)$ over $a>0$ uniformly for all $x$.

\subsection{Geometric Brownian motion}
In this subsection, let us consider the underlying process $X_t$ as a geometric Brownian motion that satisfies
\begin{equation*}
    dX_t=\mu X_tdt+\sigma X_tdW_t,
\end{equation*}
with $\mu>0$ and $\sigma>0$. We still consider the Put payoff function $g(x)=(1-x)^+$ with the strike price $K=1$. Recall the definition of $f(r)$ in \eqref{eq:frho}.

\noindent\textbf{Characterization of the equilibrium.} 
First,  although the attitude function $\phi(x)$ does not satisfy the assumption  \textbf{C-(iii)},  we will show that under some other conditions, there still exists an $a^*$ such that $R=[0,a]$ is an equilibrium if and only if $a\geq a^*$. Let $\rho^*:=\sup\{r:r\in\text{supp}(\rho)\}$.

\begin{proposition} \label{prop:eq}
If the underlying process $X_t$ is a geometric Brownian motion, the stopping policy $R$ is an equilibrium if and only if $R=[0,a]$ for some $a\geq a^*$ where $a^*$ is defined by
\begin{equation*}
    a^*=\begin{cases}
        \frac{\int_0^\infty f(r)\,dF_\rho(r)}{\int_0^\infty f(r)\,dF_\rho(r)+1},&\text{when}\ 1-\alpha\leq\frac{\int_0^\infty f(r)\,dF_\rho(r)}{\int_0^\infty f(r)\,dF_\rho(r)+1},
        \\
         1-\alpha,&\text{when}\ \frac{\int_0^\infty f(r)\,dF_\rho(r)}{\int_0^\infty f(r)\,dF_\rho(r)+1} < 1-\alpha\leq\frac{f(\rho^*)}{f(\rho^*)+1},
        \\
        \gamma,&\text{when}\ 1-\alpha>\frac{f(\rho^*)}{f(\rho^*)+1},
    \end{cases}
\end{equation*}
and $\gamma$ is the smaller root in $[0,1]$ of the equation
\begin{equation*}
    (1-a)a^{f(\rho^*)} = \alpha(1-\alpha)^{f(\rho^*)}.
\end{equation*}
\end{proposition}

\begin{proof}
Adopting the same notation in \eqref{Lambda_eq}, we have that
\begin{equation*}
    \Lambda(x,a) = \int_0^\infty\min\Big\{(1-a)^+\Big(\frac a x\Big)^{f(r)},\alpha\Big\}dF_\rho(r)=\int_0^\infty\min\Big\{(1-a)\Big(\frac a x\Big)^{f(r)},\alpha\Big\}dF_\rho(r)
\end{equation*}
for all $0\leq a\leq1$, and $x\geq a$. We have the following two distinct cases.

\noindent\textbf{Case 1}: $1-\alpha\leq\frac{f(\rho^*)}{f(\rho^*)+1}$.

By definition, $R=[0,a]$ is an equilibrium if and only if $\Lambda(x,a)\geq\phi(g(x))=\min\{(1-x)^+,\alpha\}$ for all $x\geq a$. When $a<1-\alpha$, we know that there exists $\epsilon>0$ such that when $f(r)\geq f(\rho^*)-\epsilon$,
\begin{equation*}
    (1-a)a^{f(r)}\leq \alpha(1-\alpha)^{f(r)}.
\end{equation*}
It therefore follows that
\begin{align*}
    \Lambda(1-\alpha,a)&=\int_0^\infty\min\Big\{(1-a)\Big(\frac a {1-\alpha}\Big)^{f(r)},\alpha\Big\}\mathbbm{1}\{f(r)\geq f(\rho^*)-\epsilon\}dF_\rho(r)
    \\
    &\qquad+\int_0^\infty \min\Big\{(1-a)\Big(\frac a {1-\alpha}\Big)^{f(r)},\alpha\Big\}\mathbbm{1}\{f(r)<f(\rho^*)-\epsilon\}dF_\rho(r)
    \\
    &\leq (1-a)\Big(\frac a {1-\alpha}\Big)^{f(\rho^*)-\epsilon}\times\P(f(\rho)\geq f(\rho^*)-\epsilon)+\alpha \P(f(\rho)<f(\rho^*)-\epsilon)
    \\
    &<\alpha\P(f(\rho)\geq f(\rho^*)-\epsilon)+\alpha \P(f(\rho)<f(\rho^*)-\epsilon)\\
    &=\alpha=\phi(g(1-\alpha)),
\end{align*}
where the last inequality holds because $\P(f(\rho)\geq f(\rho^*)-\epsilon)>0$ for all $\epsilon>0$. As a result, $[0,a]$ is not an equilibrium when $a<1-\alpha$. 

On the other hand, if $a\geq 1-\alpha$, we have $x\geq a\geq 1-\alpha$ and 
\begin{equation*}
    \Lambda(x,a) = \int_0^\infty\frac{(1-a)a^{f(r)}}{x^{f(r)}}dF_\rho(r)\text{ and }\phi(g(x)) = (1-x)^+.
\end{equation*}
Using the same calculation in the proof of Theorem \ref{small}, we know that if $a\geq \frac{\int_0^\infty f(r)\,dF_\rho(r)}{\int_0^\infty f(r)\,dF_\rho(r)+1}$ and $x\geq 1$, $\Lambda(x,a)\geq 0 = \phi(g(x))$. Moreover, if $a\geq \frac{\int_0^\infty f(r)\,dF_\rho(r)}{\int_0^\infty f(r)\,dF_\rho(r)+1}$ and $x\leq 1$, then $\int_0^\infty\frac{(1-a)a^{f(r)}}{x^{f(r)}}dF_\rho(r)$ is non-increasing with respect to $a$ and 
\begin{equation*}
    \Lambda(x,a)\geq \Lambda (x,x) = (1-x) = \phi(g(x)).
\end{equation*}
Therefore $[0,a]$ is an equilibrium if $a\geq \max\big\{1-\alpha,\frac{\int_0^\infty f(r)\,dF_\rho(r)}{\int_0^\infty f(r)\,dF_\rho(r)+1}\big\}$. Meanwhile, if $a<\frac{\int_0^\infty f(r)\,dF_\rho(r)}{\int_0^\infty f(r)\,dF_\rho(r)+1}$, we can calculate
\begin{equation*}
   \lim_{x\downarrow a} \frac{\partial }{\partial x}\Lambda(x,a)= \lim_{x\downarrow a}\int_0^\infty - \frac{f(r)(1-a)a^{f(r)}}{x^{f(r)+1}}d F_\rho(r)=\frac{(1-a)}{a}\cdot\int_0^\infty - f(r)dF_\rho(r)<-1=\frac{\d}{\d x}\phi(g(x)).
\end{equation*}
Together with the fact that $\Lambda(a,a)=\phi(g(a))$, there exists some $x$ close to $a$ such that $\Lambda(x,a)<\phi(g(x))$. Hence, $[0,a]$ is not an equilibrium. In summary, when $1-\alpha\leq\frac{f(\rho^*)}{f(\rho^*)+1}$, we have
\begin{equation*}
    a^*=\begin{cases}
        \frac{\int_0^\infty f(r)\,dF_\rho(r)}{\int_0^\infty f(r)\,dF_\rho(r)+1},&\text{ when }1-\alpha\leq\frac{\int_0^\infty f(r)\,dF_\rho(r)}{\int_0^\infty f(r)\,dF_\rho(r)+1},
        \\
         1-\alpha,&\text{ when }\frac{\int_0^\infty f(r)\,dF_\rho(r)}{\int_0^\infty f(r)\,dF_\rho(r)+1} < 1-\alpha\leq\frac{f(\rho^*)}{f(\rho^*)+1}.
    \end{cases}
\end{equation*}

\noindent\textbf{Case 2}:  $1-\alpha>\frac{f(\rho^*)}{f(\rho^*)+1}$.

Let us first verify that $\gamma$ is well defined. Consider the function $h(a):=(1-a)a^{f(\rho^*)}$, it holds that $h'(a)=(f(\rho^*)-(f(\rho^*)+1)a)a^{f(\rho^*)-1}$. Therefore, $h(a)$ is increasing when $0\leq a<\frac{f(\rho^*)}{1+f(\rho^*)}$ and decreasing when $\frac{f(\rho^*)}{f(\rho^*)+1}<a\leq 1$ and $h(0)=h(1)=0$. It is not hard to see that $1-\alpha$ is a root of the equation
\begin{equation*}
     (1-\gamma)\gamma^{f(\rho^*)} = \alpha(1-\alpha)^{f(\rho^*)},
\end{equation*}
and our choice of $\gamma$ satisfies
\begin{equation}\label{ineq:lambda}
    (1-a)a^{f(\rho^*)}\begin{cases}
    < \alpha(1-\alpha)^{f(\rho^*)},&\text{when }0\leq a<\gamma,\text{ or }1-\alpha<a\leq 1,
    \\
    > \alpha(1-\alpha)^{f(\rho^*)},&\text{when }\gamma<a<1-\alpha,
    \\
    =  \alpha(1-\alpha)^{f(\rho^*)},&\text{when }a=\gamma,\text{ or }a=1-\alpha.
    \end{cases}
\end{equation}
If $a<\gamma$, similarly, we have that
\begin{equation*}
    \begin{aligned}
    \Lambda(1-\alpha,a)&\leq (1-a)\Big(\frac a {1-\alpha}\Big)^{f(\rho^*)-\epsilon}\P(f(\rho)\leq f(\rho^*)-\epsilon)+\alpha \P(f(\rho)>f(\rho^*)-\epsilon)
    \\
    &<\alpha\P(f(\rho)\leq f(\rho^*)-\epsilon)+\alpha \P(f(\rho)>f(\rho^*)-\epsilon)=\alpha=\phi(g(1-\alpha)),
    \end{aligned}
\end{equation*}
where the last inequality holds because $(1-a)a^{f(\rho^*)}<\alpha(1-\alpha)^{f(\rho^*)}$ when $a<\gamma$ by \eqref{ineq:lambda} and it follows that $[0,a]$ is not an equilibrium if $a<1-\gamma$. 

If $a\geq\gamma$ and $x\geq1$, it is trivial to see that  $\phi(g(x))=0$ and $\Lambda(x,a)\geq 0 = \phi(g(x))$. For $x\in[a,1-\alpha)$, we have
\begin{equation*}
\frac{(1-a)a^{f(r)}}{x^{f(r)}}\geq \frac{(1-a)a^{f(\rho^*)}}{x^{f(\rho^*)}}>\frac{\alpha(1-\alpha)^{f(\rho^*)}}{(1-\alpha)^{f(\rho^*)}}=\alpha,
\end{equation*}
for any $r$ inside the support of $\rho$ and the last inequality holds because $x<1-\alpha$ and \eqref{ineq:lambda}. Therefore,
\begin{equation*}
\Lambda(x,a)= \int_0^\infty\min\Big\{\frac{(1-a)a^{f(r)}}{x^{f(r)}},\alpha\Big\}dF_\rho(r)\geq \alpha = \min\{1-x,\alpha\}=\phi(g(x)).
\end{equation*}
Finally, if $x\in[1-\alpha,1]$, we have $\phi(g(x))=1-x$ and 
\begin{equation*}
\Lambda(x,a)=  \int_0^\infty\min\Big\{\frac{(1-a)a^{f(r)}}{x^{f(r)}},\alpha\Big\}dF_\rho(r)\geq  \min\Big\{\frac{(1-a)a^{f(\rho^*)}}{x^{f(\rho^*)}},\alpha\Big\}.
\end{equation*}
As $x\geq 1-\alpha$, it suffices to show that 
\begin{equation}\label{ineq1}
(1-a)a^{f(\rho^*)}\geq (1-x)x^{f(\rho^*)}.
\end{equation}
Note that $h(\tilde a):=(1-\tilde a)\tilde a^{f(\rho^*)}$ is decreasing when $\tilde a>1-\alpha$, therefore, if $a>1-\alpha$ then \eqref{ineq1} holds due to $x\geq a$. If $\gamma\leq a<1-\alpha$, by \eqref{ineq:lambda}, we know that $(1-a)a^{f(\rho^*)}\geq\alpha(1-\alpha)^{f(\rho^*)}\geq (1-x)x^{f(\rho^*)}.$ and \eqref{ineq1} still holds, which completes the proof that $[0,a]$ is an equilibrium when $a\geq \gamma$. In summary, if $1-\alpha>\frac{f(\rho^*)}{f(\rho^*)+1}$, we have $a^*=\gamma$ and our conclusion follows.
\end{proof}

Based on the characterization of all equilibria, let us turn to the discussion on the existence of the optimal equilibrium. As new contributions to the literature, we will show in next three examples respectively that under different choices of $\alpha$ and distributions of $\rho$ (more precisely $f(\rho)$), the optimal equilibrium may: i) exist and equal to the smallest equilibrium in which the threshold $a^*$ now depends on $\alpha$, i.e., the optimal equilibrium region depends on the aggregation attitude function $\phi(x)$; ii) exist but not equal to the smallest equilibrium, and the optimal equilibrium threshold $a^{**}$ depends on the attitude function $\phi(x)$; and iii) not exist.

\ \\
\noindent\textbf{Example 1.} An illustrative example to show the impact on the optimal equilibrium by the aggregation attitude, i.e. the value of $\alpha$. Suppose that $1-\alpha\leq \frac{f(\rho^*)}{f(\rho^*)+1}$.

By Proposition \ref{prop:eq}, we get that $a^* = \max\Big\{1-\alpha,\frac{\int_0^\infty f(r)dF_\rho(r)}{\int_0^\infty f(r)dF_\rho(r)+1}\Big\}$ and for any $a\geq a^*$, we have
\begin{equation*}
    \Lambda(x,a)=\int_0^\infty\frac{(1-a)a^{f(r)}}{x^{f(r)}}d F_\rho(r).
\end{equation*}
Following the argument in the proof of Theorem \ref{small}, if $a\geq \frac{\int_0^\infty f(r)dF_\rho(r)}{\int_0^\infty f(r)dF_\rho(r)+1}$, we have that $\Lambda(x,a)$ is non-increasing in $a$ and it holds that
\begin{equation*}
\Lambda(x,a)\leq \Lambda(x,a^*),
\end{equation*}
for any $x\geq a\geq a^*$. Therefore, the smallest equilibrium $[0,a^*]$ is still the optimal equilibrium, which is similar to the main result in Theorem \ref{small}.

\begin{remark}
In this case, we note that the optimal equilibrium boundary $a^*$ is decreasing in $\alpha$. In our choice of $\phi(x)=\min(x,\alpha)$, the value $\alpha$ depicts the range of diverse attitudes that the social planner takes into account in the aggregation. Small $\alpha$ reflects that the social planner focuses more on the side of larger discount rates, i.e., the more impatient group members such as old generations in the total population. Our example illustrates that if $\alpha$ is smaller, the optimal equilibrium stopping region $[0,a^*]$ is larger and it is more likely that the social planner will quit from the investment and redeem the immediate payoff for the benefit of the impatient group members.    

\end{remark}

\noindent\textbf{Example 2.} An illustrative example to show the impact on the optimal equilibrium by the diversity distribution, i.e. the value of $f(\rho^*)$. Suppose that $1-\alpha > \frac{f(\rho^*)}{1+f(\rho^*)}$, $\P(f(\rho)=0)=\P(f(\rho)=f(\rho^*))=\frac12$.

By Proposition \ref{prop:eq}, we get that $a^*=\gamma$ where $\gamma$ is the smallest root in $[0,1]$ for the equation
\begin{equation*}
    (1-a)a^{f(\rho^*)}=\alpha\times(1-\alpha)^{f(\rho^*)}.
\end{equation*}
Then for all $x\geq a\geq \gamma$, we have that
\begin{equation*}
    \Lambda(x,a) = \frac12 \min \Big\{ 1-a, \alpha \Big\} + \frac12 \min \Big\{ (1-a) \Big(\frac{a}{x}\Big)^{f(\rho^*)} , \alpha \Big\},
\end{equation*}
and $\gamma<1-\alpha$. If $a\geq 1-\alpha$, it follows that
\begin{equation*}
    \Lambda(x,a)= \frac{1-a}2 + \frac {(1-a)a^{f(\rho^*)}}{2x^{f(\rho^*)}}
\end{equation*}
is decreasing in $a\in[1-\alpha,1]$ for any $x\geq a$. When $\gamma\leq a\leq1-\alpha$, we have that 
\begin{equation*}
    \Lambda(x,a)= \frac{1}2 \alpha + \frac1 2 \min\Big\{\frac {(1-a)a^{f(\rho^*)}}{x^{f(\rho^*)}},\alpha\Big\},
\end{equation*}
which achieves the maximum at $\frac{f(\rho^*)}{f(\rho^*)+1} $ for any $x\geq a$. Therefore, for any $x>\frac{f(\rho^*)}{f(\rho^*)+1}$, if $a\leq x$, then $\Lambda(x,a)\leq\Lambda(x,\frac{f(\rho^*)}{f(\rho^*)+1})$, if $a>x$, then $\Lambda(x,a)=\phi(g(x))\leq\Lambda(x,\frac{f(\rho^*)}{f(\rho^*)+1})$ as $[0, \frac{f(\rho^*)}{f(\rho^*)+1} ]$ is an equilibrium. For $\gamma\leq a<x\leq \frac{f(\rho^*)}{f(\rho^*)+1}$, we have that
\begin{equation*}
    \begin{aligned}
    \Lambda(x,a)=&\,\frac12 \min \Big\{ 1-a, \alpha \Big\} + \frac12 \min \Big\{ (1-a) \Big(\frac{a}{x}\Big)^{f(\rho^*)} , \alpha \Big\}
    \\
    \leq&
    \, \alpha  + \frac12 \min \Big\{ (1-x) \Big(\frac{x}{x}\Big)^{f(\rho^*)} , \alpha\Big\}=\phi(g(x))=\Lambda\Big(x,\frac{f(\rho^*)}{f(\rho^*)+1}\Big).
    \end{aligned}
\end{equation*}
Further, for any $x\leq\min\{a,\frac{f(\rho^*)}{f(\rho^*)+1}\}$, we have $\Lambda(x,a)=\Lambda(x,\frac{f(\rho^*)}{f(\rho^*)+1})=\phi(g(x))$. In conclusion, for any $x\geq 0$ and $a\geq \gamma$, we have
\begin{equation*}
\Lambda(x,a)\leq \Lambda\Big(x,\frac{f(\rho^*)}{f(\rho^*)+1}\Big).
\end{equation*}
That is, $\big[0,\frac{f(\rho^*)}{f(\rho^*)+1}\big]$ is an optimal equilibrium, but it is not the smallest equilibrium.

\begin{remark}
In this example, we note that the optimal equilibrium boundary $\frac{f(\rho^*)}{1+f(\rho^*)}$ is increasing in $\rho^*$ as $f(r)$ is increasing in $r$. Our theoretical finding again matches with the real life situation because the larger $\rho^*$ indicates the higher level of the most impatience among group members. The wider range of discount rates forces the social planner to take into account these more impatient members into the decision making and quit the investment more likely, which is consistent with the larger optimal equilibrium stopping region given the larger value of $\rho^*$.

More importantly, it is illustrated that the optimal equilibrium $\big[0,\frac{f(\rho^*)}{f(\rho^*)+1}\big]$ can differ from the smallest equilibrium $[0,\gamma]$ in our time inconsistent stopping problem under the particular choice of $\phi(x)=\min(x,\alpha)$ and $\P(f(\rho)=0)=\P(f(\rho)=f(\rho^*))=\frac12$ where $1-\alpha>\frac{f(\rho^*)}{1+f(\rho^*)}$. Generally speaking, the same conclusion holds that the optimal equilibrium may no longer coincide with the smallest equilibrium due to the nonlinear aggregation of the diversity distribution when $\phi(x)$ does not satisfy the assumption \textbf{C-(iii)}. For a single agent's time inconsistent stopping under decreasing impatient discounting, it was shown in \cite{HuangZhou20} that the optimal equilibrium can be characterized as the intersection of all time consistent equilibria, i.e. the smallest equilibrium. As a sharp contrast, in our framework with a general attitude function $\phi(x)$, it is an appealing open problem to find a general characterization of the optimal equilibrium if it exists, which will be left for our future research.

\end{remark}

\noindent\textbf{Example 3.} In the next example, we will show an extreme impact by the aggregation attitude and the diversity distribution such that there is no global optimal equilibrium. Let us consider $\alpha = \frac14$ and $\P(f(\rho)=1)=\P(f(\rho)=2)=\frac12$.

Similar to the previous case, we have that $a^*=\gamma=\frac{1+\sqrt{13}}{8}$. For all $x\geq a \geq \gamma$, it holds that
\begin{equation*}
\Lambda(x,a) = \frac12 \min \Big\{ (1-a)\Big(\frac{a}{x}\Big) , \frac14 \Big\} + \frac12 \min \Big\{ (1-a) \Big(\frac{a}{x}\Big)^2 , \frac14 \Big\}.
\end{equation*}
In what follows, for a fixed $x\geq \gamma$, we look for $a^{**}(x)$ that maximizes $V(x,[0,a])=\phi(g(x))\Lambda(x,a)$. We plot in Figure \ref{phaseplot} to illustrate the piecewise definition of $\Lambda(x,a)$, which satisfies that
\begin{equation*}
    \Lambda(x,a)=\begin{cases}
        \frac{1-a}2\cdot\big(\frac{a}x+\big(\frac{a}x\big)^2\big),&\text{ when }x\geq 4a(1-a) \text{ and }x\geq a\text{ (Blue region)},
        \\
        \frac{1}8 + \frac {(1-a)a^2}{2x^2},&\text{ when }2a\sqrt{1-a}\leq x< 4a(1-a)\text{ (Green region)},
        \\
        \frac 14,&\text{ when }a\leq x< 2a\sqrt{1-a}\text{ 
 (Yellow region)},
        \\
        \phi(g(x))=\min\{1-x,\frac14\},&\text{ when }x< a\text{ (Red region)}.
    \end{cases}
\end{equation*}

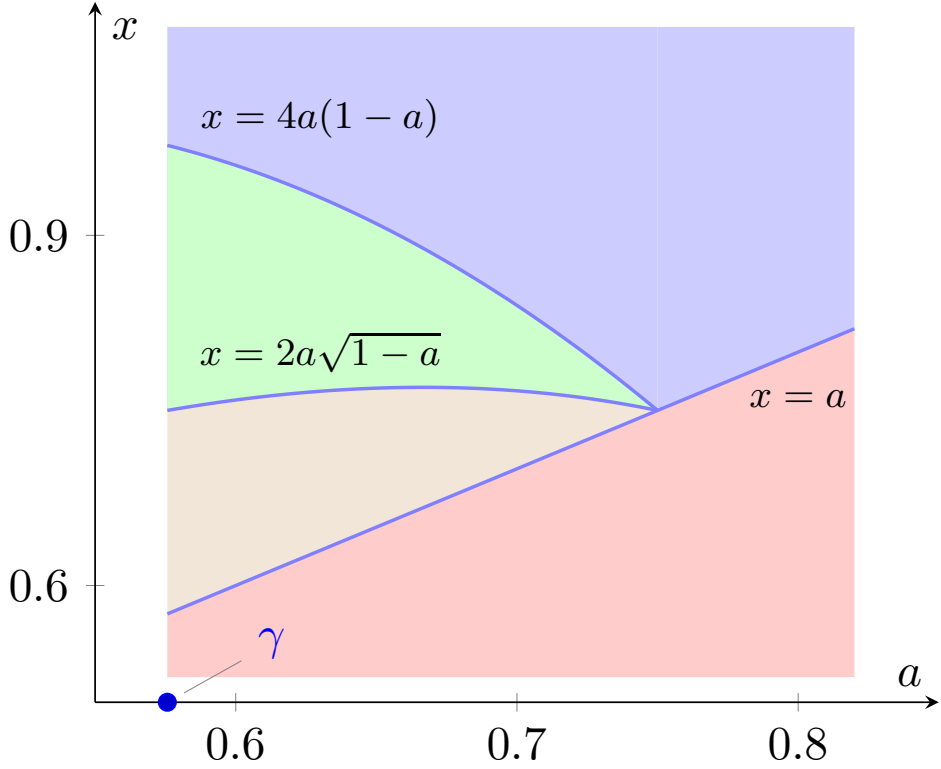
\begin{figure}[ht] 
\centering
\resizebox{0.8\columnwidth}{!}{
\begin{tikzpicture}
   \begin{axis}[
      xmin=0.55, xmax=0.85,
      ymin=0.5, ymax=1.1,
      axis lines=center,
      xtick distance=0.1, ytick distance=0.3,
      xlabel = {$a$},
      ylabel = {$x$},
    ]

    \addplot [domain=((1+sqrt(13))/8:0.82, samples=100, name path=f, thick, color=blue!50] {x};

    \addplot [domain=((1+sqrt(13))/8:0.82, samples=100, name path=x05, thick, color=white!50] {0.52};
    
    \addplot [domain=((1+sqrt(13))/8:0.82, samples=100, name path=x11, thick, color=white!50] {1.08};

    \addplot [domain=(1+sqrt(13))/8:0.75, samples=100, name path=g, thick, color=blue!50] {4*x*(1-x)};

    \addplot [domain=(1+sqrt(13))/8:0.75, samples=100, name path=x2, thick, color=blue!50] {2*x*sqrt(1-x)};

    \addplot[brown!20, opacity=0.4] fill between[of=f and x2, soft clip={domain=-2:0.75}];

    \addplot[blue!20, opacity=0.4] fill between[of=x2 and x11, soft clip={domain=-2:0.75}];
    \addplot[blue!20, opacity=0.4] fill between[of=f and x11, soft clip={domain=0.75:0.82}];
    \addplot[red!20, opacity=0.4] fill between[of=f and x05, soft clip={domain=-2:0.82}];
    \addplot[green!20, opacity=0.4] fill between[of=g and x2, soft clip={domain=-2:0.75}];

    \node[color=black, font=\footnotesize] at (axis cs: 0.8,0.76) {$x=a$};
    \node[color=black, font=\footnotesize] at (axis cs: 0.63,1) {$x=4a(1-a)$};
    \node[color=black, font=\footnotesize] at (axis cs: 0.63,0.8) {$x=2a\sqrt{1-a}$};
    \addplot+ coordinates {((1+sqrt(13))/8,0.5)}
            node [pin=20:{$\gamma$}] {};

   \end{axis}
  
\end{tikzpicture}
}
\caption{Different regions to determine $\Lambda(x,a)$.}\label{phaseplot}
\end{figure}

If $(x,a)$ falls into the blue region in Figure \ref{phaseplot}, we have
\begin{equation*}
    \Lambda(x,a)=\frac{1-a}2\cdot\Big(\frac{a}x+\Big(\frac{a}x\Big)^2\Big).
\end{equation*}
Taking derivative with respect to $a$, we get
\begin{equation*}
    \frac{\partial\Lambda(x,a)}{\partial a}= \frac{x-a(2x-2)-3a^2}{2x^2},
\end{equation*}
with two roots $a_1(x)=\frac{-x+1+\sqrt{x^2+x+1}}{3}$ and $a_2(x)=\frac{-x+1-\sqrt{x^2+x+1}}{3}<0$. Therefore, $\Lambda(x,a)$ is increasing when $0<a<a_1(x)$, and decreasing when $a>a_1(x)$. Note that $(x,a)$ is in the blue region if and only if $x\geq \max\{a,4a(1-a)\}$, which is equivalent to $\frac{1+\sqrt{1-x}}{2}\leq a\leq x$. To conclude, if $(x,a)$ is in the blue region, for any fixed $x$, the maximum of $a$ is attained by $\max\{a_1(x),\frac{1+\sqrt{1-x}}{2}\}$, which is illustrated by the red line in Figure \ref{phaseplot2}.

If $(x,a)$ falls into the green region, we have that
\begin{equation*}
    \Lambda(x,a)=\frac{1}8 + \frac {(1-a)a^2}{2x^2}
\end{equation*}
is increasing when $0<a<\frac23$, and decreasing when decreasing when $a>\frac 23$. Therefore, the red line in Figure \ref{phaseplot2} dominates the yellow region.

If $(x,a)$ falls into the yellow or brown region, $\Lambda(x,a)$ does not move with respect to $a$, therefore the shade region in Figure \ref{phaseplot2} corresponds to $a^{**}(x).$ In summary, we can conclude that

\begin{equation*}
    a^{**}(x)=\begin{cases}
        \max\{a_1(x),\gamma\},&\text{ when }x\geq \frac{7\sqrt{33}-9}{32},
        \\
        \frac12+\frac{\sqrt{1-x}}2, &\text{ when } \frac89 \leq x< \frac{7\sqrt{33}-9}{32},
        \\
        \frac23, &\text{ when } \frac{4\sqrt{3}}9 \leq x< \frac89,
        \\
        \big\{a\big|2a\sqrt{1-a}\geq x\big\}, &\text{ when }\frac34\leq x<\frac{4\sqrt{3}}9,
        \\
        \big\{\gamma\leq a\leq 1\big\},&\text{ when } \gamma\leq x< \frac 34.
    \end{cases}
\end{equation*}
As a result, in this example, there does not exist an optimal equilibrium such that its value function can dominate the ones under other equilibria for all $x>0$.

{
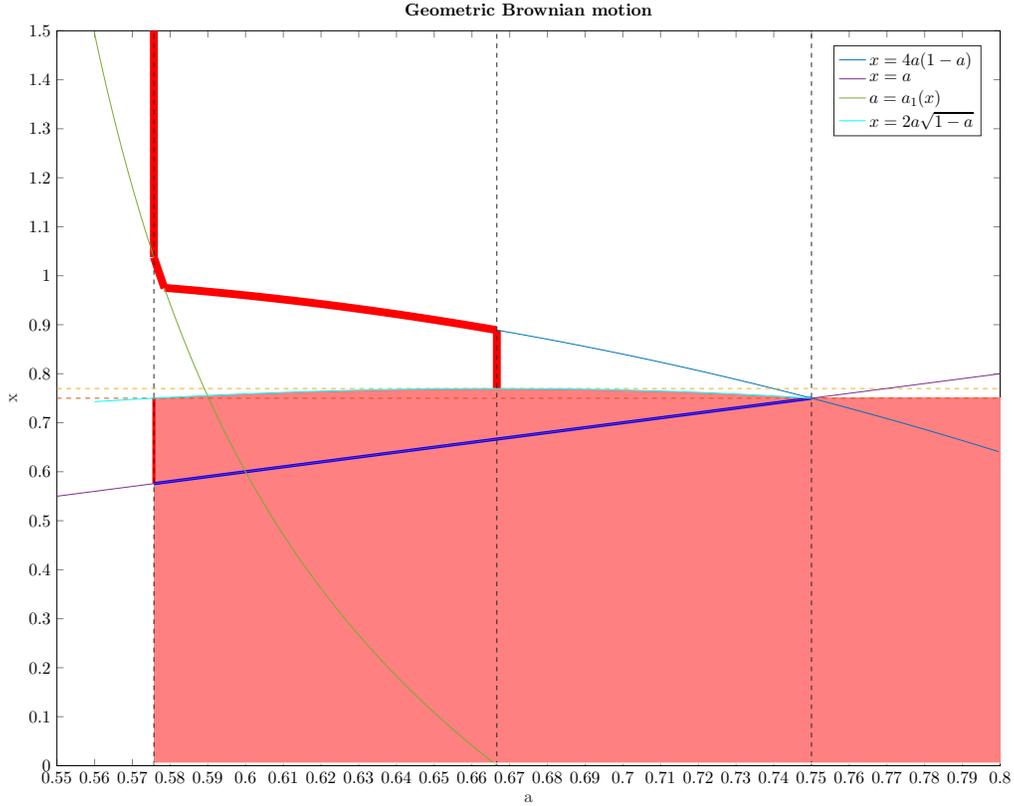
\begin{figure}[ht] 
\centering
\resizebox{\columnwidth}{!}{

\begin{tikzpicture}

\begin{axis}[%
width=8.288in,
height=6.463in,
at={(1.39in,0.872in)},
scale only axis,
xmin=0.55,
xmax=0.8,
xlabel style={font=\color{white!15!black}},
xlabel={a},
ymin=0,
ymax=1.5,
ylabel style={font=\color{white!15!black}},
ylabel={x},
axis background/.style={fill=white},
title style={font=\bfseries},
title={Geometric Brownian motion},
legend style={legend cell align=left, align=left, draw=white!15!black}
]
\addplot [color=red, line width=5.0pt, forget plot]
  table[row sep=crcr]{%
0.5757	1.03785\\
0.5757	1.5\\
};
\addplot [color=red, line width=2.0pt, forget plot]
  table[row sep=crcr]{%
0.5757	0.5757\\
0.5757	0.75\\
};
\addplot [color=blue, line width=2.0pt, forget plot]
  table[row sep=crcr]{%
0.5757	0.5757\\
0.75	0.75\\
};
\addplot [color=red, line width=5.0pt, forget plot]
  table[row sep=crcr]{%
0.6666	0.7698\\
0.6666	0.8889\\
};
\addplot [color=red, line width=5.0pt, forget plot]
  table[row sep=crcr]{%
0.5785	0.975351\\
0.579	0.975036\\
0.5795	0.974719\\
0.58	0.9744\\
0.5805	0.974079\\
0.581	0.973756\\
0.5815	0.973431\\
0.582	0.973104\\
0.5825	0.972775\\
0.583	0.972444\\
0.5835	0.972111\\
0.584	0.971776\\
0.5845	0.971439\\
0.585	0.9711\\
0.5855	0.970759\\
0.586	0.970416\\
0.5865	0.970071\\
0.587	0.969724\\
0.5875	0.969375\\
0.588	0.969024\\
0.5885	0.968671\\
0.589	0.968316\\
0.5895	0.967959\\
0.59	0.9676\\
0.5905	0.967239\\
0.591	0.966876\\
0.5915	0.966511\\
0.592	0.966144\\
0.5925	0.965775\\
0.593	0.965404\\
0.5935	0.965031\\
0.594	0.964656\\
0.5945	0.964279\\
0.595	0.9639\\
0.5955	0.963519\\
0.596	0.963136\\
0.5965	0.962751\\
0.597	0.962364\\
0.5975	0.961975\\
0.598	0.961584\\
0.5985	0.961191\\
0.599	0.960796\\
0.5995	0.960399\\
0.6	0.96\\
0.6005	0.959599\\
0.601	0.959196\\
0.6015	0.958791\\
0.602	0.958384\\
0.6025	0.957975\\
0.603	0.957564\\
0.6035	0.957151\\
0.604	0.956736\\
0.6045	0.956319\\
0.605	0.9559\\
0.6055	0.955479\\
0.606	0.955056\\
0.6065	0.954631\\
0.607	0.954204\\
0.6075	0.953775\\
0.608	0.953344\\
0.6085	0.952911\\
0.609	0.952476\\
0.6095	0.952039\\
0.61	0.9516\\
0.6105	0.951159\\
0.611	0.950716\\
0.6115	0.950271\\
0.612	0.949824\\
0.6125	0.949375\\
0.613	0.948924\\
0.6135	0.948471\\
0.614	0.948016\\
0.6145	0.947559\\
0.615	0.9471\\
0.6155	0.946639\\
0.616	0.946176\\
0.6165	0.945711\\
0.617	0.945244\\
0.6175	0.944775\\
0.618	0.944304\\
0.6185	0.943831\\
0.619	0.943356\\
0.6195	0.942879\\
0.62	0.9424\\
0.6205	0.941919\\
0.621	0.941436\\
0.6215	0.940951\\
0.622	0.940464\\
0.6225	0.939975\\
0.623	0.939484\\
0.6235	0.938991\\
0.624	0.938496\\
0.6245	0.937999\\
0.625	0.9375\\
0.6255	0.936999\\
0.626	0.936496\\
0.6265	0.935991\\
0.627	0.935484\\
0.6275	0.934975\\
0.628	0.934464\\
0.6285	0.933951\\
0.629	0.933436\\
0.6295	0.932919\\
0.63	0.9324\\
0.6305	0.931879\\
0.631	0.931356\\
0.6315	0.930831\\
0.632	0.930304\\
0.6325	0.929775\\
0.633	0.929244\\
0.6335	0.928711\\
0.634	0.928176\\
0.6345	0.927639\\
0.635	0.9271\\
0.6355	0.926559\\
0.636	0.926016\\
0.6365	0.925471\\
0.637	0.924924\\
0.6375	0.924375\\
0.638	0.923824\\
0.6385	0.923271\\
0.639	0.922716\\
0.6395	0.922159\\
0.64	0.9216\\
0.6405	0.921039\\
0.641	0.920476\\
0.6415	0.919911\\
0.642	0.919344\\
0.6425	0.918775\\
0.643	0.918204\\
0.6435	0.917631\\
0.644	0.917056\\
0.6445	0.916479\\
0.645	0.9159\\
0.6455	0.915319\\
0.646	0.914736\\
0.6465	0.914151\\
0.647	0.913564\\
0.6475	0.912975\\
0.648	0.912384\\
0.6485	0.911791\\
0.649	0.911196\\
0.6495	0.910599\\
0.65	0.91\\
0.6505	0.909399\\
0.651	0.908796\\
0.6515	0.908191\\
0.652	0.907584\\
0.6525	0.906975\\
0.653	0.906364\\
0.6535	0.905751\\
0.654	0.905136\\
0.6545	0.904519\\
0.655	0.9039\\
0.6555	0.903279\\
0.656	0.902656\\
0.6565	0.902031\\
0.657	0.901404\\
0.6575	0.900775\\
0.658	0.900144\\
0.6585	0.899511\\
0.659	0.898876\\
0.6595	0.898239\\
0.66	0.8976\\
0.6605	0.896959\\
0.661	0.896316\\
0.6615	0.895671\\
0.662	0.895024\\
0.6625	0.894375\\
0.663	0.893724\\
0.6635	0.893071\\
0.664	0.892416\\
0.6645	0.891759\\
0.665	0.8911\\
0.6655	0.890439\\
0.666	0.889776\\
0.6665	0.889111\\
};
\addplot [color=mycolor1]
  table[row sep=crcr]{%
0.6666	0.88897776\\
0.6676	0.88764096\\
0.6686	0.88629616\\
0.6696	0.88494336\\
0.6706	0.88358256\\
0.6716	0.88221376\\
0.6726	0.88083696\\
0.6736	0.87945216\\
0.6746	0.87805936\\
0.6756	0.87665856\\
0.6766	0.87524976\\
0.6776	0.87383296\\
0.6786	0.87240816\\
0.6796	0.87097536\\
0.6806	0.86953456\\
0.6816	0.86808576\\
0.6826	0.86662896\\
0.6836	0.86516416\\
0.6846	0.86369136\\
0.6856	0.86221056\\
0.6866	0.86072176\\
0.6876	0.85922496\\
0.6886	0.85772016\\
0.6896	0.85620736\\
0.6906	0.85468656\\
0.6916	0.85315776\\
0.6926	0.85162096\\
0.6936	0.85007616\\
0.6946	0.84852336\\
0.6956	0.84696256\\
0.6966	0.84539376\\
0.6976	0.84381696\\
0.6986	0.84223216\\
0.6996	0.84063936\\
0.7006	0.83903856\\
0.7016	0.83742976\\
0.7026	0.83581296\\
0.7036	0.83418816\\
0.7046	0.83255536\\
0.7056	0.83091456\\
0.7066	0.82926576\\
0.7076	0.82760896\\
0.7086	0.82594416\\
0.7096	0.82427136\\
0.7106	0.82259056\\
0.7116	0.82090176\\
0.7126	0.81920496\\
0.7136	0.81750016\\
0.7146	0.81578736\\
0.7156	0.81406656\\
0.7166	0.81233776\\
0.7176	0.81060096\\
0.7186	0.80885616\\
0.7196	0.80710336\\
0.7206	0.80534256\\
0.7216	0.80357376\\
0.7226	0.80179696\\
0.7236	0.80001216\\
0.7246	0.79821936\\
0.7256	0.79641856\\
0.7266	0.79460976\\
0.7276	0.79279296\\
0.7286	0.79096816\\
0.7296	0.78913536\\
0.7306	0.78729456\\
0.7316	0.78544576\\
0.7326	0.78358896\\
0.7336	0.78172416\\
0.7346	0.77985136\\
0.7356	0.77797056\\
0.7366	0.77608176\\
0.7376	0.77418496\\
0.7386	0.77228016\\
0.7396	0.77036736\\
0.7406	0.76844656\\
0.7416	0.76651776\\
0.7426	0.76458096\\
0.7436	0.76263616\\
0.7446	0.76068336\\
0.7456	0.75872256\\
0.7466	0.75675376\\
0.7476	0.75477696\\
0.7486	0.75279216\\
0.7496	0.75079936\\
0.7506	0.74879856\\
0.7516	0.74678976\\
0.7526	0.74477296\\
0.7536	0.74274816\\
0.7546	0.74071536\\
0.7556	0.73867456\\
0.7566	0.73662576\\
0.7576	0.73456896\\
0.7586	0.73250416\\
0.7596	0.73043136\\
0.7606	0.72835056\\
0.7616	0.72626176\\
0.7626	0.72416496\\
0.7636	0.72206016\\
0.7646	0.71994736\\
0.7656	0.71782656\\
0.7666	0.71569776\\
0.7676	0.71356096\\
0.7686	0.71141616\\
0.7696	0.70926336\\
0.7706	0.70710256\\
0.7716	0.70493376\\
0.7726	0.70275696\\
0.7736	0.70057216\\
0.7746	0.69837936\\
0.7756	0.69617856\\
0.7766	0.69396976\\
0.7776	0.69175296\\
0.7786	0.68952816\\
0.7796	0.68729536\\
0.7806	0.68505456\\
0.7816	0.68280576\\
0.7826	0.68054896\\
0.7836	0.67828416\\
0.7846	0.67601136\\
0.7856	0.67373056\\
0.7866	0.67144176\\
0.7876	0.66914496\\
0.7886	0.66684016\\
0.7896	0.66452736\\
0.7906	0.66220656\\
0.7916	0.65987776\\
0.7926	0.65754096\\
0.7936	0.65519616\\
0.7946	0.65284336\\
0.7956	0.65048256\\
0.7966	0.64811376\\
0.7976	0.64573696\\
0.7986	0.64335216\\
0.7996	0.64095936\\
};
\addlegendentry{$x=4a(1-a)$}

\addplot [color=mycolor2, dashed, name path=x75, forget plot]
  table[row sep=crcr]{%
0.55	0.75\\
0.552	0.75\\
0.554	0.75\\
0.556	0.75\\
0.558	0.75\\
0.56	0.75\\
0.562	0.75\\
0.564	0.75\\
0.566	0.75\\
0.568	0.75\\
0.57	0.75\\
0.572	0.75\\
0.574	0.75\\
0.576	0.75\\
0.578	0.75\\
0.58	0.75\\
0.582	0.75\\
0.584	0.75\\
0.586	0.75\\
0.588	0.75\\
0.59	0.75\\
0.592	0.75\\
0.594	0.75\\
0.596	0.75\\
0.598	0.75\\
0.6	0.75\\
0.602	0.75\\
0.604	0.75\\
0.606	0.75\\
0.608	0.75\\
0.61	0.75\\
0.612	0.75\\
0.614	0.75\\
0.616	0.75\\
0.618	0.75\\
0.62	0.75\\
0.622	0.75\\
0.624	0.75\\
0.626	0.75\\
0.628	0.75\\
0.63	0.75\\
0.632	0.75\\
0.634	0.75\\
0.636	0.75\\
0.638	0.75\\
0.64	0.75\\
0.642	0.75\\
0.644	0.75\\
0.646	0.75\\
0.648	0.75\\
0.65	0.75\\
0.652	0.75\\
0.654	0.75\\
0.656	0.75\\
0.658	0.75\\
0.66	0.75\\
0.662	0.75\\
0.664	0.75\\
0.666	0.75\\
0.668	0.75\\
0.67	0.75\\
0.672	0.75\\
0.674	0.75\\
0.676	0.75\\
0.678	0.75\\
0.68	0.75\\
0.682	0.75\\
0.684	0.75\\
0.686	0.75\\
0.688	0.75\\
0.69	0.75\\
0.692	0.75\\
0.694	0.75\\
0.696	0.75\\
0.698	0.75\\
0.7	0.75\\
0.702	0.75\\
0.704	0.75\\
0.706	0.75\\
0.708	0.75\\
0.71	0.75\\
0.712	0.75\\
0.714	0.75\\
0.716	0.75\\
0.718	0.75\\
0.72	0.75\\
0.722	0.75\\
0.724	0.75\\
0.726	0.75\\
0.728	0.75\\
0.73	0.75\\
0.732	0.75\\
0.734	0.75\\
0.736	0.75\\
0.738	0.75\\
0.74	0.75\\
0.742	0.75\\
0.744	0.75\\
0.746	0.75\\
0.748	0.75\\
0.75	0.75\\
0.752	0.75\\
0.754	0.75\\
0.756	0.75\\
0.758	0.75\\
0.76	0.75\\
0.762	0.75\\
0.764	0.75\\
0.766	0.75\\
0.768	0.75\\
0.77	0.75\\
0.772	0.75\\
0.774	0.75\\
0.776	0.75\\
0.778	0.75\\
0.78	0.75\\
0.782	0.75\\
0.784	0.75\\
0.786	0.75\\
0.788	0.75\\
0.79	0.75\\
0.792	0.75\\
0.794	0.75\\
0.796	0.75\\
0.798	0.75\\
0.8	0.75\\
};
\addplot [color=mycolor3, dashed, forget plot]
  table[row sep=crcr]{%
0.55	0.7698\\
0.552	0.7698\\
0.554	0.7698\\
0.556	0.7698\\
0.558	0.7698\\
0.56	0.7698\\
0.562	0.7698\\
0.564	0.7698\\
0.566	0.7698\\
0.568	0.7698\\
0.57	0.7698\\
0.572	0.7698\\
0.574	0.7698\\
0.576	0.7698\\
0.578	0.7698\\
0.58	0.7698\\
0.582	0.7698\\
0.584	0.7698\\
0.586	0.7698\\
0.588	0.7698\\
0.59	0.7698\\
0.592	0.7698\\
0.594	0.7698\\
0.596	0.7698\\
0.598	0.7698\\
0.6	0.7698\\
0.602	0.7698\\
0.604	0.7698\\
0.606	0.7698\\
0.608	0.7698\\
0.61	0.7698\\
0.612	0.7698\\
0.614	0.7698\\
0.616	0.7698\\
0.618	0.7698\\
0.62	0.7698\\
0.622	0.7698\\
0.624	0.7698\\
0.626	0.7698\\
0.628	0.7698\\
0.63	0.7698\\
0.632	0.7698\\
0.634	0.7698\\
0.636	0.7698\\
0.638	0.7698\\
0.64	0.7698\\
0.642	0.7698\\
0.644	0.7698\\
0.646	0.7698\\
0.648	0.7698\\
0.65	0.7698\\
0.652	0.7698\\
0.654	0.7698\\
0.656	0.7698\\
0.658	0.7698\\
0.66	0.7698\\
0.662	0.7698\\
0.664	0.7698\\
0.666	0.7698\\
0.668	0.7698\\
0.67	0.7698\\
0.672	0.7698\\
0.674	0.7698\\
0.676	0.7698\\
0.678	0.7698\\
0.68	0.7698\\
0.682	0.7698\\
0.684	0.7698\\
0.686	0.7698\\
0.688	0.7698\\
0.69	0.7698\\
0.692	0.7698\\
0.694	0.7698\\
0.696	0.7698\\
0.698	0.7698\\
0.7	0.7698\\
0.702	0.7698\\
0.704	0.7698\\
0.706	0.7698\\
0.708	0.7698\\
0.71	0.7698\\
0.712	0.7698\\
0.714	0.7698\\
0.716	0.7698\\
0.718	0.7698\\
0.72	0.7698\\
0.722	0.7698\\
0.724	0.7698\\
0.726	0.7698\\
0.728	0.7698\\
0.73	0.7698\\
0.732	0.7698\\
0.734	0.7698\\
0.736	0.7698\\
0.738	0.7698\\
0.74	0.7698\\
0.742	0.7698\\
0.744	0.7698\\
0.746	0.7698\\
0.748	0.7698\\
0.75	0.7698\\
0.752	0.7698\\
0.754	0.7698\\
0.756	0.7698\\
0.758	0.7698\\
0.76	0.7698\\
0.762	0.7698\\
0.764	0.7698\\
0.766	0.7698\\
0.768	0.7698\\
0.77	0.7698\\
0.772	0.7698\\
0.774	0.7698\\
0.776	0.7698\\
0.778	0.7698\\
0.78	0.7698\\
0.782	0.7698\\
0.784	0.7698\\
0.786	0.7698\\
0.788	0.7698\\
0.79	0.7698\\
0.792	0.7698\\
0.794	0.7698\\
0.796	0.7698\\
0.798	0.7698\\
0.8	0.7698\\
};
\addplot [color=mycolor4]
  table[row sep=crcr]{%
0.55	0.55\\
0.552	0.552\\
0.554	0.554\\
0.556	0.556\\
0.558	0.558\\
0.56	0.56\\
0.562	0.562\\
0.564	0.564\\
0.566	0.566\\
0.568	0.568\\
0.57	0.57\\
0.572	0.572\\
0.574	0.574\\
0.576	0.576\\
0.578	0.578\\
0.58	0.58\\
0.582	0.582\\
0.584	0.584\\
0.586	0.586\\
0.588	0.588\\
0.59	0.59\\
0.592	0.592\\
0.594	0.594\\
0.596	0.596\\
0.598	0.598\\
0.6	0.6\\
0.602	0.602\\
0.604	0.604\\
0.606	0.606\\
0.608	0.608\\
0.61	0.61\\
0.612	0.612\\
0.614	0.614\\
0.616	0.616\\
0.618	0.618\\
0.62	0.62\\
0.622	0.622\\
0.624	0.624\\
0.626	0.626\\
0.628	0.628\\
0.63	0.63\\
0.632	0.632\\
0.634	0.634\\
0.636	0.636\\
0.638	0.638\\
0.64	0.64\\
0.642	0.642\\
0.644	0.644\\
0.646	0.646\\
0.648	0.648\\
0.65	0.65\\
0.652	0.652\\
0.654	0.654\\
0.656	0.656\\
0.658	0.658\\
0.66	0.66\\
0.662	0.662\\
0.664	0.664\\
0.666	0.666\\
0.668	0.668\\
0.67	0.67\\
0.672	0.672\\
0.674	0.674\\
0.676	0.676\\
0.678	0.678\\
0.68	0.68\\
0.682	0.682\\
0.684	0.684\\
0.686	0.686\\
0.688	0.688\\
0.69	0.69\\
0.692	0.692\\
0.694	0.694\\
0.696	0.696\\
0.698	0.698\\
0.7	0.7\\
0.702	0.702\\
0.704	0.704\\
0.706	0.706\\
0.708	0.708\\
0.71	0.71\\
0.712	0.712\\
0.714	0.714\\
0.716	0.716\\
0.718	0.718\\
0.72	0.72\\
0.722	0.722\\
0.724	0.724\\
0.726	0.726\\
0.728	0.728\\
0.73	0.73\\
0.732	0.732\\
0.734	0.734\\
0.736	0.736\\
0.738	0.738\\
0.74	0.74\\
0.742	0.742\\
0.744	0.744\\
0.746	0.746\\
0.748	0.748\\
0.75	0.75\\
0.752	0.752\\
0.754	0.754\\
0.756	0.756\\
0.758	0.758\\
0.76	0.76\\
0.762	0.762\\
0.764	0.764\\
0.766	0.766\\
0.768	0.768\\
0.77	0.77\\
0.772	0.772\\
0.774	0.774\\
0.776	0.776\\
0.778	0.778\\
0.78	0.78\\
0.782	0.782\\
0.784	0.784\\
0.786	0.786\\
0.788	0.788\\
0.79	0.79\\
0.792	0.792\\
0.794	0.794\\
0.796	0.796\\
0.798	0.798\\
0.8	0.8\\
};
\addlegendentry{$x=a$}

\addplot [color=mycolor5]
  table[row sep=crcr]{%
0.56	1.49333333333333\\
0.5605	1.47536570247934\\
0.561	1.45768032786885\\
0.5615	1.44027032520325\\
0.562	1.42312903225806\\
0.5625	1.40625\\
0.563	1.38962698412698\\
0.5635	1.37325393700787\\
0.564	1.357125\\
0.5645	1.34123449612403\\
0.565	1.32557692307692\\
0.5655	1.31014694656489\\
0.566	1.29493939393939\\
0.5665	1.2799492481203\\
0.567	1.26517164179104\\
0.5675	1.25060185185185\\
0.568	1.23623529411765\\
0.5685	1.22206751824818\\
0.569	1.20809420289855\\
0.5695	1.19431115107914\\
0.57	1.18071428571428\\
0.5705	1.16729964539007\\
0.571	1.15406338028169\\
0.5715	1.14100174825175\\
0.572	1.12811111111111\\
0.5725	1.11538793103448\\
0.573	1.10282876712329\\
0.5735	1.09043027210884\\
0.574	1.07818918918919\\
0.5745	1.06610234899329\\
0.575	1.05416666666667\\
0.5755	1.04237913907285\\
0.576	1.03073684210526\\
0.5765	1.01923692810457\\
0.577	1.00787662337662\\
0.5775	0.996653225806451\\
0.578	0.985564102564101\\
0.5785	0.974606687898089\\
0.579	0.963778481012657\\
0.5795	0.953077044025157\\
0.58	0.942499999999999\\
0.5805	0.9320450310559\\
0.581	0.921709876543208\\
0.5815	0.911492331288343\\
0.582	0.901390243902438\\
0.5825	0.891401515151515\\
0.583	0.881524096385541\\
0.5835	0.871755988023952\\
0.584	0.862095238095237\\
0.5845	0.852539940828402\\
0.585	0.843088235294116\\
0.5855	0.833738304093567\\
0.586	0.824488372093022\\
0.5865	0.815336705202312\\
0.587	0.806281609195401\\
0.5875	0.797321428571428\\
0.588	0.788454545454544\\
0.5885	0.779679378531073\\
0.589	0.770994382022471\\
0.5895	0.762398044692737\\
0.59	0.753888888888888\\
0.5905	0.745465469613259\\
0.591	0.737126373626372\\
0.5915	0.728870218579234\\
0.592	0.720695652173912\\
0.5925	0.712601351351351\\
0.593	0.704586021505375\\
0.5935	0.696648395721925\\
0.594	0.688787234042552\\
0.5945	0.681001322751322\\
0.595	0.673289473684209\\
0.5955	0.665650523560209\\
0.596	0.658083333333332\\
0.5965	0.650586787564766\\
0.597	0.643159793814432\\
0.5975	0.635801282051282\\
0.598	0.628510204081631\\
0.5985	0.621285532994923\\
0.599	0.614126262626261\\
0.5995	0.607031407035175\\
0.6	0.599999999999999\\
0.6005	0.593031094527363\\
0.601	0.586123762376236\\
0.6015	0.579277093596059\\
0.602	0.57249019607843\\
0.6025	0.565762195121951\\
0.603	0.559092233009708\\
0.6035	0.552479468599033\\
0.604	0.545923076923076\\
0.6045	0.539422248803827\\
0.605	0.532976190476189\\
0.6055	0.526584123222748\\
0.606	0.520245283018867\\
0.6065	0.513958920187793\\
0.607	0.507724299065419\\
0.6075	0.501540697674418\\
0.608	0.495407407407406\\
0.6085	0.489323732718894\\
0.609	0.483288990825687\\
0.6095	0.477302511415525\\
0.61	0.471363636363635\\
0.6105	0.465471719457013\\
0.611	0.459626126126125\\
0.6115	0.453826233183856\\
0.612	0.448071428571427\\
0.6125	0.442361111111111\\
0.613	0.436694690265486\\
0.6135	0.431071585903083\\
0.614	0.425491228070174\\
0.6145	0.419953056768559\\
0.615	0.414456521739129\\
0.6155	0.409001082251082\\
0.616	0.403586206896551\\
0.6165	0.398211373390557\\
0.617	0.392876068376067\\
0.6175	0.387579787234042\\
0.618	0.382322033898304\\
0.6185	0.377102320675105\\
0.619	0.371920168067226\\
0.6195	0.36677510460251\\
0.62	0.361666666666666\\
0.6205	0.356594398340249\\
0.621	0.351557851239668\\
0.6215	0.346556584362139\\
0.622	0.341590163934425\\
0.6225	0.336658163265306\\
0.623	0.331760162601625\\
0.6235	0.326895748987854\\
0.624	0.322064516129031\\
0.6245	0.317266064257028\\
0.625	0.312499999999999\\
0.6255	0.30776593625498\\
0.626	0.303063492063491\\
0.6265	0.298392292490118\\
0.627	0.293751968503936\\
0.6275	0.289142156862745\\
0.628	0.284562499999999\\
0.6285	0.280012645914396\\
0.629	0.275492248062014\\
0.6295	0.271000965250965\\
0.63	0.266538461538461\\
0.6305	0.262104406130268\\
0.631	0.257698473282442\\
0.6315	0.253320342205323\\
0.632	0.248969696969696\\
0.6325	0.244646226415094\\
0.633	0.240349624060149\\
0.6335	0.236079588014981\\
0.634	0.231835820895521\\
0.6345	0.227618029739777\\
0.635	0.223425925925925\\
0.6355	0.21925922509225\\
0.636	0.215117647058823\\
0.6365	0.211000915750915\\
0.637	0.206908759124087\\
0.6375	0.202840909090909\\
0.638	0.198797101449274\\
0.6385	0.194777075812274\\
0.639	0.190780575539568\\
0.6395	0.18680734767025\\
0.64	0.182857142857142\\
0.6405	0.178929715302491\\
0.641	0.175024822695034\\
0.6415	0.171142226148409\\
0.642	0.167281690140844\\
0.6425	0.16344298245614\\
0.643	0.159625874125873\\
0.6435	0.155830139372822\\
0.644	0.152055555555555\\
0.6445	0.148301903114186\\
0.645	0.14456896551724\\
0.6455	0.140856529209622\\
0.646	0.137164383561643\\
0.6465	0.133492320819112\\
0.647	0.129840136054421\\
0.6475	0.126207627118643\\
0.648	0.122594594594594\\
0.6485	0.119000841750841\\
0.649	0.115426174496643\\
0.6495	0.111870401337792\\
0.65	0.108333333333332\\
0.6505	0.104814784053156\\
0.651	0.101314569536423\\
0.6515	0.0978325082508245\\
0.652	0.0943684210526305\\
0.6525	0.0909221311475404\\
0.653	0.0874934640522866\\
0.6535	0.0840822475570027\\
0.654	0.0806883116883107\\
0.6545	0.0773114886731385\\
0.655	0.0739516129032249\\
0.6555	0.0706085209003209\\
0.656	0.0672820512820504\\
0.6565	0.0639720447284339\\
0.657	0.0606783439490437\\
0.6575	0.0574007936507932\\
0.658	0.0541392405063282\\
0.6585	0.0508935331230278\\
0.659	0.0476635220125778\\
0.6595	0.044449059561128\\
0.66	0.0412499999999991\\
0.6605	0.0380661993769466\\
0.661	0.0348975155279494\\
0.6615	0.0317438080495351\\
0.662	0.0286049382716039\\
0.6625	0.0254807692307686\\
0.663	0.022371165644171\\
0.6635	0.0192759938837914\\
0.664	0.0161951219512186\\
0.6645	0.013128419452887\\
0.665	0.0100757575757566\\
0.6655	0.00703700906344351\\
0.666	0.00401204819277023\\
0.6665	0.00100075075075012\\
};
\addlegendentry{$a=a_1(x)$}

\addplot [color=red, line width=5.0pt, forget plot]
  table[row sep=crcr]{%
0.5757	1.03770495376486\\
0.5759	1.03305382081686\\
0.5761	1.02842555847569\\
0.5763	1.02381998689384\\
0.5765	1.01923692810457\\
0.5767	1.01467620599739\\
0.5769	1.01013764629389\\
0.5771	1.005621076524\\
0.5773	1.00112632600259\\
0.5775	0.996653225806451\\
0.5777	0.992201608751609\\
0.5779	0.987771309370989\\
0.5781	0.983362163892444\\
0.5783	0.978974010217113\\
0.5785	0.974606687898089\\
};
\addplot [color=mycolor6,name path=gg]
  table[row sep=crcr]{%
0.56	0.74292395303961\\
0.5605	0.74316466513149\\
0.561	0.743404382553668\\
0.5615	0.743643104250957\\
0.562	0.743880829165532\\
0.5625	0.744117556236916\\
0.563	0.74435328440197\\
0.5635	0.744588012594884\\
0.564	0.744821739747169\\
0.5645	0.745054464787642\\
0.565	0.74528618664242\\
0.5655	0.745516904234907\\
0.566	0.745746616485787\\
0.5665	0.745975322313011\\
0.567	0.746203020631785\\
0.5675	0.746429710354565\\
0.568	0.746655390391042\\
0.5685	0.746880059648134\\
0.569	0.747103717029972\\
0.5695	0.747326361437893\\
0.57	0.747547991770428\\
0.5705	0.747768606923291\\
0.571	0.747988205789369\\
0.5715	0.74820678725871\\
0.572	0.748424350218511\\
0.5725	0.748640893553111\\
0.573	0.748856416143976\\
0.5735	0.749070916869691\\
0.574	0.749284394605947\\
0.5745	0.749496848225528\\
0.575	0.749708276598305\\
0.5755	0.749918678591219\\
0.576	0.750128053068274\\
0.5765	0.750336398890524\\
0.577	0.75054371491606\\
0.5775	0.75075\\
0.578	0.750955252994478\\
0.5785	0.75115947274863\\
0.579	0.751362658108586\\
0.5795	0.751564807917454\\
0.58	0.751765921015312\\
0.5805	0.751965996239192\\
0.581	0.752165032423071\\
0.5815	0.752363028397861\\
0.582	0.752559982991389\\
0.5825	0.752755895028395\\
0.583	0.752950763330512\\
0.5835	0.753144586716256\\
0.584	0.753337364001016\\
0.5845	0.75352909399704\\
0.585	0.75371977551342\\
0.5855	0.753909407356083\\
0.586	0.754097988327777\\
0.5865	0.754285517228059\\
0.587	0.75447199285328\\
0.5875	0.754657413996576\\
0.588	0.754841779447852\\
0.5885	0.75502508799377\\
0.589	0.755207338417735\\
0.5895	0.755388529499886\\
0.59	0.755568660017076\\
0.5905	0.755747728742866\\
0.591	0.755925734447505\\
0.5915	0.756102675897923\\
0.592	0.756278551857713\\
0.5925	0.75645336108712\\
0.593	0.756627102343023\\
0.5935	0.75679977437893\\
0.594	0.756971375944956\\
0.5945	0.757141905787812\\
0.595	0.757311362650792\\
0.5955	0.75747974527376\\
0.596	0.757647052393131\\
0.5965	0.757813282741864\\
0.597	0.757978435049441\\
0.5975	0.758142508041859\\
0.598	0.75830550044161\\
0.5985	0.75846741096767\\
0.599	0.758628238335484\\
0.5995	0.758787981256952\\
0.6	0.758946638440411\\
0.6005	0.759104208590626\\
0.601	0.759260690408769\\
0.6015	0.759416082592409\\
0.602	0.759570383835494\\
0.6025	0.759723592828339\\
0.603	0.759875708257607\\
0.6035	0.760026728806297\\
0.604	0.760176653153726\\
0.6045	0.760325479975517\\
0.605	0.76047320794358\\
0.6055	0.7606198357261\\
0.606	0.760765361987518\\
0.6065	0.760909785388518\\
0.607	0.761053104586007\\
0.6075	0.761195318233106\\
0.608	0.761336424979128\\
0.6085	0.761476423469565\\
0.609	0.761615312346069\\
0.6095	0.761753090246439\\
0.61	0.761889755804605\\
0.6105	0.762025307650606\\
0.611	0.76215974441058\\
0.6115	0.762293064706744\\
0.612	0.762425267157379\\
0.6125	0.76255635037681\\
0.613	0.762686312975394\\
0.6135	0.762815153559498\\
0.614	0.762942870731485\\
0.6145	0.763069463089698\\
0.615	0.763194929228438\\
0.6155	0.76331926773795\\
0.616	0.763442477204406\\
0.6165	0.763564556209886\\
0.617	0.76368550333236\\
0.6175	0.763805317145672\\
0.618	0.763923996219519\\
0.6185	0.764041539119438\\
0.619	0.764157944406783\\
0.6195	0.764273210638709\\
0.62	0.764387336368153\\
0.6205	0.764500320143818\\
0.621	0.76461216051015\\
0.6215	0.764722856007325\\
0.622	0.764832405171224\\
0.6225	0.764940806533421\\
0.623	0.765048058621156\\
0.6235	0.765154159957325\\
0.624	0.765259109060454\\
0.6245	0.765362904444682\\
0.625	0.765465544619743\\
0.6255	0.765567028090944\\
0.626	0.765667353359146\\
0.6265	0.765766518920748\\
0.627	0.76586452326766\\
0.6275	0.76596136488729\\
0.628	0.76605704226252\\
0.6285	0.766151553871687\\
0.629	0.766244898188562\\
0.6295	0.766337073682332\\
0.63	0.766428078817576\\
0.6305	0.766517912054245\\
0.631	0.766606571847646\\
0.6315	0.766694056648413\\
0.632	0.766780364902493\\
0.6325	0.76686549505112\\
0.633	0.766949445530799\\
0.6335	0.767032214773278\\
0.634	0.767113801205532\\
0.6345	0.767194203249738\\
0.635	0.767273419323255\\
0.6355	0.767351447838603\\
0.636	0.767428287203436\\
0.6365	0.767503935820527\\
0.637	0.76757839208774\\
0.6375	0.767651654398009\\
0.638	0.767723721139317\\
0.6385	0.767794590694673\\
0.639	0.767864261442086\\
0.6395	0.767932731754546\\
0.64	0.768\\
0.6405	0.768066064541326\\
0.641	0.768130923736312\\
0.6415	0.768194575937633\\
0.642	0.768257019492826\\
0.6425	0.768318252744265\\
0.643	0.76837827402914\\
0.6435	0.768437081679431\\
0.644	0.768494674021883\\
0.6445	0.768551049377984\\
0.645	0.768606206063937\\
0.6455	0.768660142390641\\
0.646	0.768712856663657\\
0.6465	0.768764347183192\\
0.647	0.76881461224407\\
0.6475	0.768863650135705\\
0.648	0.768911459142078\\
0.6485	0.768958037541711\\
0.649	0.769003383607641\\
0.6495	0.769047495607391\\
0.65	0.76909037180295\\
0.6505	0.769132010450742\\
0.651	0.7691724098016\\
0.6515	0.76921156810074\\
0.652	0.769249483587737\\
0.6525	0.769286154496492\\
0.653	0.769321579055209\\
0.6535	0.769355755486368\\
0.654	0.769388682006696\\
0.6545	0.769420356827138\\
0.655	0.769450778152833\\
0.6555	0.769479944183083\\
0.656	0.769507853111325\\
0.6565	0.769534503125104\\
0.657	0.769559892406043\\
0.6575	0.769584019129815\\
0.658	0.769606881466116\\
0.6585	0.76962847757863\\
0.659	0.769648805625007\\
0.6595	0.769667863756829\\
0.66	0.76968565011958\\
0.6605	0.769702162852619\\
0.661	0.76971740008915\\
0.6615	0.769731359956186\\
0.662	0.769744040574528\\
0.6625	0.769755440058724\\
0.663	0.769765556517048\\
0.6635	0.76977438805146\\
0.664	0.769781932757583\\
0.6645	0.769788188724665\\
0.665	0.76979315403555\\
0.6655	0.769796826766648\\
0.666	0.769799204987898\\
0.6665	0.769800286762742\\
0.667	0.769800070148087\\
0.6675	0.769798553194276\\
0.668	0.769795733945051\\
0.6685	0.769791610437526\\
0.669	0.769786180702148\\
0.6695	0.769779442762666\\
0.67	0.769771394636096\\
0.6705	0.769762034332689\\
0.671	0.769751359855895\\
0.6715	0.769739369202329\\
0.672	0.769726060361737\\
0.6725	0.769711431316958\\
0.673	0.769695480043894\\
0.6735	0.76967820451147\\
0.674	0.769659602681601\\
0.6745	0.769639672509155\\
0.675	0.769618411941918\\
0.6755	0.769595818920555\\
0.676	0.769571891378577\\
0.6765	0.7695466272423\\
0.677	0.769520024430814\\
0.6775	0.769492080855937\\
0.678	0.769462794422186\\
0.6785	0.769432163026735\\
0.679	0.769400184559375\\
0.6795	0.769366856902479\\
0.68	0.769332177930964\\
0.6805	0.769296145512247\\
0.681	0.769258757506211\\
0.6815	0.769220011765165\\
0.682	0.769179906133799\\
0.6825	0.769138438449152\\
0.683	0.769095606540565\\
0.6835	0.769051408229645\\
0.684	0.76900584133022\\
0.6845	0.768958903648303\\
0.685	0.768910592982045\\
0.6855	0.768860907121698\\
0.686	0.76880984384957\\
0.6865	0.768757400939984\\
0.687	0.768703576159237\\
0.6875	0.768648367265553\\
0.688	0.768591772009042\\
0.6885	0.768533788131661\\
0.689	0.76847441336716\\
0.6895	0.768413645441048\\
0.69	0.768351482070543\\
0.6905	0.76828792096453\\
0.691	0.768222959823514\\
0.6915	0.768156596339575\\
0.692	0.768088828196323\\
0.6925	0.768019653068852\\
0.693	0.767949068623695\\
0.6935	0.767877072518772\\
0.694	0.767803662403352\\
0.6945	0.767728835917995\\
0.695	0.767652590694515\\
0.6955	0.767574924355922\\
0.696	0.767495834516383\\
0.6965	0.767415318781167\\
0.697	0.767333374746596\\
0.6975	0.76725\\
0.698	0.767165192119663\\
0.6985	0.767078948674776\\
0.699	0.766991267225384\\
0.6995	0.766902145322335\\
0.7	0.766811580507233\\
0.7005	0.76671957031238\\
0.701	0.766626112260729\\
0.7015	0.766531203865831\\
0.702	0.766434842631779\\
0.7025	0.766337026053159\\
0.703	0.766237751614993\\
0.7035	0.766137016792688\\
0.704	0.76603481905198\\
0.7045	0.765931155848879\\
0.705	0.765826024629615\\
0.7055	0.765719422830582\\
0.706	0.765611347878282\\
0.7065	0.765501797189269\\
0.707	0.76539076817009\\
0.7075	0.765278258217232\\
0.708	0.765164264717061\\
0.7085	0.765048785045764\\
0.709	0.764931816569294\\
0.7095	0.764813356643305\\
0.71	0.7646934026131\\
0.7105	0.764571951813562\\
0.711	0.764449001569104\\
0.7115	0.7643245491936\\
0.712	0.764198591990328\\
0.7125	0.764071127251907\\
0.713	0.763942152260235\\
0.7135	0.763811664286426\\
0.714	0.763679660590748\\
0.7145	0.763546138422558\\
0.715	0.763411095020239\\
0.7155	0.763274527611134\\
0.716	0.763136433411484\\
0.7165	0.762996809626358\\
0.717	0.762855653449589\\
0.7175	0.762712962063711\\
0.718	0.762568732639885\\
0.7185	0.762422962337835\\
0.719	0.762275648305782\\
0.7195	0.76212678768037\\
0.72	0.761976377586602\\
0.7205	0.761824415137767\\
0.721	0.761670897435369\\
0.7215	0.76151582156906\\
0.722	0.761359184616565\\
0.7225	0.76120098364361\\
0.723	0.761041215703854\\
0.7235	0.760879877838808\\
0.724	0.760716967077769\\
0.7245	0.760552480437741\\
0.725	0.76038641492336\\
0.7255	0.760218767526822\\
0.726	0.760049535227803\\
0.7265	0.759878714993386\\
0.727	0.75970630377798\\
0.7275	0.759532298523242\\
0.728	0.759356696158005\\
0.7285	0.759179493598187\\
0.729	0.759000687746724\\
0.7295	0.758820275493479\\
0.73	0.758638253715168\\
0.7305	0.758454619275273\\
0.731	0.758269369023964\\
0.7315	0.75808249979801\\
0.732	0.757894008420703\\
0.7325	0.757703891701765\\
0.733	0.75751214643727\\
0.7335	0.757318769409553\\
0.734	0.757123757387126\\
0.7345	0.75692710712459\\
0.735	0.756728815362544\\
0.7355	0.756528878827504\\
0.736	0.756327294231803\\
0.7365	0.756124058273508\\
0.737	0.755919167636329\\
0.7375	0.75571261898952\\
0.738	0.755504408987797\\
0.7385	0.755294534271234\\
0.739	0.755082991465177\\
0.7395	0.754869777180144\\
0.74	0.754654888011732\\
0.7405	0.75443832054052\\
0.741	0.754220071331969\\
0.7415	0.754000136936327\\
0.742	0.753778513888529\\
0.7425	0.753555198708097\\
0.743	0.753330187899038\\
0.7435	0.753103477949744\\
0.744	0.752875065332888\\
0.7445	0.752644946505323\\
0.745	0.752413117907975\\
0.7455	0.75217957596574\\
0.746	0.751944317087376\\
0.7465	0.751707337665398\\
0.747	0.751468634075967\\
0.7475	0.751228202678787\\
0.748	0.750986039816986\\
0.7485	0.750742141817015\\
0.749	0.750496504988531\\
0.7495	0.750249125624282\\
0.75	0.75\\
};
\addlegendentry{$x=2a\sqrt{1-a}$}

\addplot [color=white!15!black, dashed, forget plot]
  table[row sep=crcr]{%
0.5757	0\\
0.5757	1.5\\
};
\addplot [color=white!15!black, dashed, forget plot]
  table[row sep=crcr]{%
0.6666	0\\
0.6666	1.5\\
};
\addplot [color=white!15!black, dashed, forget plot]
  table[row sep=crcr]{%
0.75	0\\
0.75	1.5\\
};

 \addplot [domain=(0.5757:0.8, samples=100, name path=x005, thick, color=white!50] {0.005};

  \addplot[red!50, opacity=0.2] fill between[of=x005 and x75, soft clip={domain=0.5757:0.8}];
  \addplot[red!50, opacity=0.2] fill between[of=gg and x75, soft clip={domain=0.5757:0.75}];

\end{axis}

\begin{axis}[%
width=10.694in,
height=7.931in,
at={(0in,0in)},
scale only axis,
xmin=0,
xmax=1,
ymin=0,
ymax=1,
axis line style={draw=none},
ticks=none,
axis x line*=bottom,
axis y line*=left
]
\end{axis}
\end{tikzpicture}%
}
\caption{Optimal barrier $a^{**}(x)$ that attains the maximum of $\Lambda(x,a)$ for geometric Brownian Motion. (Red line and red region)}\label{phaseplot2}
\end{figure}}

\subsection{The case of Bessel process with $n=3$}\label{sec:Bessel_3}

We next focus on another example of underlying process $X_t$ modelled by the Bessel process with degree $\nu=\frac12$(or equivalently $n=2\nu+2=3$). More precisely, we consider that $X_t=\sqrt{(W^1_t)^2+(W^2_t)^2+(W^3_t)^2}$, where $(\left( W^1_t, W^2_t, W^3_t \right)$ is a three-dimensional standard Brownian motion. Let us consider the same Put payoff function $g(x)=(1-x)^+$ and the aggregation attitude $\phi(x)=\min(x, \alpha)$. 

\textbf{Characterization of the equilibrium.}
We will first give the characterization that $R=[0,a]$ is an equilibrium if and only if $a > a^*$ for some threshold $a^*$. In particular, we shall analyze how the parameter $\alpha$ in the aggregation function $\phi(x)$ and the diversity distribution may affect the value of $a^*$, and discuss the existence of the optimal equilibrium. This subsection can also serve as a comparison model with respect to the case of geometric Brownian motion, and demonstrates how the underlying state process may affect the conclusion of the optimal equilibrium.

In the current context, $\Lambda(x,a)$ defined in \eqref{Lambda_eq} can be written by
$$
\Lambda(x,a) =\int_0^\infty\min \Big\{ (1-a)\Big(\frac{a}{x}\Big) e^{-\sqrt{2 r}(x-a)}, \alpha \Big\} dF_\rho(r).
$$
Let $h_{r}(x):= (1-x) x e^{\sqrt{2 r} x}$, it follows that $h'_r(x) = e^{\sqrt{2r} x} \left[ 1-2x+(1-x)x  \sqrt{2r} \right]$, and consequently
$$
h'_r(x)\geq 0 \Leftrightarrow - \sqrt{\frac14+\frac{1}{2r}} -\frac{1}{\sqrt{2 r}} + \frac12 \leq x \leq \sqrt{\frac14+\frac{1}{2r}} -\frac{1}{\sqrt{2r}} + \frac12.
$$
Denote $x^*(r):=  \sqrt{\frac14+\frac{1}{2r}} -\frac{1}{\sqrt{2 r}} + \frac12$. Then on the interval $[0,1]$, $x \mapsto h_{r}(x)$ is increasing on $[0, x^*(r)]$, and decreasing on $[x^*(r), 1]$. Moreover, $x^*(r)'= \frac12 (2 r)^{-2} \big[ (2r)^{\frac12} - \big( \frac14+\frac{1}{2r} \big)^{-\frac12} \big] > 0$ for $r > 0$, hence $r \mapsto x^*(r)$ is increasing. We note here that $x \mapsto h_{r}(x)$ has the same monotonicity as the map $x \mapsto (1-x) x^{f(r)}$ in the previous example of geometric Brownian motion, hence the results are similar when all equilibria are restricted to the one-barrier form $[0,a]$ for some $a>0$.

\begin{proposition} \label{prop:Bessel}
$R$ is an equilibrium if and only if $R = [0, a]$ for some $a \geq a^*$, and $a^*$ is defined as:
$$
a^*=\left\{\begin{aligned}
& x^*\left(\frac12 \left( \int_0^\infty\sqrt{2 r} \,dF_\rho(r) \right)^2 \right), & \mbox{ if } 1-\alpha \leq x^*\left(\frac12 \left( \int_0^\infty\sqrt{2 r} \,dF_\rho(r)  \right)^2 \right),  \\
&1-\alpha, & \mbox{ if } x^*\left(\frac12 \int_0^\infty\sqrt{2 r} \,dF_\rho(r)  \right) < 1-\alpha \leq x^*(\rho^*), \\
&\gamma, & \mbox{ if } 1-\alpha > x^*(\rho^*),
\end{aligned} \right.
$$
where $\gamma$ is the unique solution to the equation 
\begin{equation} \label{eq:gamma_bessel}
(1-\gamma) \gamma e^{\sqrt{2\rho^*}} = \alpha (1-\alpha) e^{\sqrt{2\rho^*} (1-\alpha)}
\end{equation}
on the interval $[0,1-\alpha)$.
\end{proposition}

\begin{proof}
\textbf{Case 1}: $1-\alpha \leq x^*(\rho^*) = \sqrt{\frac14+\frac{1}{2\rho^*}} -\frac{1}{\sqrt{2\rho^*}} + \frac12$.

\hspace{-8mm} \textbf{(a).} If $a<1-\alpha$, as $h_{r}$ is increasing on $[0, x^*(r)]$, there exists $\epsilon>0$ such that 
$$
   (1-a) a e^{\sqrt{2 (\rho^*-\epsilon)} a} = h_{\rho^*-\epsilon}(a) < h_{\rho^*-\epsilon}(1-\alpha) = \alpha (1-\alpha) e^{\sqrt{2 (\rho^*-\epsilon)} (1-\alpha)}.
$$
Then for $x=1-\alpha$, we have that
$$
 \begin{aligned}
    \Lambda(1-\alpha,a)
    =&\int_0^\infty\min\Big\{(1-a)\Big(\frac a {1-\alpha}\Big) e^{-\sqrt{2r}(1-\alpha-a)},\alpha\Big\}d F_\rho(r)
    \\
    \leq& (1-a)\Big(\frac a {1-\alpha}\Big) e^{-\sqrt{2 (\rho^*-\epsilon)}(1-\alpha-a)}\cdot\P(\rho\geq\rho^*-\epsilon)+\alpha \P(\rho<\rho^*-\epsilon)
    \\
    <&\alpha\P(\rho\geq\rho^*-\epsilon)+\alpha \P(\rho<\rho^*-\epsilon)=\alpha=\phi(g(1-\alpha)),
    \end{aligned}
$$
where in the second line we have used the fact that $r \mapsto x^*(r)$ is increasing. It follows that this case is not an equilibrium.

\hspace{-8mm} \textbf{(b).} If $a \geq 1-\alpha$, then $\phi(g(x)) = \min \{ 1-x, \alpha \}= 1-x$, for all $x \geq a \geq 1-\alpha$. Hence 
$$
\Lambda(x,a) = \int_0^\infty\min \Big\{ (1-a)\Big(\frac{a}{x}\Big) e^{-\sqrt{2r}(x-a)}, \alpha \Big\} d F_\rho(r)= \int_0^\infty (1-a)(\frac{a}{x}) e^{-\sqrt{2r}(x-a)} d F_\rho(r),
$$
where we have used the fact that $1-a \leq \alpha$, and $(\frac{a}{x}) e^{-\sqrt{2r}(x-a)} \leq 1$. Similar as Theorem \ref{small}, it follows from that fact that $x \mapsto \Lambda(x,a)$ is convex and $\Lambda(a,a) = \phi(g(a))=1-a$, that we have
$$
\Lambda(x,a) \geq  \phi(g(x)) \Longleftrightarrow \lim_{x\downarrow a}\partial_x \Lambda(x,a)  \geq  \phi'(g(a)) = -1.
$$
In addition, note that
$$
\lim_{x\downarrow a}\partial_x \Lambda =\lim_{x\downarrow a} \int_0^\infty \bigg(-(1-a)ax^{-2} e^{-(x-a) \sqrt{2r}} - (1-a)ax^{-1} e^{-(x-a) \sqrt{2r}} \sqrt{ 2r } \bigg)d F_\rho(r).
$$
Requiring $\lim_{x\downarrow a}\partial_x \Lambda \geq -1$ leads to 
$$
 - \int_0^\infty \sqrt{2r} \,dF_\rho(r) a^2 + \Big( \int_0^\infty \sqrt{2r} \,d F_\rho(r) - 2\Big) a + 1 \leq 0,
$$
and consequently $a \geq -\frac{1}{\int_0^\infty \sqrt{2r} \,dF_\rho(r)} + \sqrt{\frac14+\frac{1}{(\int_0^\infty \sqrt{2r} \,dF_\rho(r))^2}} + \frac12 = x^*(\frac12 \left( \int_0^\infty \sqrt{2r} \,dF_\rho(r)\right)^2 )$.
(The solution $a \leq -\frac{1}{\int_0^\infty \sqrt{2r} \,dF_\rho(r)} - \sqrt{\frac14+\frac{1}{(\int_0^\infty \sqrt{2r} \,dF_\rho(r))^2}} + \frac12$
is ruled out in view of $a > 0$.)

In summary, when $x^*(\frac12 \left(\int_0^\infty \sqrt{2r} \,dF_\rho(r) \right)^2 ) \leq 1-\alpha \leq x^*(\rho^*)$, we have $a^*= 1- \alpha$; when $x^*(\frac12 \left( \int_0^\infty \sqrt{2r} \,dF_\rho(r)\right)^2 ) \geq 1-\alpha $, we have $a^*= x^*(\frac12 \left( \int_0^\infty \sqrt{2r} \,dF_\rho(r)\right)^2 )$.

\textbf{Case 2}: $1-\alpha \geq x^*(\rho^*) = \sqrt{\frac14+\frac{1}{2\rho^*}} -\frac{1}{\sqrt{2\rho^*}} + \frac12$.

Let us first examine the equation \eqref{eq:gamma_bessel}. It follows from $x \mapsto h_{\rho^*}(x)$ is increasing on $[0,x^*(\rho^*)]$ and decreasing on $[x^*(\rho^*), 1-\alpha)$, together with the continuity of $x \mapsto h_{\rho^*}(x)$ and $h_{\rho^*}(0)=0$,  that \eqref{eq:gamma_bessel} has a unique solution on $[0, 1-\alpha)$. In what follows, we shall prove that $a^*=\gamma$ is the smallest $a$ such that $[0,a]$ is an equilibrium.

\hspace{-8mm} \textbf{(a).}  If $a < \gamma$, similar to \textbf{Case 1 a}, one can derive that $[0,a]$ is not an equilibrium.

\hspace{-8mm} \textbf{(b).} If $a \geq \gamma$, we distinguish two separate cases: $x \in [a, 1-\alpha)$ and $x \in [1-\alpha, 1]$. 

If $x \in [1-\alpha, 1]$, then $\phi(g(x))=\min \{ 1-x, \alpha \}=1-x$. In addition, we can verify that 
\begin{equation} \label{eq:subcase_inequality_Bessel}
h_{\rho^*}(x) \leq h_{\rho^*}(a), ~~~~\forall x \in [1-\alpha, 1], a \in [\gamma, x].
\end{equation}
We need to consider two separate cases $a \in [\gamma, x^*(\rho^*)]$ and $a \in [x^*(\rho^*),x]$, and use the fact that $x \mapsto h_{\rho^*}(x)$ is increasing on $[0,x^*(\rho^*)]$ and decreasing on $[x^*(\rho^*), 1]$.  For $a \in [x^*(\rho^*),x]$, $h_{\rho^*}(x) \leq h_{\rho^*}(a)$ follows directly from the monotonicity. For $a \in [\gamma, x^*(\rho^*)]$,  we have that $h_{\rho^*}(a) \geq h_{\rho^*}(\gamma) =  h_{\rho^*}(1-\alpha) \geq  h_{\rho^*}(x)$. In conclusion, \eqref{eq:subcase_inequality_Bessel} holds valid. It follows that 
\begin{align*}
\Lambda(x,a) &= \int_0^\infty \min \Big\{ (1-a)\Big(\frac{a}{x}\Big) e^{-\sqrt{2r}(x-a)}, \alpha \Big\} d F_\rho(r) \geq  \min \Big\{ (1-a)\Big(\frac{a}{x}\Big) e^{-\sqrt{2\rho^*}(x-a)}, \alpha \Big\}\\
 &\geq  \min \{ 1-x, \alpha \} = \phi(g(x)),
\end{align*}
 where we have used \eqref{eq:subcase_inequality_Bessel} in the second inequality.
 
If $x \in [a,1-\alpha)$, then $\phi(g(x))=\min \{ 1-x, \alpha \}=\alpha$. In addition, taking $x=1-\alpha$ in \eqref{eq:subcase_inequality_Bessel}, we get $f_{\rho^*}(1-\alpha) \leq f_{\rho^*}(a)$, and consequently
$$
(1-a)\Big(\frac{a}{x}\Big) e^{-\sqrt{2\rho^*}(x-a)} \geq (1-a)\Big(\frac{a}{1-\alpha}\Big) e^{-\sqrt{2\rho^*}(1 - \alpha -a)} \geq \alpha.
$$
It follows that
$$
\begin{aligned}
\Lambda(x,a) &= \int_0^\infty \min \Big\{ (4-a)\Big(\frac{a}{x}\Big) e^{-\sqrt{2r}(x-a)}, \alpha \Big\} dF_\rho(r) 
\\
&\geq  \min \Big\{ (4-a)\Big(\frac{a}{x}\Big) e^{-\sqrt{2\rho^*}(x-a)}, \alpha \Big\}  \geq  \alpha = \phi(g(x)).
\end{aligned}
$$
In conclusion, any $a$ such that $a \geq \gamma$ is an equilibrium.
\end{proof}

Based on the characterization of all equilibria with the threshold $a^*$, we will show some examples within the framework of Bessel process that the optimal equilibrium may: i) exist and coincide with the smallest equilibrium, and the optimal equilibrium threshold $a^*$ depends on the aggregation attitude function $\phi(x)$; and ii) not exist. In Remark \ref{bess-remark}, we will also elaborate that the optimal equilibrium needs to coincide with the smallest equilibrium if it exists in the model of Bessel process, which shows an interesting distinction from the previous example with geometric Brownian motion. 

\ \\
\noindent\textbf{Example 1.} An illustrative example to show the impact on the optimal equilibrium by the aggregation attitude. Suppose that $1-\alpha\leq x^*(\rho^*)$.

By Proposition \ref{prop:Bessel}, we have that $a^* = \max\Big\{1-\alpha,x^*\left(\frac12 \left( \int_0^\infty \sqrt{2r}\,dF_\rho(r) \right)^2 \right) \Big\}$ and for any $a\geq a^*$, we have
\begin{equation*}
    \Lambda(x,a)=\int_0^\infty\frac{(1-a)a}{x} e^{-\sqrt{2r} (x-a)}d F_\rho(r).
\end{equation*}
Using the argument in Theorem \ref{small}, when $a\geq x^*\left(\frac12 \left( \int_0^\infty \sqrt{2r}\,dF_\rho(r) \right)^2 \right)$, we have that $\Lambda(x,a)$ is non-increasing with respect to $a$ and hence
\begin{equation*}
\Lambda(x,a)\leq \Lambda(x,a^*)
\end{equation*}
for any $x\geq a\geq a^*$. Hence $[0,a^*]$ is the optimal equilibrium. Similar to Example 1 in the model of geometric Brownian motion, the optimal equilibrium boundary $a^*$ depends on the value $\alpha$ in the aggregation function $\phi(x)$ and is also decreasing in $\alpha$, which is consistent with the real life situation that the social planner will quit the investment more likely or earlier if he focuses more on these impatient group members with larger discount rates. 

\ \\
\noindent\textbf{Example 2.} An example when there is no global optimal equilibrium. Suppose that $\alpha = \frac15$, $\P(\rho=0)=\P(\rho=4)=\frac12$.

 It is clear that $\int_0^\infty \sqrt{2r}\,dF_\rho(r) = \sqrt{2}$, and $\rho^*=4$. From the previous discussion, we have that when  $1-\alpha <-\frac{1}{2\sqrt{2}}+\sqrt{\frac14+\frac18}+\frac12$, or equivalently $\alpha \geq \frac12-\sqrt{\frac18} + \sqrt{\frac38} \approx 0.2411$, the minimal equilibrium $a^*=\max \{ 1-\alpha, \frac{\sqrt{3}}{2} + \frac12 - \frac{1}{\sqrt{2}} \}$. When $a < \frac12-\sqrt{\frac18} + \sqrt{\frac38} \approx 0.2411$, the minimal equilibrium $a^*=\gamma$, with $\gamma \in (0, 1-\alpha)$ being uniquely determined by $(1-\gamma) \gamma e^{\sqrt{2\rho^*} \gamma} = \alpha (1-\alpha) e^{\sqrt{2\rho^*} (1-\alpha)}$.

In what follows, we take $\alpha=\frac15$. One can get that 
$$
\phi(1-x) = \min (1-x, \frac15),
$$
and
$$
\Lambda(x,a) = \frac12 \min \Big\{ (1-a)\Big(\frac{a}{x}\Big) , \frac15 \Big\} + \frac12 \min \Big\{ (1-a) \Big(\frac{a}{x}\Big) e^{-2\sqrt{2} (x-a)} , \frac15 \Big\}.
$$

Let us determine the critical points for $x$:
$$
 (1-a)\Big(\frac{a}{x}\Big) = \frac15 \Rightarrow x = x^*:= 5 a(1-a);  1-x=\frac15 \Rightarrow x=\frac45.
$$
 It is clear that when $\frac15<a<\frac45$, it holds that $x^*>\frac45$; while when $0<a<\frac15$ or $\frac45<a<1$, it holds that $x^*<\frac45$.
In addition
$$
 (1-a) (\frac{a}{x}) e^{-2\sqrt{2} (x-a)} = \frac15 \Leftrightarrow (1-a) a e^{2\sqrt{2} a} = \frac15 x e^{2\sqrt{2} x} \Rightarrow x = x^{**}(a).
$$
This equation cannot be solved analytically, and in the following we shall use the fact that 
$$
x \mapsto q(x):=  (1-a) \Big(\frac{a}{x}\Big) e^{-2\sqrt{2} (x-a)}
$$
is decreasing. Let us compare $\frac45$ with $x^{**}$. As $q(\frac45)=\frac15 \Leftrightarrow a(1-a)=\frac{4}{25}e^{2\sqrt{2}(\frac45-a)}$.  Solving the equation gives that $a=\frac45$ or $a= \gamma \approx 0.71305$. When $\gamma < a < \frac45$, $q(\frac45)>\frac15=q(x^{**})$, hence $x^{**}>\frac45$. When $a<\gamma$ or $a > \frac45$, $q(\frac45)<\frac15$, hence $x^{**} < \frac45$.

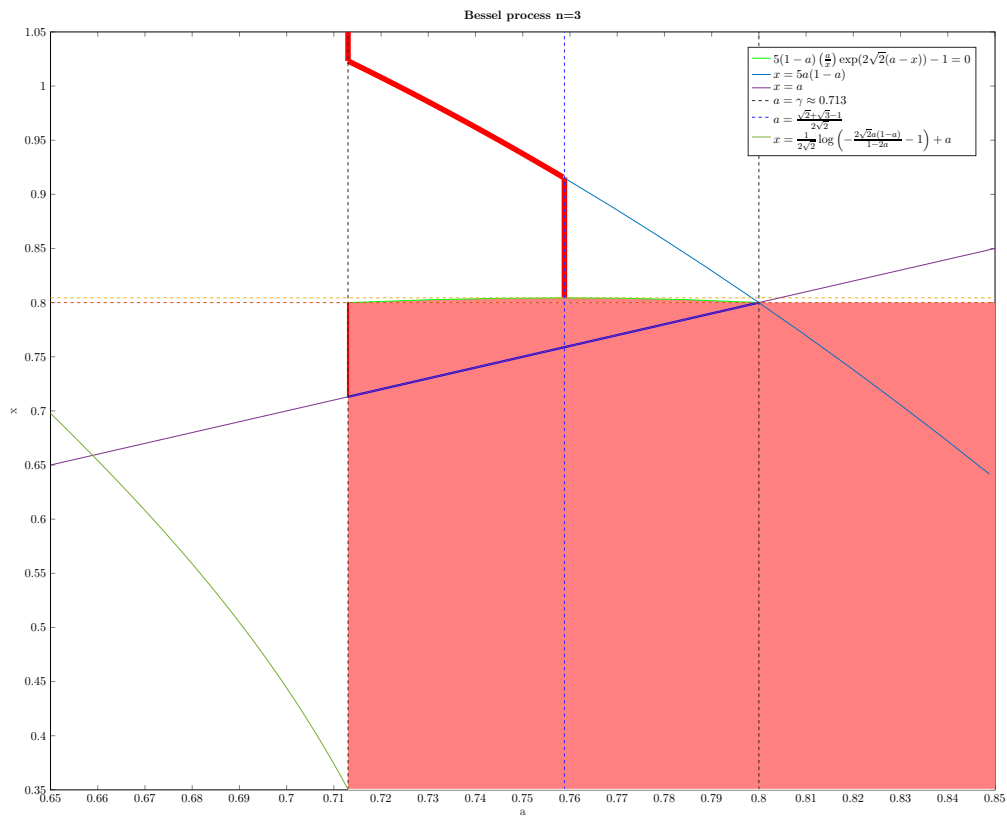
\begin{figure}[ht] 
\centering
\resizebox{\columnwidth}{!}{

\begin{tikzpicture}

\begin{axis}[%
width=11.797in,
height=9.474in,
at={(1.979in,1.279in)},
scale only axis,
xmin=0.65,
xmax=0.85,
xlabel style={font=\color{white!15!black}},
xlabel={a},
ymin=0.35,
ymax=1.05,
ylabel style={font=\color{white!15!black}},
ylabel={x},
axis background/.style={fill=white},
title style={font=\bfseries},
title={Bessel process n=3},
legend style={legend cell align=left, align=left, draw=white!15!black}
]

\addplot [color=green, name path=f]
  table[row sep=crcr]{%
0.713	0.8\\
0.73	 0.8026\\
0.745	0.80396\\
0.7588	0.8044\\
0.77	0.8041\\
0.785	0.8027\\
0.8	0.8\\
};
  \addlegendentry{$5(1-a)\left(\frac{a}{x}\right) \exp(2\sqrt{2}(a-x)) - 1=0$}

\addplot [color=red, line width=5.0pt, forget plot]
  table[row sep=crcr]{%
0.713	1.02316\\
0.713	1.05\\
};
\addplot [color=red, line width=2.0pt, forget plot]
  table[row sep=crcr]{%
0.713	0.713\\
0.713	0.8\\
};
\addplot [color=blue, line width=2.0pt, forget plot]
  table[row sep=crcr]{%
0.713	0.713\\
0.8	0.8\\
};
\addplot [color=red, line width=5.0pt, forget plot]
  table[row sep=crcr]{%
0.758819	0.804394\\
0.758819	0.9156\\
};
\addplot [color=red, line width=5.0pt, forget plot]
  table[row sep=crcr]{%
0.713	1.023155\\
0.7135	1.02208875\\
0.714	1.02102\\
0.7145	1.01994875\\
0.715	1.018875\\
0.7155	1.01779875\\
0.716	1.01672\\
0.7165	1.01563875\\
0.717	1.014555\\
0.7175	1.01346875\\
0.718	1.01238\\
0.7185	1.01128875\\
0.719	1.010195\\
0.7195	1.00909875\\
0.72	1.008\\
0.7205	1.00689875\\
0.721	1.005795\\
0.7215	1.00468875\\
0.722	1.00358\\
0.7225	1.00246875\\
0.723	1.001355\\
0.7235	1.00023875\\
0.724	0.99912\\
0.7245	0.99799875\\
0.725	0.996875\\
0.7255	0.99574875\\
0.726	0.99462\\
0.7265	0.99348875\\
0.727	0.992355\\
0.7275	0.99121875\\
0.728	0.99008\\
0.7285	0.98893875\\
0.729	0.987795\\
0.7295	0.98664875\\
0.73	0.9855\\
0.7305	0.98434875\\
0.731	0.983195\\
0.7315	0.98203875\\
0.732	0.98088\\
0.7325	0.97971875\\
0.733	0.978555\\
0.7335	0.97738875\\
0.734	0.97622\\
0.7345	0.97504875\\
0.735	0.973875\\
0.7355	0.97269875\\
0.736	0.97152\\
0.7365	0.97033875\\
0.737	0.969155\\
0.7375	0.96796875\\
0.738	0.96678\\
0.7385	0.96558875\\
0.739	0.964395\\
0.7395	0.96319875\\
0.74	0.962\\
0.7405	0.96079875\\
0.741	0.959595\\
0.7415	0.95838875\\
0.742	0.95718\\
0.7425	0.95596875\\
0.743	0.954755\\
0.7435	0.95353875\\
0.744	0.95232\\
0.7445	0.95109875\\
0.745	0.949875\\
0.7455	0.94864875\\
0.746	0.94742\\
0.7465	0.94618875\\
0.747	0.944955\\
0.7475	0.94371875\\
0.748	0.94248\\
0.7485	0.94123875\\
0.749	0.939995\\
0.7495	0.93874875\\
0.75	0.9375\\
0.7505	0.93624875\\
0.751	0.934995\\
0.7515	0.93373875\\
0.752	0.93248\\
0.7525	0.93121875\\
0.753	0.929955\\
0.7535	0.92868875\\
0.754	0.92742\\
0.7545	0.92614875\\
0.755	0.924875\\
0.7555	0.92359875\\
0.756	0.92232\\
0.7565	0.92103875\\
0.757	0.919755\\
0.7575	0.91846875\\
0.758	0.91718\\
0.7585	0.91588875\\
};

\addplot [color=mycolor1]
  table[row sep=crcr]{%
0.7588	0.9151128\\
0.7688	0.8887328\\
0.7788	0.8613528\\
0.7888	0.8329728\\
0.7988	0.8035928\\
0.8088	0.7732128\\
0.8188	0.7418328\\
0.8288	0.7094528\\
0.8388	0.6760728\\
0.8488	0.6416928\\
};
\addlegendentry{$x=5a(1-a)$}

\addplot [color=mycolor2, dashed,name path=x08, forget plot]
  table[row sep=crcr]{%
0.65	0.8\\
0.675	0.8\\
0.7	0.8\\
0.725	0.8\\
0.75	0.8\\
0.775	0.8\\
0.8	0.8\\
0.825	0.8\\
0.85	0.8\\
};
\addplot [color=mycolor3, dashed,name path=x08044, forget plot]
  table[row sep=crcr]{%
0.65	0.8044\\
0.675	0.8044\\
0.7	0.8044\\
0.725	0.8044\\
0.75	0.8044\\
0.775	0.8044\\
0.8	0.8044\\
0.825	0.8044\\
0.85	0.8044\\
};

\addplot [color=mycolor4]
  table[row sep=crcr]{%
0.65	0.65\\
0.675	0.675\\
0.7	0.7\\
0.725	0.725\\
0.75	0.75\\
0.775	0.775\\
0.8	0.8\\
0.825	0.825\\
0.85	0.85\\
};
\addlegendentry{$x=a$}

\addplot [color=white!15!black, dashed]
  table[row sep=crcr]{%
0.713	0.35\\
0.713	1.05\\
};
\addlegendentry{$a=\gamma\approx 0.713$}

\addplot [color=white!15!blue, dashed]
  table[row sep=crcr]{%
0.758819	0.35\\
0.758819	1.05\\
};
\addlegendentry{$a=\frac{\sqrt{2}+\sqrt{3}-1}{2\sqrt2}$}

\addplot [color=mycolor5]
  table[row sep=crcr]{%
0.65	0.697838976953865\\
0.6505	0.695692557563705\\
0.651	0.693542436210658\\
0.6515	0.691388517124209\\
0.652	0.689230703807635\\
0.6525	0.687068899016251\\
0.653	0.684903004735258\\
0.6535	0.682732922157214\\
0.654	0.680558551659088\\
0.6545	0.678379792778903\\
0.655	0.67619654419193\\
0.6555	0.674008703686449\\
0.656	0.671816168139032\\
0.6565	0.669618833489347\\
0.657	0.667416594714463\\
0.6575	0.665209345802647\\
0.658	0.662996979726617\\
0.6585	0.660779388416253\\
0.659	0.658556462730734\\
0.6595	0.656328092430088\\
0.66	0.654094166146131\\
0.6605	0.651854571352783\\
0.661	0.649609194335721\\
0.6615	0.647357920161378\\
0.662	0.645100632645226\\
0.6625	0.642837214319356\\
0.663	0.640567546399307\\
0.6635	0.638291508750129\\
0.664	0.636008979851651\\
0.6645	0.633719836762929\\
0.665	0.631423955085841\\
0.6655	0.629121208927809\\
0.666	0.626811470863608\\
0.6665	0.624494611896237\\
0.667	0.622170501416817\\
0.6675	0.619839007163489\\
0.668	0.617499995179264\\
0.6685	0.615153329768804\\
0.669	0.612798873454082\\
0.6695	0.610436486928899\\
0.67	0.608066029012203\\
0.6705	0.60568735660017\\
0.671	0.603300324617015\\
0.6715	0.600904785964476\\
0.672	0.598500591469924\\
0.6725	0.596087589833063\\
0.673	0.593665627571146\\
0.6735	0.591234548962687\\
0.674	0.588794195989574\\
0.6745	0.586344408277566\\
0.675	0.58388502303508\\
0.6755	0.581415874990232\\
0.676	0.578936796326047\\
0.6765	0.576447616613788\\
0.677	0.573948162744317\\
0.6775	0.571438258857434\\
0.678	0.568917726269095\\
0.6785	0.566386383396451\\
0.679	0.563844045680605\\
0.6795	0.56129052550702\\
0.68	0.558725632123472\\
0.6805	0.556149171555457\\
0.681	0.553560946518958\\
0.6815	0.550960756330469\\
0.682	0.548348396814154\\
0.6825	0.545723660206045\\
0.683	0.543086335055163\\
0.6835	0.540436206121409\\
0.684	0.537773054270138\\
0.6845	0.535096656363246\\
0.685	0.532406785146647\\
0.6855	0.529703209133987\\
0.686	0.526985692486433\\
0.6865	0.524253994888387\\
0.687	0.521507871418945\\
0.6875	0.518747072418914\\
0.688	0.51597134335323\\
0.6885	0.513180424668531\\
0.689	0.510374051645733\\
0.6895	0.507551954247335\\
0.69	0.504713856959281\\
0.6905	0.501859478627087\\
0.691	0.498988532286021\\
0.6915	0.496100724985043\\
0.692	0.493195757604239\\
0.6925	0.490273324665452\\
0.693	0.487333114135804\\
0.6935	0.484374807223766\\
0.694	0.481398078167463\\
0.6945	0.478402594014815\\
0.695	0.475388014395165\\
0.6955	0.472353991281963\\
0.696	0.469300168746107\\
0.6965	0.466226182699464\\
0.697	0.463131660628121\\
0.6975	0.460016221314846\\
0.698	0.45687947455024\\
0.6985	0.453721020832013\\
0.699	0.450540451051795\\
0.6995	0.44733734616884\\
0.7	0.444111276869985\\
0.7005	0.440861803215125\\
0.701	0.437588474267477\\
0.7015	0.434290827707836\\
0.702	0.430968389431977\\
0.7025	0.427620673130306\\
0.703	0.424247179848821\\
0.7035	0.420847397530359\\
0.704	0.417420800535064\\
0.7045	0.413966849138932\\
0.705	0.410484989009209\\
0.7055	0.406974650655353\\
0.706	0.403435248854175\\
0.7065	0.399866182047683\\
0.707	0.39626683171206\\
0.7075	0.392636561696085\\
0.708	0.388974717527222\\
0.7085	0.385280625683433\\
0.709	0.381553592828684\\
0.7095	0.37779290500992\\
0.71	0.373997826813189\\
0.7105	0.370167600476344\\
0.711	0.366301444955665\\
0.7115	0.362398554943436\\
0.712	0.358458099833411\\
0.7125	0.354479222630759\\
0.713	0.350461038802907\\
};
\addlegendentry{$x=\frac{1}{2\sqrt{2}} \log\left(- \frac{2\sqrt{2}a(1-a)}{1-2a}-1\right)+a$}

\addlegendentry{$ $}

\addplot [color=white!15!black, dashed, forget plot]
  table[row sep=crcr]{%
0.8	0.35\\
0.8	1.05\\
};

  \addplot [domain=(0.713:0.85, samples=100, name path=x035, thick, color=white!50] {0.3505};

  \addplot[red!50, opacity=0.2] fill between[of=x035 and x08, soft clip={domain=0.713:0.85}];
  \addplot[red!50, opacity=0.2] fill between[of=f and x08, soft clip={domain=0.713:0.8}];

\end{axis}

\begin{axis}[%
width=15.222in,
height=11.625in,
at={(0in,0in)},
scale only axis,
xmin=0,
xmax=1,
ymin=0,
ymax=1,
axis line style={draw=none},
ticks=none,
axis x line*=bottom,
axis y line*=left
]
\end{axis}
\end{tikzpicture}%
}
\caption{Optimal barrier $a^{**}(x)$ that attains the maximum of $\Lambda(x,a)$ for Bessel case. (Red line and red region)}\label{phaseplot3}
\end{figure}

It follows that there are four separate cases:
\begin{itemize}
\item[(i)] $0 < a < \frac15$, in this case $ x^{**} < x^{*} < \frac45$;

\item[(ii)] $\frac15 < a < \gamma$, in this case $x^{**} < \frac45 < x^{*}$;

\item[(iii)]  $\gamma < a < \frac45$, in this case $\frac45 < x^{**} < x^{*}$;

\item[(iv)] $\frac45 < a < 1$, in this case $x^{*} < x^{**} < \frac45$.
\end{itemize}

In cases (i) and (ii), $[0,a]$ is not an equilibrium. We now focus on case (iii) and (iv). For case (iii), $\phi(1-x)=\min(1-x,\frac15)$.  When $x < x^{**}$, we have $\Lambda(x,a)=1$. When $x^{**} < x < x^*$, we have 
$$
\Lambda(x,a) = \frac12 (1-a)(\frac{a}{x}) e^{-2\sqrt{2}(x-a)} + \frac12.
$$
When $x > x^*$, it holds that
$$
\Lambda(x,a) = \frac12 (1-a)\Big(\frac{a}{x}\Big) + \frac12 (1-a)\Big(\frac{a}{x}\Big) e^{-2\sqrt{2}(x-a)}.
$$
The two vertical lines $a=\gamma$, $a=\frac45$ and the three curves $x=\gamma$, $x=x^{**}(a)$ and $x=5a(1-a)$ divide the $x$-$a$ plane into three regions. In the following, given a fixed $x \in [ \gamma ,\infty)$, we look for $a^{**}(x)$ which maximises $\Lambda(x,a)$. When $ \gamma< x < \frac45$, $a \in [\gamma, \frac45] \mapsto \Lambda(x,a)=1$ is a constant. Notice that by the definition, maximising $a \mapsto x^{**}(a)$ is equivalent to maximising $a \mapsto a(1-a) e^{2\sqrt{2}a}$, we get $a = \hat{a}=\frac{\sqrt{2}+\sqrt{3}-1}{2\sqrt{2}}\approx 0.75882$. As $x^{**}(\hat a) \approx 0.80439$, it follows that when $x \in [\frac45, x^{**}(\hat a)]$, any $a \in [a_1, a_2]$ will maximise $\Lambda(x,a)$, with $a_1, a_2$ being the two solutions to the equation $x^{**}(a) = x$ in the region $[\gamma, \frac45]$. As $5 \hat a(1-\hat a)=\frac54(\sqrt{3}-1) \approx 0.91506$, we have that for $x \in [x^{**}(\hat a), \frac54(\sqrt{3}-1)]$, $a^{**}=\hat a$. 
Finally, for $x > 5a(1-a)$, we can solve $\frac{\partial \Lambda(x,a)}{\partial a} = 0$, and get 
$$
x=\hat x(a) = \frac{1}{2\sqrt{2}} \ln \left( - \frac{2\sqrt{2}a(1-a)}{1-2a} - 1 \right)+a.
$$
 As $\hat x(a) < x^{**}(a)$ for all $a \geq \gamma$,
we have that $a^{**}(x)$ satisfies $5 a^{**}(1-a^{**})=x$ for $x \in [\frac54(\sqrt{3}-1), 5\gamma(1-\gamma)]$, i.e. $a^{**}=\frac{1}{10} \left(5+\sqrt{25-20x}\right)$, and $a^{**}(x)=\gamma$ for $x \geq 5\gamma(1-\gamma)$.

For case (iv), $\phi(1-x)=\min(1-x,\frac15)=1-x$.  When $x < x^{*}$, $\Lambda(x,a)=1$. When $x^{*} < x < x^{**}$, 
$$
\Lambda(x,a) = \frac12 (1-a)(\frac{a}{x}) + \frac12.
$$
When $x > x^{**}$, 
$$
\Lambda(x,a) = \frac12 (1-a)(\frac{a}{x}) + \frac12 (1-a)(\frac{a}{x}) e^{-2\sqrt{2}(x-a)}.
$$
In this case (iv) $a^{**}=\frac45$ maximises $\Lambda(x,a)$, and in addition, case (iv) is dominated by case (iii).

In conclusion, we have that
$$
a^{**}(x) = \left\{\begin{aligned}
&\gamma,   &~~~\mbox{if } x \geq 5\gamma(1-\gamma), \\
&\frac{5+\sqrt{25-20x}}{10},   &~~~\mbox{if } \frac54(\sqrt{3}-1) \leq x < 5\gamma(1-\gamma), \\
&\hat{a},    &~~~\mbox{if } x^{**}(\hat a)  \leq x < \frac54(\sqrt{3}-1), \\
&\left\{ \gamma \leq a \leq \frac45 | x \leq x^{**}(a) \right\},   &~~~\mbox{if } \frac45  \leq x < x^{**}(\hat a), \\
&\left\{ \gamma \leq a \leq 1 \right\},   &~~~\mbox{if } \gamma  \leq x < \frac45,
\end{aligned}\right.
$$
where we recall that $\gamma \approx 0.71305$, $\hat a= \frac{\sqrt{2}+\sqrt{3}-1}{2\sqrt{2}}\approx 0.75882$, $x^{**}(\hat a) \approx 0.80439$, and $5\gamma(1-\gamma) \approx 1.02316$.
As a result, in this example, there does not exist an optimal equilibrium such that its value function dominates the ones under other equilibria for any $x>0$. 

\begin{remark}\label{bess-remark}

Unlike the previous geometric Brownian motion case, in Bessel process with $n=3$, there is no scenario under which the optimal equilibrium not only exists, but also differs from the minimal equilibrium. The reason is the following: for each $r \in \mathrm{supp}(\rho)$, one can find $a^{**}_{r}(x): = \argmax_a (1-a)(\frac{a}{x}) e^{-\sqrt{2r}(x-a)}$. When $x$ is large enough, the optimal $a^{**}$ is some combination of all these $a^{**}_{r}(x)$'s. In particular when $x$ changes, the weight of the above combination also changes, hence there is no universal maximizer $a^{**}$. 

In the previous geometric Brownian motion case, with $\rho$ being distributed on $0$ and another point, the optimization is carried out with respect to the sum of a constant($1-a$) and $(1-a)(\frac{a}{x})^{f(r)}$. In particular, when $x$(the weight) changes, the optimal equilibrium will always stay at the optimum of $(1-a)(\frac{a}{x})^{f(r)}$.

\end{remark}

\ \\
\noindent
\textbf{Acknowledgements}: The authors are grateful to two anonymous referees for their helpful comments and suggestions. S. Deng is supported by the Hong Kong University of Science and Technology Start-up grant no. R9826. X. Yu is supported by the Hong Kong
RGC General Research Fund (GRF) under grant no. 15304122 and by the Research Centre for Quantitative Finance at the Hong Kong Polytechnic University under grant no. P0042708.
J. Zhang is supported by the Chinese University of Hong Kong startup grant (4937261).
\\

\end{document}